%% file: document.tex
\documentclass[submission,copyright,creativecommons]{eptcs}
\providecommand{\event}{SOS 2007} % Name of the event you are submitting to

\usepackage{iftex}

\ifpdf
  \usepackage{underscore}         % Only needed if you use pdflatex.
  \usepackage[T1]{fontenc}        % Recommended with pdflatex
\else
  \usepackage{breakurl}           % Not needed if you use pdflatex only.
\fi

\usepackage[english]{babel}
\usepackage{graphicx}
\usepackage{amsmath,amsthm}
\usepackage{amssymb} 
\usepackage{color}
\usepackage{textgreek}

\usepackage{isabelle, isabellesym}
\isabellestyle{it} % optional
\newcommand{\SNIP}[2]{\expandafter\newcommand\csname snippet--#1\endcsname{#2}}
\input{snips} % assuming the snippets are saved as snips.tex

\newcommand{\GetSnip}[1]{%
    \ifcsname snippet--#1\endcsname%
        \csname snippet--#1\endcsname%
    \else%
        \PackageWarning{snips}{Snippet ``#1'' is undefined.}%
        \emph{Warning: Snippet ``#1'' is undefined.}%
    \fi%
}

\newcommand{\RawCartouche}[3]{\GetSnip{#1-#2-#3}}

\newcommand{\Cartouche}[3]{%
    {\isacartoucheopen}%
    \RawCartouche{#1}{#2}{#3}%
    {\isacartoucheclose}%
}

\newcommand{\Snippet}[1]{{%
  \newcount\i
  \i=0
  \loop
    \GetSnip{#1-\the\i}
    \advance \i 1
  \ifcsname snippet--#1-\the\i\endcsname
  \repeat
}}

\newcommand{\SnippetPart}[3]{{%
  \newcount\i
  \i=#1
  \loop
    \ifnum \i=#2
      \renewcommand{\isanewline}{}%
    \fi
    \GetSnip{#3-\the\i}
    \advance \i 1
    \ifnum \i>#2 {}
    \else \repeat
}}

\def\isadelimdocument{}\def\endisadelimdocument{}

\def\isatagdocument{}

\def\isadelimproof{}\def\endisadelimproof{}
\def\isatagproof{}\def\endisatagproof{}
\def\isafoldproof{}

\usepackage{hyperref}
\usepackage{multirow}
\usepackage{tikz}
\usetikzlibrary{shadings}
\usetikzlibrary{shapes}
\usetikzlibrary{arrows, positioning}
\tikzset{>=latex}

\newtheorem{theorem}{Theorem}%[subsection]
\newtheorem{proposition}[theorem]{Proposition}

\newtheorem{lemma}[theorem]{Lemma}

\theoremstyle{definition}

\theoremstyle{remark}

\title{Stalnaker's Epistemic Logic in Isabelle/HOL.\thanks{Supported by NSF CAREER Award CNS-1552934 and NSF:CCRI Award CNS-2016592.}}
\def\titlerunning{Stalnaker's Epistemic Logic in Isabelle/HOL}

\author{Laura P. Gamboa Guzman%\orcidID{0000-0003-3133-5148}
\institute{Iowa State University \\ Ames, IA 50011, USA 
\email{lpgamboa@iastate.edu}}
 \and
Kristin Y. Rozier%\orcidID{0000-0002-6718-2828}
\institute{Iowa State University \\ Ames, IA 50011, USA 
\email{kyrozier@iastate.edu}}
}
\def\authorrunning{L.P. Gamboa Guzman and K. Y. Rozier}
\begin{document}
\maketitle

\begin{abstract} % 15-250 words
% L:: First line: motivating question or problem... like: there is a need for this kind of thing...
The foundations of formal models for epistemic and doxastic logics often rely on certain logical aspects of modal logics such as S4 and S4.2 and their semantics; however, the corresponding mathematical results are often stated in papers or books without including a detailed proof, or a reference to it, that allows the reader to convince themselves about them. We reinforce the foundations of the epistemic logic S4.2 for countably many agents by formalizing its soundness and completeness results for the class of all weakly-directed pre-orders in the proof assistant Isabelle/HOL. This logic corresponds to the knowledge fragment, i.e., the logic for formulas that may only include knowledge modalities in Stalnaker's system for knowledge and belief. Additionally, we formalize the equivalence between two axiomatizations for S4, which are used depending on the type of semantics given to the modal operators, as one is commonly used for the relational semantics, and the other one arises naturally from the topological semantics.

% L:: Name the other results mentioning why they contribute
\end{abstract}

%%%%%%%%%%%%%%%%%%%%%%%%%%%%%%%%%%%%%%%%%%%%%%%%%%%%%%%%%%%%%%%%%%
\section{Introduction}
% Introduce the setup of the problem and explains why is it relevant. Mention that all systems already formalized (as far as I know, or at least by From in \texttt{Epistemic\_logic.ty}) were complete with respect to a class of frames that is characterized  by a universal formula, whereas S4.2 is characterized by a $\forall \exists$ formula, which makes it harder to formalize it's completeness result.

%KYR: test
% I'm citing something here: \cite{From2021, Hintikka1962}.
%
% Need to mention closely related work here including that LTL is already formulated in Isabelle here: https://www.isa-afp.org/entries/LTL.html
% The normal form is here: https://www.isa-afp.org/entries/LTL_Normal_Form.html
% TLA (which is not full LTL, but an extension for practical use in verification) is also in Isabelle/HOL: https://isabelle.in.tum.de/library/HOL/HOL-TLA/README.html
% Point: people use temporal logics so formulating all of the ones we use is useful; this logic was missing before but is useful so we add it here.

% Motivation: epistemic logics allow representing important safety properties for system verification using model checking and related formal fault-detection techniques \cite{cimatti2016lazy,bozzano2014formal,tonetta2015formal}.

Epistemic logics are a family of logics that allow us to reason about knowledge among a group of agents, as well as their knowledge about other's knowledge\cite{Fagin1995-FAGRAK-2}. Reasoning about knowledge is useful for detecting and identifying faults during the operation of complex critical systems \cite{bozzano2014formal,tonetta2015formal}, where important safety properties are formalized using a modal language that combines temporal, in particular, LTL (Linear Temporal Logic), and epistemic modal operators, so to verify the correctness of the system using model checking and related formal fault-detection techniques \cite{cimatti2016lazy,HOL_TLA,LTL_AFP}.

When it comes to modal logics for knowledge, most of these logics correspond to normal logics between S4 and S5 \cite{VanDitmarsch2008,Stalnaker2006}. In particular, we consider Stalnaker's epistemic logic, which coincides with the logic S4.2. It is known that this logic strictly stronger than S4, but weaker than S5 \cite{chagrov1997modal}. This logic is known to be sound and complete with respect to all weakly directed S4-frames, that is, all frames consisting of reflexive and transitive binary relations that are confluent \cite{sep-logic-epistemic}, but this proof is often omitted in textbooks where most extensions to system K (the weakest normal modal logic) are usually treated informally.

Additionally, we encode in Isabelle/HOL the axiomatization of S4 obtained from the study of the topological interpretation for modal languages, which was introduced prior to the relational one that is more commonly found in the literature. This topological interpretation is done by reading the modal necessity operator as an interior operator on a topological space, for which is known that the modal logic S4 is complete with respect to all topological spaces \cite{Aiello2003}. The preferred axiomatization for the logic of topological spaces differs from the one presented in \cite{From2021}, not only from the set of axioms, but also the deductive rules, since it captures the axioms for an interior operator instead of a reflexive and transitive binary relation. As a consequence, this makes the topological axiomatization not directly recognizable as a normal modal logic, since the deduction rules seem to be weaker at first glance. Since several authors have been recently developing topological semantics for notions of knowledge and belief \cite{2019TopoEvidMulti,BaltagEtAl2016JustTopoBel,BaltagOzgun2013topobelief,GamboaG_2021}, we provide a formalization for this result, which often gets briefly mentioned and applied without being proved in detail.

\subsection*{Contributions}

We formalize Stalnaker's epistemic logic, which is expressively equivalent to S4.2 \cite{Stalnaker2006}, as well as some intermediate results for the underlying propositional logic and the modal logics K, .2, and S4 mainly regarding rewriting rules, properties for maximal consistent sets of formulas, and frame properties that are induced by the chosen set of axioms in the proof assistant Isabelle/HOL \cite{Stalnaker_Logic-AFP}. Our main result is a formalization of the soundness and completeness of Stalnaker's epistemic logic (restricted to countably many agents) with respect to all weakly directed(also refered to as \emph{confluent} or \emph{convergent} in the literature \cite{DegreesIgnoranceMuCalculus,sep-logic-epistemic}) S4 frames, this is, all frames consisting of a non-empty set $W$ and a binary relation $R_i$ on $W$, one for each agent label $i$, that is reflexive, transitive, and that satisfies the property described by the following condition
\[ \forall x\, \forall y\, \forall z\, ( x R_i y \wedge x R_i z \implies \exists w\, y R_i w \wedge z R_i w ). \]
The proof uses a Henkin-style completeness method, which is commonly used for these kinds of logics \cite{Blackburn2001} and was already available on Isabelle's Archive of Formal Proofs \cite{EpistemicLogicAFP}.

As far as we know, all systems corresponding to some multi-agent epistemic logic already formalized in Isabelle/HOL, which are all contained in \cite{EpistemicLogicAFP} (the ground base for our formalization), were complete with respect to a class of frames characterized by a universal formula, i.e., a property of frames given by a first order formula of the form $\forall \overline{x}\,\phi(\overline{x})$, where $\phi(\overline{x})$ is a quantifier-free formula with variable symbols in $\overline{x}=(x_1, \ldots,x_n)$. However, the logic S4.2 is complete with respect to a class of frames that cannot be characterized using a universal formula; instead it is characterized by a universal-existential formula. This universal-existential characterization makes it harder to formalize its completeness result, since one has to show the existence of an object in the universe of the canonical model satisfying a condition on a union of consistent sets of formulas. For this, we followed an argument given by Stalnaker in \cite{lectureNotesStalnaker} that includes a set of theorems that are consequences of the axiom (.2) in K, which imply the consistency of a set obtained by taking the union of all known facts for an agent in two different worlds that were accessible from a third one.

Nonetheless, our formalization also includes some intermediate results that are well-known for all normal modal logics, and that are commonly used when dealing with formal proofs in Hilbert-style systems. Finally, we formalized the equivalence between two of the most used axiomatizations of S4, the one presented in \cite{From2021} which is commonly used when dealing with the relational semantics \cite{Blackburn2001}, and the one introduced by McKinsey and Tarski for the topological semantics \cite{Aiello2003}.

\subsection*{Related work}
% Show what has already been done, specially the formalization of systems S4 and S5 by A.H. From.
Our ground base is the Isabelle/HOL theory \texttt{EpistemicLogic.thy} \cite{EpistemicLogicAFP}, which contains not only the formalization of other epistemic logics such as S4 and S5, but also formalizes the definition of an abstract canonical model, as well as a simple and convenient way to work with any desired normal modal logic by adding necessary the axioms to the basic system K \cite{From2021}. A related paper to this formalization is \cite{Villadsen2022}, which contains a broad and updated summary on formalizations of logical systems and correspondent important results using different theorem provers. Other ways to formalize logical systems and their completeness results on Isabelle have also been studied in \cite{dionisio2005defining} and \cite{blanchette2017coinductive}, which include (but are not limited to) the use of natural deduction rules, sequent-style rules, and tableau rules for the formal systems, and coinductive methods for soundness and completeness results. 

However, given the already existent formalization of LTL in Isabelle/HOL \cite{LTL_AFP}, as well as the prevalence of Isabelle as a tool for formal verification of safety requirements for critical systems, it becomes important to provide this formalization for this particular proof assistant. In addition to this, in \cite{KZ2011} the authors defined and investigated the notion of $KB_R$-structures, which are used to represent a description of the epistemic status of a rational agent that is not necessarily aware of their ignorance, and provided a result that matches them with models of the epistemic logic S4.2. Modal logics between S4.2 and S5 are of special interest for applications in epistemic logic, since they allow formalizations of several degrees of ignorance for each one of the agents \cite{DegreesIgnoranceMuCalculus}.

The paper is organized as follows: Section \ref{sec_background} introduces the necessary background on epistemic logic, including its relational and topological semantics, and how the syntax and the relational semantics were formalized in Isabelle/HOL in \cite{EpistemicLogicAFP}. Section \ref{sec_formalization} explains our formalization of Stalnaker's epistemic logic, including the intermediate results necessary to prove the main results, and the limitations of these to only countably many agents. Section \ref{sec_topoS4} explains our formalization of the equivalence between the two most common axiomatizations for S4, the one that arises from the topological semantics, and the one commonly used when working with the relational semantics. Finally, in Section \ref{sec_results_future_work}, we conclude with a discussion about the results, limitations, and future work.

%%%%%%%%%%%%%%%%%%%%%%%%%%%%%%%%%%%%%%%%%%%%%%%%%%%%%%%%%%%%%%%%%%
\section{Background} \label{sec_background} %KYR - TODO: I changed some passive voice to active voice...
% Here we present the axiomatic system developed by R. Stalnaker (for both knowledge and belief!), present (without a proof) the main result for the ``knowledge formulas'' (i.e., that these correspond to the multimodal system S4.2), and explain why we can't formalize yet the whole system but only the knowledge fragment.
%%% --------------------------------------------------------- %%%
\subsection{Stalnaker's Epistemic logic}
We briefly present the axiomatic system developed by R. Stalnaker for both notions of knowledge and belief, as well as the main result for the ``knowledge formulas'' (i.e., for those formulas that do not contain any belief modal operators), which correspond to the multimodal system S4.2 \cite{Stalnaker2006}. We omit the proof for this result, as the Isabelle theory ``Epistemic logic: Completeness of Modal Logics'' \cite{EpistemicLogicAFP} does not support belief formulas.

Consider the well-formed formulas obtained from the following grammar, where $x$ ranges over the set of propositional symbols and $i$ ranges over the set of agent labels:
\[ \phi, \psi ::= \bot \,|\, x \,|\, \phi \vee \psi \,|\, \phi \wedge \psi \,|\, \phi \to \psi \,|\, K_i \phi \,|\, B_i \phi. \]
% Stalnaker does not clarify which logical symbols he uses for his formalization, but all formulas presented there use negations and implications... but the formalization is made on top of Form's which uses the ones presented.
The operators $K_i$ and $B_i$ mean ``agent $i$ knows'' and ``agent $i$ believes,'' respectively. Although Stalnaker does not present his logic of knowledge and belief using this exact set of propositional connectives, but a proper subset of these, we added the remaining ones given that From's formalization includes all of them \cite{From2021}.

% \subsection{Proof system}
% The formal system for this  is the one obtained by adding to the axioms and rules of the logic S4 for the knowledge operators the following axioms: %KYR: this sentence is confusing
Stalnaker's principles (axioms) for knowledge and belief appear in Table \ref{tab_axioms_Stal}, along with their interpretations in natural language. \textit{Stalnaker's logic for knowledge and belief} corresponds to the formal system obtained by adding these axioms to the axioms and rules of the multi-modal logic S4, that is, the smallest logic containing the following axioms:
\begin{itemize}
    \item all propositional tautologies,
    \item axiom K: $(K_i(\phi \to \psi) \wedge K_i \phi) \to K_i \psi$,
    \item axiom T: $K_i \phi \to \phi$, and
    \item axiom 4: $K_i \phi \to K_i K_i \phi$;
\end{itemize}
and that is closed under Modus Ponens and the Necessitation rule, ``from $\phi$ infer $K_i \phi$'', where $i$ ranges over the set of agents.
\begin{table}[h]
\begin{center}
\caption{Axioms for knowledge and belief.}
\label{tab_axioms_Stal}
\begin{tabular}{ll} \hline
	$B_i \phi \to K_i B_i \phi$ & Positive introspection \\
	$\neg B_i \phi \to K_i \neg B_i \phi$ & Negative introspection \\
	$K_i \phi \to B_i \phi$ & Knowledge implies belief \\
	$B_i \phi \to \neg B_i \neg \phi$ & Consistency of belief \\
	$B_i \phi \to B_i K_i \phi$ & Strong belief \\ \hline
\end{tabular}
\end{center}
\end{table}

The following proposition summarizes some relevant properties of this logic.

\begin{proposition} \label{prop_knowledge_belief_logics}
The following are some key properties of Stalnaker's logic for knowledge and belief \cite{Stalnaker2006}.
\begin{enumerate}
%   \item It yields a pure belief logic, KD45; %KYR: need citation for KD45 definition? What does this acronym stand for? How does this relate to S4 and S4.2?  Why is the "pure belief logic" part important? What would an impure belief logic be?
  \item The following equivalences, one for each agent label $i$ and formula $\phi$, are theorems in this logic: \[ B_i \phi \longleftrightarrow \neg K_i \neg K_i \phi. \] %KYR: not a sentence...
  \item As a consequence of the previous property, by replacing `$B_i$' with `$\neg K_i \neg K_i$' in the \emph{Consistency of belief} axiom, we get that $\neg K_i \neg K_i \phi \to K_i \neg K_i \neg \phi$ (also known as axiom .2) is a theorem in this logic. This implies that the knowledge formulas of this logic correspond exactly to the logic given by the system S4.2, i.e., those that can be obtained from the rules and axioms of the multi-modal logic S4 in presence of the axiom .2.
\end{enumerate}
\end{proposition}

%KYR - TODO: language needs to be more precise/concrete...
The above proposition allows us to interpret this logic by giving a semantics only for the propositional variables, Boolean connectives and knowledge operators. Formally, we use structures $\mathfrak{M} = (\mathcal{F},\pi)$ known as Kripke models, where the frame $\mathcal{F} = (W, (R_i)_{i})$ is a pair consisting of a non-empty set of worlds $W$, a set of binary accessibility relations $R_i \subseteq W \times W$, one for each agent $i$, and $\pi: \operatorname{Var} \to 2^W$ is a valuation of propositional symbols. Formula satisfiability at a given world $w \in W$ is defined as follows:
%KYR: need mathematical definitions of world, relation, agent; clarify that pi is a valuation of propositional symbols FROM WHICH SET
%KYR - TODO: basically, give all background needed to check what definitions are already in the theorem prover
\begin{table}[!h]
% \vspace{-10pt}
% \noindent
\begin{tabular}{lcl}
	$\mathfrak{M}, w \not\models \bot$ & & \\
	$\mathfrak{M},w \models x$ & iff & $w \in \pi(x)$ \\
	$\mathfrak{M},w \models \phi \vee \psi$ & iff & $\mathfrak{M},w \models \phi \text{ or } \mathfrak{M},w \models \psi$ \\
% \end{tabular} \end{table}
% \begin{table}[ht] \begin{tabular}{lcl}
	$\mathfrak{M},w \models \phi \wedge \psi$ & iff & $\mathfrak{M},w \models \phi \text{ and } \mathfrak{M},w \models \psi$ \\
	$\mathfrak{M},w \models \phi \to \psi$ & iff & $\mathfrak{M},w \not\models \phi \text{ or } \mathfrak{M},w \models \psi$ \\
	$\mathfrak{M},w \models K_i \phi$ & iff & $\forall v \in W\, (wR_{i}v \to \mathfrak{M},v \models \phi)$ \\
	$\mathfrak{M},w \models B_i \phi$ & iff & $\mathfrak{M},w \models \neg K_i \neg K_i \phi$.
\end{tabular}
\vspace{-10pt}
\end{table}

One can use functions $\mathcal{K} : W \to 2^W$ instead of sets of ordered pairs $R \subseteq W \times W$, as there is a correspondence between these objects by setting \[ w R v \, \Longleftrightarrow v \in \mathcal{K} (w), \] for all $w,v \in W$.

%%% --------------------------------------------------------- %%%
\subsection{Epistemic Logic: Completeness of Modal Logics}

The ``Epistemic Logic: Completeness of Modal Logics'' entry on Isabelle's AFP \cite{EpistemicLogicAFP} contains not only a formalization for the completeness results for some epistemic logics, but also a formalization of the general strategy for Henkin-style proofs for completeness. This is what enabled us to formalize a proof for the completeness result for S4.2. We show here the formalization of the Kripke models from \cite{EpistemicLogicAFP}, which are structures consisting of a set of worlds of type $'w$, a truth assignment for each propositional variable on each world given by the function denoted $\pi$, and a set of accessible worlds from each possible world for each agent $'i$.
%, which uses functions instead of sets of ordered pairs.
\begin{isabelle}
  \SnippetPart{0}{1 }{datatype:kripke}
\end{isabelle}

Consequently, given a Kripke model $\mathfrak{M} = (W, (\mathcal{K}_i)_{i\in I}, \pi)$ with accessibility functions $\mathcal{K}_i$ for each agent $i$, formula satisfiability is defined by setting
\[ \mathfrak{M},w \models K_i \phi \quad \text{iff} \quad \forall v \in \mathcal{K}_i(w)\, (\mathfrak{M},v \models \phi). \]

Additionally, the \emph{dual} operator for each knowledge operator $K_i$ is denoted in this formalization as $L_i$ and is defined as a short hand for ``agent $i$ does not know if something is false.'' In other words, $L_i \phi := \neg K_i (\neg \phi)$, for all formulas $\phi$. The Kripke semantics for this operator corresponds to \cite{Blackburn2001,chagrov1997modal}
\[ \mathfrak{M},w \models L_i \phi \quad \text{iff} \quad \exists v \in \mathcal{K}_i(w)\, (\mathfrak{M},v \models \phi). \]

We show here the corresponding formalization presented in \cite{From2021} for the Kripke semantics, which is defined inductively on formulas for each world.
\begin{isabelle}
  \Snippet{primrec:semantics}
  % \SnippetPart{0}{0}{primrec:semantics}
  % \SnippetPart{7}{7}{primrec:semantics}
\end{isabelle}

From's formalization then focuses on proving the soundness and completeness results for each of the most commonly found \textit{normal modal logics} in the literature concerning certain \textit{classes of frames} \cite{Blackburn2001,chagrov1997modal,From2021}. We now summarize the relevant ones for our formalization.

\begin{enumerate}
    \item The basic logic, K, whose corresponding axiomatic system consists of all propositional tautologies and the axiom K, and is closed under Modus Ponens and the Necessitation Rule, is sound and complete with respect to the class of all frames.

    \item The logic S4 is sound and complete with respect to the class of all transitive and reflexive frames.

\end{enumerate}

Notice that From's formalization does not include modal operators for belief, this restricts us to the knowledge fragment of the language. However, Proposition \ref{prop_knowledge_belief_logics} tells us that belief is equivalent to knowledge, thus we do not lose any information by restricting to the knowledge fragment.

%%% --------------------------------------------------------- %%%
\subsection{Topological semantics and its axioms}

The topological semantics for modal logics was introduced before the relational semantics that presently dominate the field \cite{Aiello2003}, and the first semantics completeness proof for S4 was derives from there. Recall the notion of this topological semantics for a language with a single modal operator $\square$. Let $\mathcal{L}$ be the language composed of all formulas given by the following grammar:
\[ \phi,\psi ::=\, x \,|\, \neg \phi \,|\, \phi \wedge \psi \,|\, \phi \vee \psi \,|\, \phi \to \psi \,|\, \square \phi, \]
where $x$ ranges over the set of propositional symbols $\operatorname{Var}$. Formulas in $\mathcal{L}$ are interpreted in a topological model $M = \langle W, \tau, v \rangle$, consisting of a non-empty set $W$, a topology $\tau$ over $W$, and a valuation $v: \operatorname{Var} \to 2^W$ in the following way:
\begin{itemize}
    \item $M,w \models x$ iff $x \in v(x)$;
    \item $M,w \models \neg \phi$ iff $M,x \not\models \phi$;
    \item $M,w \models \phi \wedge \psi$ iff $M,x \models \phi$ and $M,x \models \psi$;
    \item $M,w \models \phi \vee \psi$ iff $M,x \models \phi$ or $M,x \models \psi$;
    \item $M,w \models \phi \to \psi$ iff $M,x \not\models \phi$ or $M,x \models \psi$;
    \item $M,w \models \square \phi$ iff there exists $U\in \tau$ such that $w\in U$ and $M,y \models \phi$  for all $y\in U$.
\end{itemize}

Although there is nothing inherently wrong with using the deductive system presented in \cite{From2021} for the logic S4, the following axiomatization is often preferable when working with the topological semantics, as the meaning of the axioms and rules under this semantics resembles some well-known properties of topological spaces \cite{Aiello2003}.

\begin{table}[ht]
\begin{center}
\caption{Topological S4 axioms and rules.}
\label{tab_axioms_TopoS4}
\begin{tabular}{c|l|c|l}
\hline
Axiom & Formula & Rule               & Formula                  \\ \hline
N     & $\square \top$  & \multirow{2}{*}{MP} & \multirow{2}{*}{$\displaystyle \frac{\phi \to \psi \quad \phi}{\psi}$} \\ \cline{1-2}
R     & $\square (\phi \wedge \psi) \leftrightarrow (\square \phi \wedge \square \psi)$ &   &                   \\ \hline
T     & $\square \phi \to \psi$ & \multirow{2}{*}{M}  & \multirow{2}{*}{$\displaystyle \frac{\phi \to \psi}{\square \phi \to \square \psi}$}  \\ \cline{1-2}
4     & $\square \phi \to \square \square \phi$ &                     &                   \\ \hline
\end{tabular}
\end{center}
\end{table}

Notice that at first it is not obvious weather or not the logic obtained from this axiomatization is a normal modal logic, often defined as a logic that \emph{extends} system K \cite{Blackburn2001}, as neither axiom K nor the Necessitation Rule are present in the list of axioms and rules. However, we formalized a proof for the equivalence between both axiomatizations in the context of a multi-agent epistemic logic, as in recent years several authors have been developing topological semantics for notions of knowledge and belief \cite{2019TopoEvidMulti,BaltagEtAl2016JustTopoBel,BaltagOzgun2013topobelief}, where this result is often briefly mentioned and applied, but not proved in detail. Nonetheless, it is also worth noticing here that the relational semantics of S4 is no more than a particular case for the topological semantics, as one can assign a binary relation to each topological space by defining what is known as the \emph{specialization order} \cite{Aiello2003}.

%%%%%%%%%%%%%%%%%%%%%%%%%%%%%%%%%%%%%%%%%%%%%%%%%%%%%%%%%%%%%%%%%%
\section{Formalization} \label{sec_formalization}

%KYR - TODO: languge could be more direct and proof-like; let's discuss
We now consider the epistemic logic based on the axioms in Table \ref{tab_axioms_Stal} and the results in Proposition \ref{prop_knowledge_belief_logics} for the knowledge fragment of the language. We prove the soundness and completeness of the pure epistemic logic obtained from this system with respect to all frames consisting of weakly directed pre-orders by combining and applying the results formalized in \cite{From2021} with some auxiliary results provided in the \textit{Utility} section of our Isabelle theory. This allows us to utilize the canonical model strategy to prove completeness of the obtained system. We do not formalize a logic for both knowledge and belief, since we aimed to work on top of the formalization in \cite{EpistemicLogicAFP}, which considers modalities only for knowledge.
Formalizing the whole logic for both knowledge and belief will then require developing a new theory almost from scratch that includes modalities for both notions.

In order to do this, we prove first some intermediate results towards the completeness of the system obtained by adding axiom .2 to the system K, also known as system G in the literature \cite{chagrov1997modal}, including some results about the underlying propositional logic. This system is known to be complete with respect to the class of weakly directed frames, and, although we do not formalize this result completely, we do formalize a version of the Truth lemma for this system, which is needed so that we can combine it with the results for system S4 formalized in \cite{From2021} to achieve our goal of formalizing the completeness result for the logic S4.2 with respect to all weakly directed pre-orders.

%%% --------------------------------------------------------- %%%
\subsection{Rewriting propositional and modal formulas} \label{sec_rewritting}

In the deductive system formalized in \cite{EpistemicLogicAFP}, deduction from a set of premises is defined as follows: given a set of formulas $\Gamma \cup \{\phi\}$, we say that ``$\phi$ is derived from $\Gamma$'' (denoted $\Gamma \vdash \phi$) if there are formulas $\psi_1, \ldots, \psi_k$ in $\Gamma$ such that the formula $\psi_1 \to (\psi_2 \to \ldots (\psi_k \to \phi))$ is a theorem in the system, where $k$ is a non-negative integer. It is well-known that this formula is logically equivalent to $(\psi_1 \wedge \ldots \wedge \psi_k) \to \phi$ in classical propositional logic, thus the notion can be defined by requiring the latter to be a theorem in the system instead. Being able to translate between these two equivalent formulas in our formal deductive system plays an important role for the proof of our main result, thus we provided a formalization of several results of this kind in the \textit{Utility} section of our Isabelle theory, which includes some results that were not used later but that might become handy for the development of the formalizations of other related theories in the future.

Similarly to the \textbf{imply} function in \cite{EpistemicLogicAFP}, which produces, from a list of formulas $[\psi_1, \ldots, \psi_k]$ and a formula $\phi$, the formula $\psi_1 \to (\psi_2 \to \ldots (\psi_k \to \phi))$, we introduce the function \textbf{conjunct}, which takes a list of formulas $[\psi_1, \ldots, \psi_k]$ and produces the formula $\psi_1 \wedge \ldots \wedge \psi_k \wedge \top$. (Notice that the input may be an empty list, in which case the output is $\top$.) We formalized some results about logical equivalences, and derived rules and maximal consistent sets regarding the \textbf{imply} and \textbf{conjunct} functions that are well-known for the logic K, some of which are presented in the following lemmas. These are required to prove in section \ref{sec_Ax2} that the axiom .2 induces the weakly directed property on all frames that satisfy it, following \cite{lectureNotesStalnaker}. We include the proofs for those lemmas that require elaborated arguments.

\begin{lemma}[Derived rules]
For all formulas $\psi_1, \ldots, \psi_k, \phi$, it is the case that $\vdash (\psi_1 \wedge \ldots \wedge \psi_k) \to \phi$ if and only if $\vdash \psi_1 \to (\psi_2 \to \ldots (\psi_k \to \phi))$.
\end{lemma}
  % These are easily proven by induction on the length of the list \texttt{G}.

\begin{lemma}[Logical equivalences]
The following two lemmas capture the fact that in system K, hence in any normal modal logic, the formulas $(K_i \psi_1 \wedge \ldots \wedge K_i \psi_k)$ and $K_i (\psi_1 \wedge \ldots \wedge \psi_k)$ are equivalent, for any finite set of formulas $\psi_1,\ldots, \psi_k$ and any agent $i$.
\end{lemma}

\begin{lemma}[Closure under conjunctions for MCSs]
The following lemma proves that maximal consistent sets of formulas are closed under conjunctions, that is, if $\Gamma$ is a maximal consistent set of formulas and $\psi_1, \ldots, \psi_k$ are some formulas in $\Gamma$, then $\psi_1 \wedge \ldots \wedge \psi_k \in \Gamma$.
\end{lemma}

\begin{figure}[h!]
    \centering
    \begin{tikzpicture}[ scale=.7,
    ne1/.style={rectangle, draw=black, fill=blue!15},
    ne2/.style={rectangle, draw=black, fill=orange!25},
    ne3/.style={rectangle, draw=black, fill=white},
    de1/.style={rounded rectangle, draw=black, fill=blue!15},
    de2/.style={rounded rectangle, draw=black, fill=orange!15},
    de3/.style={rounded rectangle, draw=black, fill=white} ]
      % epistemic_logic
      \node[de1] (e1) {\texttt{Extend}};
      \node[] (y2) [below=of e1] {};
      \node[ne1] (e2) [right=of e1] {\texttt{maximal\_extension}};
      \node[] (x3) [right=of e2] {};
      \node[] (y1) [below=of e2] {};
      \node[de1] (e3) [right=of x3] {\texttt{valid}};
      \node[] (x6) [below=of e3] {};
      \draw[dashed, red!50!gray, thick, opacity=.7] (-1,-.7) -- (13,-.7);
      \draw[dashed, red!50!gray, thick, opacity=.7] (-1,-.7) -- (-2.7,1.3);
      \node[gray] (el) at (10.2,1) {\tt Epistemic\_Logic.thy};
      % ax2
      \node[de1] (a1) [right=of y1] {\texttt{Ax\_2}};
      \node[] (x2) [left=of a1] {};
      \node[ne1] (a2) [below=of a1] {\texttt{Ax\_2\_weakly\_directed}};
      \node[ne1] (a3) [below=of a2] {\texttt{mcs\_2\_weakly\_directed}};
      \node[] (x1) [left=of a2] {};
      \node[ne1] (a4) [below=of a3] {\texttt{imply\_completeness\_K\_2}};
      \node[] (x4) [left=of a4] {};
      \node[ne3] (a5) [left=of x4] {\texttt{soundness\_Ax\_2}};
      \draw[dashed, red!50!gray, thick, opacity=.7] (-6,-9.7) -- (13,-9.7);
      \node[gray] (al) at (4.6,-1.4) {\tt Ax.2 Section};
      % utility
      \node[ne1] (u1) [left=.1cm of y2] {\texttt{K\_conj\_imply\_factor}};
      \node[] (y2) [below=of u1] {};
      \draw[dashed, red!50!gray, thick, opacity=.7] (-6,-2.7) -- (0.5,-2.7);
      \node[gray] (el) at (-4,-.5) {\tt Utility Section};
      % topological
      \node[ne1] (t1) [right=-.1cm of y2] {\texttt{topoS4\_S4}};
      \node[ne3] (t2) [below=of t1] {\texttt{main$_{\texttt{S4'}}$}};
      \draw[dashed, red!50!gray, thick, opacity=.7] (-6,-7.3) -- (0.5,-7.3);
      \draw[dashed, red!50!gray, thick, opacity=.7] (0.5,-.7) -- (0.5,-7.3);
      \node[gray] (el) [text width=2.5cm] at (-4.2,-5.5) {\tt Topological S4 Axioms Section};
      % system S4.2
      \node[] (x5) [below=of a4] {};
      \node[de1] (s1) [right=of x5] {$\vdash_{\texttt{S42}}$};
      \node[] (x7) [below=of s1] {};
      \node[de1] (s2) [below=of a5] {\tt AxS4\_2};
      \node[ne3] (s3) [below=of s2] {\tt soundness\_AxS4\_2};
      \node[ne3] (s4) [below=of s3] {\tt soundness$_{\texttt{S42}}$};
      \node[ne1] (s5) [below=of s1] {\tt imply\_completeness\_S4\_4};
      \node[ne1] (s6) [below=of s5] {\tt completeness$_{\texttt{S42}}$};
      \node[de2] (s7) [right=of s4] {\tt valid$_{\texttt{S42}}$};
      \node[ne3] (s8) [below=of s6] {\tt main$_{\texttt{S42}}$};
      \node[gray] (sl) [text width=2.5cm] at (-4,-14.3) {\tt System S4.2 Section};
      % arrows
      \draw[->] (e1.south west) -- (u1.north);
      \draw[->] (e1.east) -- (e2.west);
      \draw[->] (e1.south east) to[out=300, in=150, looseness=1] (a4.west);
      \draw[->] (u1.south) -- (t1.north);
      \draw[->] (t1.south) -- (t2.north);
      \draw[->] (e2.east) to[out=340, in=40, looseness=1.2] (a2.north);
      \draw[->] (a1.south) -- (a2.north);
      \draw[->] (a1.west) to[out=185, in=40, looseness=1.3] (a5.north east);
      \draw[->] (a1.east) to[out=350, in=45, looseness=1.3] (s1.north east);
      \draw[->] (a1.west) to[out=200, in=15, looseness=1.2] (s2.east);
      \draw[->] (a2.south) -- (a3.north);
      \draw[->] (a3.south) -- (a4.north);
      \draw[->] (a5.south) to[out=205, in=105, looseness=1.3] (s3.north west);
      %\draw[->] (a5.south) -- (s4.north west);
      \draw[->] (s2.south) -- (s3.north);
      \draw[->] (s3.south) -- (s4.north);
      \draw[->] (s1.south west) -- (s4.north east);
      \draw[->] (s1.south) -- (s5.north);
      \draw[->] (e1.south east) to[out=320, in=150, looseness=1.2] (s5.west);
      \draw[->] (a3.south east) to[out=320, in=60] (s5.north);
      \draw[->] (e3.south) to[out=280, in=45, looseness=1] (s5.north);
      %\draw[->] (s1.south west) to[out=210, in=160, looseness=1.6] (s6.west);
      \draw[->] (s5.south) -- (s6.north);
      \draw[->] (s7.south east) -- (s8.north west);
      \draw[->] (s4.south) to[out=310, in=180, looseness=0.5] (s8.west);
      \draw[->] (s6.south) -- (s8.north);
    \end{tikzpicture}
    \caption{Dependency graph showing the main results and the definitions, abbreviations, and intermediate results from their proofs that require the countable type condition. The dotted lines and the gray text show the files or sections of the Isabelle theory corresponding to our formalization where these can be found. Definitions and abbreviations appear in rounded rectangles, whereas lemmas and theorems appear in rectangles. Those that explicitly mention the countability condition are colored in blue, and the color orange means that this result is applied using the set of natural numbers to label the agents.}
    \label{fig:dependency}
\end{figure}
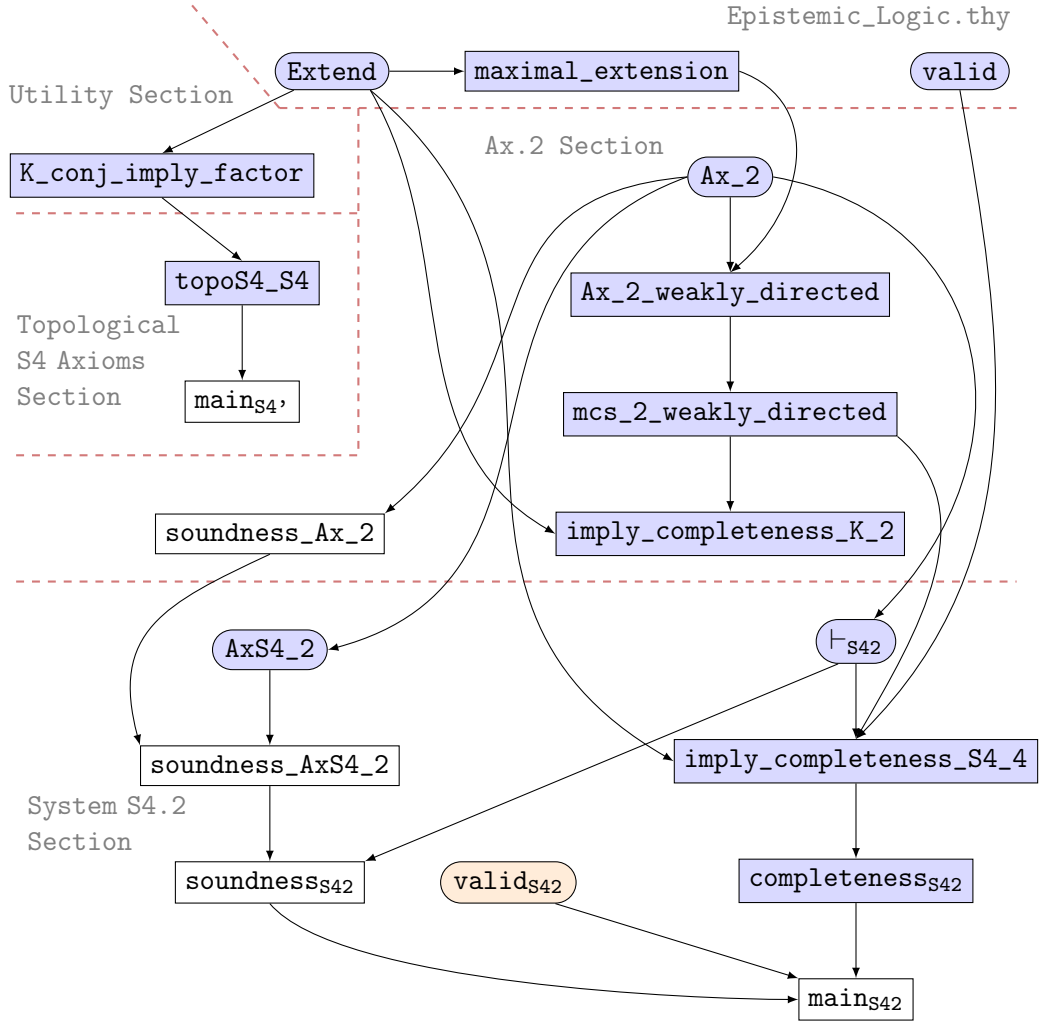

\begin{lemma}
    For all formulas $\phi,\psi,\theta$, it is the case that
    \[ \vdash ((K_i \phi \wedge K_i \psi) \to \theta) \to (K_i (\phi \wedge \psi) \to \theta) \]
    for any agent label $i$, as long as the type of $i$ is countable.
    \begin{isabelle}
      \SnippetPart{0}{2}{lemma:K-conj-imply-factor}
    \end{isabelle}
\end{lemma}
The assumption over the set of agent labels for the previous lemma is imposed by the proof that was formalized for it, as it relies on the proof for the completeness of K in \cite{EpistemicLogicAFP}, which requires this condition, as depicted in Figure \ref{fig:dependency}.

Additionally, we formalize the following lemma, which plays a significant role in the proof of the completeness result for Stalnaker's epistemic logic that follows \cite{lectureNotesStalnaker}.
\begin{lemma}
  Given any pair of formulas $\phi, \psi$, $(K_i \phi \wedge L_i \psi) \to L_i (\phi \wedge \psi)$ is a theorem in system K, for every agent label $i$.
\end{lemma}
\begin{proof}
    Notice that, for any formulas $\phi$ and $\psi$, $\vdash \phi \to (\neg (\phi \wedge \psi) \to \neg \psi)$, hence 
    \[ \vdash K_i \phi \to K_i (\neg (\phi \wedge \psi) \to \neg \psi). \] 
    On the other hand, we have that $\vdash K_i (\neg (\phi \wedge \psi) \to \neg \psi) \to ((K_i \neg (\phi \wedge \psi)) \to K_i \neg \psi)$, thus 
    \[ \vdash K_i \phi \to ((K_i \neg (\phi \wedge \psi)) \to K_i \neg \psi). \] 
    This implies that $\vdash K_i \phi \to (L_i \psi \to L_i (\phi \wedge \psi))$, which is equivalent to 
    \[ \vdash (K_i \phi \wedge L_i \psi) \to L_i (\phi \wedge \psi). \]
\end{proof}

%%% --------------------------------------------------------- %%%
\subsection{Axiom .2} \label{sec_Ax2}

We formalize axiom schema .2, which when added to the axioms and rules of system K imposes a structure on the canonical model, namely, we obtain a weakly directed frame. For this, the \textbf{inductive} command lets us define the axiom schema in such a way that $i$ and $p$ can be instantiated at will, as long as the type for the agents labels is countable.
\begin{isabelle}
  \Snippet{inductive:Ax-2}
\end{isabelle}
% The reader might notice here that we restrict the definition to countable types for the agents labels ($'i$), which in principle does not have anything to do with the definition of the axiom schema. However, {\color{blue} this is done as many of the main results are only obtained under that assumption, which had to be introduced since the beginning for the arguments to work.  }

A frame $\mathcal{F} = (W,(R)_{i\in I})$ is said to be \textit{weakly directed} if whenever $vR_i w$ and $vR_i u$, there exists $x\in W$ such that $wR_i x$ and $uR_i x$, for each $i \in I$. Accordingly, we formalize this property for Kripke frames as follows:
\begin{isabelle}
  \Snippet{definition:weakly-directed}
\end{isabelle}

The soundness of axiom schema .2 with respect to weakly directed frames is formalized in our Isabelle theory, and it follows from the definitions for the semantics and the weakly directed property. However, proving that the property holds for the canonical model when adding the axiom to a normal modal logic is non-trivial. Unlike the frame properties imposed by the axioms considered in the Epistemic Logics formalized in \cite{From2021}, which are all universal properties, this property is universal-existential, so to prove that the canonical model has this property means that one has to show the existence of a possible world satisfying a property under some assumptions.

Recall that the \emph{canonical frame}, $\mathcal{F}^{\operatorname{can}} = \left( W^{\operatorname{can}}, (R_i^{\operatorname{can}})_{i\in I} \right)$, consists of the set all maximal consistent sets of formulas (with respect to K+.2) as the set of possible worlds, $W^{\operatorname{can}}$, and the accessibility relations $R_i^{\operatorname{can}}$ are defined as follows:
\[ \Gamma R_i^{\operatorname{can}} \Delta \quad \text{ iff } \quad \{ \phi:\, K_i \phi \in \Gamma\} \subseteq \Delta, \]
for each agent $i$. Thus, showing that the canonical frame for a system including axiom .2 is weakly directed involves verifying that if we have $\{ \phi:\, K_i \phi \in \Gamma\} \subseteq \Delta_1$ and $\{ \phi:\, K_i \phi \in \Gamma\} \subseteq \Delta_2$ for some maximal consistent sets $\Gamma$, $\Delta_1$ and $\Delta_2$, then there exists a maximal consistent set $\Theta$ such that $\{ \phi:\, K_i \phi \in \Delta_1 \} \cup \{ \phi:\, K_i \phi \in \Delta_2 \} \subseteq \Theta$. We capture this in our formalization by the following lemma.

\begin{lemma}[Weakly directed property and the axiom .2] \label{lemma_wdp_Ax2}
    Suppose that $V,U,W$ are three maximal consistent sets with respect to a normal modal logic containing the axiom .2. If $V R_i^{\operatorname{can}} U$ and $V R_i^{\operatorname{can}} W$, then there exists a maximal consistent set $X$ such that $U R_i^{\operatorname{can}} X$ and $W R_i^{\operatorname{can}} X$.
\end{lemma}

\begin{proof}
First, fix a set of formulas $A$ and three maximal consistent sets of formulas $V,U,W$ (with respect to $A$) satisfying the lemma assumptions for some agent label $i$ of a countable type. Assume towards contradiction that such a set $X$ does not exist, then
    \[ S := \{ \phi:\, K_i \phi \in W\} \cup \{ \phi:\, K_i \phi \in U\} \]
    has to be inconsistent with respect to $A$, hence there are formulas $\theta_1, \ldots, \theta_k \in \{ \phi:\, K_i \phi \in U\}$ and $\psi_1, \ldots, \psi_m \in \{ \phi:\, K_i \phi \in W\}$, for some $k,m \in \mathbb{N}$, such that
    \[ A \vdash (\alpha \wedge \beta) \to \bot, \]
    where $\alpha = \theta_1 \wedge \ldots \wedge \theta_k$ and $\beta = \psi_1 \wedge \ldots \wedge \psi_m$. This implies that $A \vdash K_i K_i (\neg (\alpha \wedge \beta))$, since $\phi \to \bot$ is equivalent to $\neg \phi$ for every formula $\phi$, by applying the Necessitation rule twice. By definition, we have that $K_i \theta_1, \ldots, K_i \theta_k \in U$ and $K_i \psi_1, \ldots, K_i \psi_m \in W$, thus
    \[ K_i \theta_1 \wedge \ldots \wedge K_i \theta_k \in U \quad \text{ and } \quad K_i \psi_1 \wedge \ldots \wedge K_i \psi_m \in W, \]
    since these sets are closed under logical consequences. We then use the logical equivalences and properties for maximal consistent sets from the \emph{Utility section} (see section \ref{sec_rewritting}) to obtain that $K_i \alpha \in U$ and $K_i \beta \in W$. This implies that the formulas $L_i K_i \alpha$ and $L_i K_i \beta$ are elements of $V$, and so is the formula $K_i L_i \alpha$, since $V$ contains every instance of axiom .2 and is closed under logical consequences. This implies that $(L_i K_i \beta) \wedge (K_i L_i \alpha) \in V$, thus $L_i (K_i \beta \wedge L_i \alpha) \in V$, so there exists a maximal consistent set $Z$ such that
    \[ V R_i^{\operatorname{can}} Z \qquad \text{and} \qquad  K_i \beta \wedge L_i \alpha \in Z. \]
    Applying the lemma \textit{K-thm} we get that $L_i (\beta \wedge \alpha) \in Z$, but notice that $K_i \neg (\alpha \wedge \beta) \in Z$, thus $K_i \neg (\beta \wedge \alpha) \in Z$, which is a contradiction because we have found a formula $\phi$ such that $\phi, \neg \phi \in Z$. 
\end{proof}

Note that we have restricted ourselves to countable types for the set of agent labels in formalization of the previous two lemmas, as we are only allowed to extend a consistent set into a maximal one when the language is countable, because of a dependency shown in Figure \ref{fig:dependency}. Unlike in the respective result for each normal modal logic formalized in \cite{From2021}, this restriction to a countable type was necessary as we were dealing with a universal-existence property and not with a purely universal one. We then prove a version of the Truth lemma for the minimal normal modal logic that includes axiom .2, which is the relevant result about this system that will allow us to prove the completeness result for system S4.2 in the next section.

\begin{lemma}[Imply completeness for Axiom .2] \label{lemma_imply_completeness_Ax2}
Let $\Gamma \cup \{ \phi\}$ be a set of formulas. Suppose that, for all weakly directed Kripke structures $M$, $M,w \models \Gamma$ implies $M,w \models \phi$, for each world $w \in M$. Then there are formulas $\gamma_1, \gamma_2, \ldots, \gamma_m \in \Gamma$ such that
\[ \vdash_{.2} \gamma_1 \to (\gamma_2 \to \ldots (\gamma_m \to \phi)\ldots). \]
\end{lemma}
We omit the proof for this lemma, since it follows the same strategy as the correspondent ones for the systems formalized in \cite{EpistemicLogicAFP}.

%%% --------------------------------------------------------- %%%
\subsection{System S4.2}

We define system S4.2 as the one obtained by adding to system K the axioms T, 4 and .2, by making use of the abbreviation $\oplus$ introduced in \cite{EpistemicLogicAFP} that allows combining axiom predicates:
\begin{isabelle}
  \Snippet{abbreviation:SystemS4-2}
\end{isabelle}

% We treat this system as Stalnaker's epistemic logic

Recall that axioms T and 4 impose reflexivity and transitivity on the canonical frame, respectively \cite{Blackburn2001}, which was formalized in \cite{EpistemicLogicAFP}. This implies that the composition of these two with axiom .2 imposes all three conditions on the canonical frame, which leads to the soundness and completeness of S4.2 with respect to all weakly directed pre-orders. To prove the completeness result, we prove first the analog results to Lemmas \ref{lemma_wdp_Ax2} and \ref{lemma_imply_completeness_Ax2} but for S4.2 and Kripke models based on weakly directed preorders.

\begin{lemma}[S4.2 and Weakly directed preorders]
Let $\Gamma \cup \{\phi\}$ be a set of formulas. Suppose that, for all countable Kripke structures $M$ based on weakly directed preorders, $M,w \models \Gamma$ implies $M,w \models \phi$, for each world $w \in M$. Then, there are formulas $\gamma_1, \ldots, \gamma_m \in \Gamma$ such that
\[ \vdash_{S42} \gamma_1 \to ( \ldots \to (\gamma_m \to \phi) \ldots ). \]

\begin{isabelle}
  \SnippetPart{0}{3}{lemma:imply-completeness-S4-2}
\end{isabelle}

This implies that if a formula is valid in all countable Kripke structures based on weakly directed preorders, then it is a theorem in S4.2.
\begin{isabelle}
  \SnippetPart{0}{2}{lemma:completenessS42}
\end{isabelle}
\end{lemma}
%We omit the proofs for these lemmas on this document, since these are done simply by combining the correspondent results for each individual axiom scheme.

Our main result follows: the completeness of S4.2 with respect to all frames consisting of weakly directed pre-orders.
\begin{theorem}[Completeness of S4.2]
A formula is valid in all countable Kripke structures based on weakly directed preorders if and only if it is a theorem in S4.2.
\begin{isabelle}
  \SnippetPart{0}{0}{theorem:mainS42}
\end{isabelle}
\end{theorem}

%%%%%%%%%%%%%%%%%%%%%%%%%%%%%%%%%%%%%%%%%%%%%%%%%%%%%%%%%%%%%%%%%%
\section{An alternative axiomatization for System S4} \label{sec_topoS4}

Inspired by the last section of \cite{EpistemicLogicAFP}, we formalize an alternative axiomatization for System S4 that is often used when dealing with the \textit{topological semantics} \cite{Aiello2003} for modal operators and we show its equivalence to the one considered in \cite{From2021}.
We formalize the system corresponding to the axioms and rules in Table \ref{tab_axioms_TopoS4} in Isabelle, which we call topoS4, and, if a formula $\phi$ is a theorem in this system, we denote this by $\vdash_{Top} \phi$.
\begin{isabelle}
  \Snippet{inductive:System-topoS4}
\end{isabelle}

To show that this formulation is equivalent to the one in \cite{From2021} (which is based on \cite{Blackburn2001}), we provide derivations of axiom K and the Necessitation Rule (from $\phi$ deduct $\square \phi$). This is enough as our system already includes axioms T and 4 in the same fashion as in \cite{From2021}, and is based on the same propositional logic.
\begin{lemma}
For all formulas $\phi$ and $\psi$, $\vdash_{top} (K_i \phi \wedge K_i (\phi \to \psi)) \to K_i \psi$ and, if $\vdash_{top} \phi$, then $\vdash_{top} K_i \phi$.
\end{lemma}
\begin{proof}
	For the first part, notice that $\vdash_{top} (\phi \wedge (\phi \to \psi)) \to \psi$, since it is an instance of a propositional tautology. Then, we apply RM to obtain that $\vdash_{top} K_i (\phi \wedge (\phi \to \psi)) \to K_i \psi$, which implies that $\vdash_{top} (K_i \phi \wedge K_i (\phi \to \psi)) \to K_i \psi$. For the second one, suppose that $\vdash_{top} \phi$ and notice that $\vdash_{top} \phi \to (\top \to \phi)$, hence $\vdash_{top} \top \to \phi$. Applying RM we get that $\vdash_{top} K_i \top \to K_i \phi$, thus $\vdash_{top} K_i \phi$.
\end{proof}
%We omit the proofs, which relies on axiom R and the rule M, as well as the equivalence between any formula $\phi$ and the formula $\top \to \phi$.
From this it then follows that any formula derivable in the classical S4 system (denoted $\vdash_{S4}$) can be derived in our system as well.
\begin{lemma}
All S4 theorems are theorems in topoS4.
\begin{isabelle}
  \SnippetPart{0}{0}{lemma:S4-topoS4}
\end{isabelle}
\end{lemma}

The converse follows by a similar argument, we show that axioms and rules from our system are all derivable in $\vdash_{S4}$, under the condition that there are only countably many agents.
\begin{lemma}
All theorems in topoS4 are theorems in S4, assuming that there are only countably many agents.
\begin{isabelle}
  \SnippetPart{0}{2}{lemma:topoS4-S4}
\end{isabelle}
\end{lemma}

By combining the last two results with the main result for S4 in \cite{From2021}, we obtain formalized soundness and completeness for this alternative axiomatization of S4 over the class of S4 frames, namely, all reflexive and transitive frames.
\begin{theorem}[Soundness and Completeness of topoS4] A formula is valid in all S4 Kripke models if and only if it is a theorem in topoS4.
\begin{isabelle}
  \SnippetPart{0}{0}{theorem:mainS4'}
\end{isabelle}
\end{theorem}

%%%%%%%%%%%%%%%%%%%%%%%%%%%%%%%%%%%%%%%%%%%%%%%%%%%%%%%%%%%%%%%%%%
\section{Results, Discussion, and Future work} \label{sec_results_future_work}

We have formalized the soundness and completeness for Stalnaker's Epistemic Logic S4.2 with respect to the class of Kripke frames consisting of weakly-directed pre-orders for countably many agents, which has not been formalized before neither in Isabelle, nor in any other publicly available proof assistant. Additionally, the equivalence between the topological axiomatization of S4 and the one in \cite{EpistemicLogicAFP} is also described in this document. The proofs for the main result, as well as for many of the intermediate results, have been sketched before in multiple sources, but we were not able to find a unique source, making this the first work of its kind. Additionally, given the recent interest in applications of the topological semantics for epistemic modal operators \cite{2019TopoEvidMulti,BaltagEtAl2016JustTopoBel,BaltagOzgun2013topobelief}, some of which coincide with Stalnaker's epistemic logic, this provides a reinforcement for the foundations of these works.

We emphasize on the assumption of the cardinality of the set of agent labels, as it was necessary to impose such restriction even in some definitions in our formalization, thus creating a discrepancy with the definitions commonly found in the literature. We present in Figure \ref{fig:dependency} a summary of the definitions and results of our formalization and the one in \cite{EpistemicLogicAFP} that rely on this condition, since the formalization for the general strategy applied to obtain the completeness results in \cite{EpistemicLogicAFP} requires it to be able to obtain maximal consistent sets. Although, in theory, it is possible to provide an argument for the general case using Zorn's lemma (this was also later noted in \cite{From2021}), which is available in \cite{ZornLemma_Isabelle}. 

Further work in formalizing in Isabelle/HOL of different formal aspects of modal logics that include S4 operators, as is the case with many temporal logics like LTL with its \emph{always} operator, for which a complete axiomatization is already known \cite{Goldblatt1992}; as well as concrete examples of epistemic scenarios based on Stalnaker's principles, like the example detailed in \cite{koutras2015relating}. We hope that this work will facilitate further work in formalizing different logical systems in Isabelle/HOL.

% ---- Bibliography ----
\bibliographystyle{eptcs}
% \nocite{*}
\bibliography{ref.bib}

\end{document}

%% file: snips.tex
\SNIP{theory:Stalnaker-Logic-0}{%
\isacommand{theory}\isamarkupfalse%
\ Stalnaker{\isacharunderscore}{\kern0pt}Logic\isanewline%
}%
\SNIP{theory:Stalnaker-Logic-1}{%
\ \ \isakeyword{imports}\ {\isachardoublequoteopen}Epistemic{\isacharunderscore}{\kern0pt}Logic{\isachardot}{\kern0pt}Epistemic{\isacharunderscore}{\kern0pt}Logic{\isachardoublequoteclose}%
}%
\SNIP{lemma:duality-taut-0-0}{%
tautology\ {\isacharparenleft}{\kern0pt}{\isacharparenleft}{\kern0pt}{\isacharparenleft}{\kern0pt}K\ i\ p{\isacharparenright}{\kern0pt}\ \isactrlbold {\isasymlongrightarrow}\ K\ i\ {\isacharparenleft}{\kern0pt}\isactrlbold {\isasymnot}q{\isacharparenright}{\kern0pt}{\isacharparenright}{\kern0pt}\isactrlbold {\isasymlongrightarrow}\ {\isacharparenleft}{\kern0pt}{\isacharparenleft}{\kern0pt}L\ i\ q{\isacharparenright}{\kern0pt}\ \isactrlbold {\isasymlongrightarrow}\ {\isacharparenleft}{\kern0pt}\isactrlbold {\isasymnot}\ K\ i\ p{\isacharparenright}{\kern0pt}{\isacharparenright}{\kern0pt}{\isacharparenright}{\kern0pt}%
}%
\SNIP{lemma:duality-taut-0}{%
\isacommand{lemma}\isamarkupfalse%
\ duality{\isacharunderscore}{\kern0pt}taut{\isacharcolon}{\kern0pt}\ {\Cartouche{lemma:duality-taut}{0}{0}}\isanewline%
}%
\SNIP{lemma:duality-taut-1}{%
\isadelimproof
\ \ %
\endisadelimproof
\isatagproof
\isacommand{by}\isamarkupfalse%
\ force%
\endisatagproof
{\isafoldproof}%
\isadelimproof%
}%
\SNIP{lemma:K-imp-trans-0}{%
\isacommand{lemma}\isamarkupfalse%
\ K{\isacharunderscore}{\kern0pt}imp{\isacharunderscore}{\kern0pt}trans{\isacharcolon}{\kern0pt}\ \isanewline%
}%
\SNIP{lemma:K-imp-trans-1-0}{%
A\ {\isasymturnstile}\ {\isacharparenleft}{\kern0pt}p\ \isactrlbold {\isasymlongrightarrow}\ q{\isacharparenright}{\kern0pt}%
}%
\SNIP{lemma:K-imp-trans-1-1}{%
A\ {\isasymturnstile}\ {\isacharparenleft}{\kern0pt}q\ \isactrlbold {\isasymlongrightarrow}\ r{\isacharparenright}{\kern0pt}%
}%
\SNIP{lemma:K-imp-trans-1}{%
\ \ \isakeyword{assumes}\ {\Cartouche{lemma:K-imp-trans}{1}{0}}\ \ {\Cartouche{lemma:K-imp-trans}{1}{1}}\isanewline%
}%
\SNIP{lemma:K-imp-trans-2-0}{%
A\ {\isasymturnstile}\ {\isacharparenleft}{\kern0pt}p\ \isactrlbold {\isasymlongrightarrow}\ r{\isacharparenright}{\kern0pt}%
}%
\SNIP{lemma:K-imp-trans-2}{%
\ \ \isakeyword{shows}\ \ {\Cartouche{lemma:K-imp-trans}{2}{0}}\isanewline%
}%
\SNIP{lemma:K-imp-trans-3}{%
\isadelimproof
\endisadelimproof
\isatagproof
\isacommand{proof}\isamarkupfalse%
\ {\isacharminus}{\kern0pt}\ \isanewline%
}%
\SNIP{lemma:K-imp-trans-4-0}{%
tautology\ \ {\isacharparenleft}{\kern0pt}{\isacharparenleft}{\kern0pt}p\ \isactrlbold {\isasymlongrightarrow}\ q{\isacharparenright}{\kern0pt}\ \isactrlbold {\isasymlongrightarrow}\ {\isacharparenleft}{\kern0pt}\ {\isacharparenleft}{\kern0pt}q\ \isactrlbold {\isasymlongrightarrow}\ r{\isacharparenright}{\kern0pt}\ \isactrlbold {\isasymlongrightarrow}\ \ {\isacharparenleft}{\kern0pt}p\ \isactrlbold {\isasymlongrightarrow}\ r{\isacharparenright}{\kern0pt}{\isacharparenright}{\kern0pt}{\isacharparenright}{\kern0pt}%
}%
\SNIP{lemma:K-imp-trans-4}{%
\ \ \isacommand{have}\isamarkupfalse%
\ {\Cartouche{lemma:K-imp-trans}{4}{0}}\isanewline%
}%
\SNIP{lemma:K-imp-trans-5}{%
\ \ \ \ \isacommand{by}\isamarkupfalse%
\ fastforce\isanewline%
}%
\SNIP{lemma:K-imp-trans-6}{%
\ \ \isacommand{then}\isamarkupfalse%
\ \isacommand{show}\isamarkupfalse%
\ {\isacharquery}{\kern0pt}thesis\isanewline%
}%
\SNIP{lemma:K-imp-trans-7}{%
\ \ \ \ \isacommand{by}\isamarkupfalse%
\ {\isacharparenleft}{\kern0pt}meson\ A{\isadigit{1}}\ R{\isadigit{1}}\ assms{\isacharparenleft}{\kern0pt}{\isadigit{1}}{\isacharparenright}{\kern0pt}\ assms{\isacharparenleft}{\kern0pt}{\isadigit{2}}{\isacharparenright}{\kern0pt}{\isacharparenright}{\kern0pt}\isanewline%
}%
\SNIP{lemma:K-imp-trans-8}{%
\isacommand{qed}\isamarkupfalse%
\endisatagproof
{\isafoldproof}%
\isadelimproof%
}%
\SNIP{lemma:K-imp-trans'-0}{%
\isacommand{lemma}\isamarkupfalse%
\ K{\isacharunderscore}{\kern0pt}imp{\isacharunderscore}{\kern0pt}trans{\isacharprime}{\kern0pt}{\isacharcolon}{\kern0pt}\ \isanewline%
}%
\SNIP{lemma:K-imp-trans'-1-0}{%
A\ {\isasymturnstile}\ {\isacharparenleft}{\kern0pt}q\ \isactrlbold {\isasymlongrightarrow}\ r{\isacharparenright}{\kern0pt}%
}%
\SNIP{lemma:K-imp-trans'-1}{%
\ \ \isakeyword{assumes}\ {\Cartouche{lemma:K-imp-trans'}{1}{0}}\isanewline%
}%
\SNIP{lemma:K-imp-trans'-2-0}{%
A\ {\isasymturnstile}\ {\isacharparenleft}{\kern0pt}{\isacharparenleft}{\kern0pt}p\ \isactrlbold {\isasymlongrightarrow}\ q{\isacharparenright}{\kern0pt}\ \isactrlbold {\isasymlongrightarrow}\ {\isacharparenleft}{\kern0pt}p\ \isactrlbold {\isasymlongrightarrow}\ r{\isacharparenright}{\kern0pt}{\isacharparenright}{\kern0pt}%
}%
\SNIP{lemma:K-imp-trans'-2}{%
\ \ \isakeyword{shows}\ \ {\Cartouche{lemma:K-imp-trans'}{2}{0}}\isanewline%
}%
\SNIP{lemma:K-imp-trans'-3}{%
\isadelimproof
\endisadelimproof
\isatagproof
\isacommand{proof}\isamarkupfalse%
\ {\isacharminus}{\kern0pt}\ \isanewline%
}%
\SNIP{lemma:K-imp-trans'-4-0}{%
tautology\ \ {\isacharparenleft}{\kern0pt}{\isacharparenleft}{\kern0pt}q\ \isactrlbold {\isasymlongrightarrow}\ r{\isacharparenright}{\kern0pt}\ \isactrlbold {\isasymlongrightarrow}\ {\isacharparenleft}{\kern0pt}{\isacharparenleft}{\kern0pt}p\ \isactrlbold {\isasymlongrightarrow}\ q{\isacharparenright}{\kern0pt}\ \isactrlbold {\isasymlongrightarrow}\ {\isacharparenleft}{\kern0pt}p\ \isactrlbold {\isasymlongrightarrow}\ r{\isacharparenright}{\kern0pt}{\isacharparenright}{\kern0pt}{\isacharparenright}{\kern0pt}%
}%
\SNIP{lemma:K-imp-trans'-4}{%
\ \ \isacommand{have}\isamarkupfalse%
\ {\Cartouche{lemma:K-imp-trans'}{4}{0}}\isanewline%
}%
\SNIP{lemma:K-imp-trans'-5}{%
\ \ \ \ \isacommand{by}\isamarkupfalse%
\ fastforce\isanewline%
}%
\SNIP{lemma:K-imp-trans'-6}{%
\ \ \isacommand{then}\isamarkupfalse%
\ \isacommand{show}\isamarkupfalse%
\ {\isacharquery}{\kern0pt}thesis\ \isanewline%
}%
\SNIP{lemma:K-imp-trans'-7}{%
\ \ \ \ \isacommand{using}\isamarkupfalse%
\ A{\isadigit{1}}\ R{\isadigit{1}}\ assms\ \isacommand{by}\isamarkupfalse%
\ blast\isanewline%
}%
\SNIP{lemma:K-imp-trans'-8}{%
\isacommand{qed}\isamarkupfalse%
\endisatagproof
{\isafoldproof}%
\isadelimproof%
}%
\SNIP{lemma:K-imply-multi-0}{%
\isacommand{lemma}\isamarkupfalse%
\ K{\isacharunderscore}{\kern0pt}imply{\isacharunderscore}{\kern0pt}multi{\isacharcolon}{\kern0pt}\isanewline%
}%
\SNIP{lemma:K-imply-multi-1-0}{%
A\ {\isasymturnstile}\ {\isacharparenleft}{\kern0pt}a\ \isactrlbold {\isasymlongrightarrow}\ b{\isacharparenright}{\kern0pt}%
}%
\SNIP{lemma:K-imply-multi-1-1}{%
A\ {\isasymturnstile}\ {\isacharparenleft}{\kern0pt}a\ \isactrlbold {\isasymlongrightarrow}\ c{\isacharparenright}{\kern0pt}%
}%
\SNIP{lemma:K-imply-multi-1}{%
\ \ \isakeyword{assumes}\ {\Cartouche{lemma:K-imply-multi}{1}{0}}\ \isakeyword{and}\ {\Cartouche{lemma:K-imply-multi}{1}{1}}\isanewline%
}%
\SNIP{lemma:K-imply-multi-2-0}{%
A\ {\isasymturnstile}\ {\isacharparenleft}{\kern0pt}a\ \isactrlbold {\isasymlongrightarrow}\ {\isacharparenleft}{\kern0pt}b\isactrlbold {\isasymand}c{\isacharparenright}{\kern0pt}{\isacharparenright}{\kern0pt}%
}%
\SNIP{lemma:K-imply-multi-2}{%
\ \ \isakeyword{shows}\ {\Cartouche{lemma:K-imply-multi}{2}{0}}\isanewline%
}%
\SNIP{lemma:K-imply-multi-3}{%
\isadelimproof
\endisadelimproof
\isatagproof
\isacommand{proof}\isamarkupfalse%
\ {\isacharminus}{\kern0pt}\isanewline%
}%
\SNIP{lemma:K-imply-multi-4-0}{%
tautology\ {\isacharparenleft}{\kern0pt}{\isacharparenleft}{\kern0pt}a\isactrlbold {\isasymlongrightarrow}b{\isacharparenright}{\kern0pt}\isactrlbold {\isasymlongrightarrow}{\isacharparenleft}{\kern0pt}a\isactrlbold {\isasymlongrightarrow}c{\isacharparenright}{\kern0pt}\isactrlbold {\isasymlongrightarrow}{\isacharparenleft}{\kern0pt}a\isactrlbold {\isasymlongrightarrow}{\isacharparenleft}{\kern0pt}b\isactrlbold {\isasymand}c{\isacharparenright}{\kern0pt}{\isacharparenright}{\kern0pt}{\isacharparenright}{\kern0pt}%
}%
\SNIP{lemma:K-imply-multi-4}{%
\ \ \isacommand{have}\isamarkupfalse%
\ {\Cartouche{lemma:K-imply-multi}{4}{0}}\isanewline%
}%
\SNIP{lemma:K-imply-multi-5}{%
\ \ \ \ \isacommand{by}\isamarkupfalse%
\ force\isanewline%
}%
\SNIP{lemma:K-imply-multi-6-0}{%
A\ {\isasymturnstile}\ {\isacharparenleft}{\kern0pt}{\isacharparenleft}{\kern0pt}a\isactrlbold {\isasymlongrightarrow}b{\isacharparenright}{\kern0pt}\isactrlbold {\isasymlongrightarrow}{\isacharparenleft}{\kern0pt}a\isactrlbold {\isasymlongrightarrow}c{\isacharparenright}{\kern0pt}\isactrlbold {\isasymlongrightarrow}{\isacharparenleft}{\kern0pt}a\isactrlbold {\isasymlongrightarrow}{\isacharparenleft}{\kern0pt}b\isactrlbold {\isasymand}c{\isacharparenright}{\kern0pt}{\isacharparenright}{\kern0pt}{\isacharparenright}{\kern0pt}%
}%
\SNIP{lemma:K-imply-multi-6}{%
\ \ \isacommand{then}\isamarkupfalse%
\ \isacommand{have}\isamarkupfalse%
\ {\Cartouche{lemma:K-imply-multi}{6}{0}}\isanewline%
}%
\SNIP{lemma:K-imply-multi-7}{%
\ \ \ \ \isacommand{using}\isamarkupfalse%
\ A{\isadigit{1}}\ \isacommand{by}\isamarkupfalse%
\ blast\isanewline%
}%
\SNIP{lemma:K-imply-multi-8-0}{%
A\ {\isasymturnstile}\ {\isacharparenleft}{\kern0pt}{\isacharparenleft}{\kern0pt}a\isactrlbold {\isasymlongrightarrow}c{\isacharparenright}{\kern0pt}\isactrlbold {\isasymlongrightarrow}{\isacharparenleft}{\kern0pt}a\isactrlbold {\isasymlongrightarrow}{\isacharparenleft}{\kern0pt}b\isactrlbold {\isasymand}c{\isacharparenright}{\kern0pt}{\isacharparenright}{\kern0pt}{\isacharparenright}{\kern0pt}%
}%
\SNIP{lemma:K-imply-multi-8}{%
\ \ \isacommand{then}\isamarkupfalse%
\ \isacommand{have}\isamarkupfalse%
\ {\Cartouche{lemma:K-imply-multi}{8}{0}}\isanewline%
}%
\SNIP{lemma:K-imply-multi-9}{%
\ \ \ \ \isacommand{using}\isamarkupfalse%
\ assms{\isacharparenleft}{\kern0pt}{\isadigit{1}}{\isacharparenright}{\kern0pt}\ R{\isadigit{1}}\ \isacommand{by}\isamarkupfalse%
\ blast\isanewline%
}%
\SNIP{lemma:K-imply-multi-10}{%
\ \ \isacommand{then}\isamarkupfalse%
\ \isacommand{show}\isamarkupfalse%
\ {\isacharquery}{\kern0pt}thesis\ \isanewline%
}%
\SNIP{lemma:K-imply-multi-11}{%
\ \ \ \ \isacommand{using}\isamarkupfalse%
\ assms{\isacharparenleft}{\kern0pt}{\isadigit{2}}{\isacharparenright}{\kern0pt}\ R{\isadigit{1}}\ \isacommand{by}\isamarkupfalse%
\ blast\isanewline%
}%
\SNIP{lemma:K-imply-multi-12}{%
\isacommand{qed}\isamarkupfalse%
\endisatagproof
{\isafoldproof}%
\isadelimproof%
}%
\SNIP{lemma:K-multi-imply-0}{%
\isacommand{lemma}\isamarkupfalse%
\ K{\isacharunderscore}{\kern0pt}multi{\isacharunderscore}{\kern0pt}imply{\isacharcolon}{\kern0pt}\isanewline%
}%
\SNIP{lemma:K-multi-imply-1-0}{%
A\ {\isasymturnstile}\ {\isacharparenleft}{\kern0pt}a\ \isactrlbold {\isasymlongrightarrow}\ b\ \isactrlbold {\isasymlongrightarrow}\ c{\isacharparenright}{\kern0pt}%
}%
\SNIP{lemma:K-multi-imply-1}{%
\ \ \isakeyword{assumes}\ {\Cartouche{lemma:K-multi-imply}{1}{0}}\isanewline%
}%
\SNIP{lemma:K-multi-imply-2-0}{%
A\ {\isasymturnstile}\ {\isacharparenleft}{\kern0pt}{\isacharparenleft}{\kern0pt}a\isactrlbold {\isasymand}b{\isacharparenright}{\kern0pt}\ \isactrlbold {\isasymlongrightarrow}\ c{\isacharparenright}{\kern0pt}%
}%
\SNIP{lemma:K-multi-imply-2}{%
\ \ \isakeyword{shows}\ {\Cartouche{lemma:K-multi-imply}{2}{0}}\isanewline%
}%
\SNIP{lemma:K-multi-imply-3}{%
\isadelimproof
\endisadelimproof
\isatagproof
\isacommand{proof}\isamarkupfalse%
\ {\isacharminus}{\kern0pt}\ \isanewline%
}%
\SNIP{lemma:K-multi-imply-4-0}{%
tautology\ {\isacharparenleft}{\kern0pt}{\isacharparenleft}{\kern0pt}a\ \isactrlbold {\isasymlongrightarrow}\ b\ \isactrlbold {\isasymlongrightarrow}\ c{\isacharparenright}{\kern0pt}\ \isactrlbold {\isasymlongrightarrow}\ {\isacharparenleft}{\kern0pt}{\isacharparenleft}{\kern0pt}a\isactrlbold {\isasymand}b{\isacharparenright}{\kern0pt}\ \isactrlbold {\isasymlongrightarrow}\ c{\isacharparenright}{\kern0pt}{\isacharparenright}{\kern0pt}%
}%
\SNIP{lemma:K-multi-imply-4}{%
\ \ \isacommand{have}\isamarkupfalse%
\ {\Cartouche{lemma:K-multi-imply}{4}{0}}\isanewline%
}%
\SNIP{lemma:K-multi-imply-5}{%
\ \ \ \ \isacommand{by}\isamarkupfalse%
\ force\isanewline%
}%
\SNIP{lemma:K-multi-imply-6-0}{%
A\ {\isasymturnstile}\ {\isacharparenleft}{\kern0pt}{\isacharparenleft}{\kern0pt}a\ \isactrlbold {\isasymlongrightarrow}\ b\ \isactrlbold {\isasymlongrightarrow}\ c{\isacharparenright}{\kern0pt}\ \isactrlbold {\isasymlongrightarrow}\ {\isacharparenleft}{\kern0pt}{\isacharparenleft}{\kern0pt}a\isactrlbold {\isasymand}b{\isacharparenright}{\kern0pt}\ \isactrlbold {\isasymlongrightarrow}\ c{\isacharparenright}{\kern0pt}{\isacharparenright}{\kern0pt}%
}%
\SNIP{lemma:K-multi-imply-6}{%
\ \ \isacommand{then}\isamarkupfalse%
\ \isacommand{have}\isamarkupfalse%
\ {\Cartouche{lemma:K-multi-imply}{6}{0}}\isanewline%
}%
\SNIP{lemma:K-multi-imply-7}{%
\ \ \ \ \isacommand{using}\isamarkupfalse%
\ A{\isadigit{1}}\ \isacommand{by}\isamarkupfalse%
\ blast\isanewline%
}%
\SNIP{lemma:K-multi-imply-8}{%
\ \ \isacommand{then}\isamarkupfalse%
\ \isacommand{show}\isamarkupfalse%
\ {\isacharquery}{\kern0pt}thesis\isanewline%
}%
\SNIP{lemma:K-multi-imply-9}{%
\ \ \ \ \isacommand{using}\isamarkupfalse%
\ assms\ R{\isadigit{1}}\ \isacommand{by}\isamarkupfalse%
\ blast\ \isanewline%
}%
\SNIP{lemma:K-multi-imply-10}{%
\isacommand{qed}\isamarkupfalse%
\endisatagproof
{\isafoldproof}%
\isadelimproof%
}%
\SNIP{lemma:K-thm-0-0}{%
A\ {\isasymturnstile}\ {\isacharparenleft}{\kern0pt}{\isacharparenleft}{\kern0pt}K\ i\ p{\isacharparenright}{\kern0pt}\ \isactrlbold {\isasymand}\ {\isacharparenleft}{\kern0pt}L\ i\ q{\isacharparenright}{\kern0pt}\ \isactrlbold {\isasymlongrightarrow}\ L\ i\ {\isacharparenleft}{\kern0pt}p\ \isactrlbold {\isasymand}\ q{\isacharparenright}{\kern0pt}{\isacharparenright}{\kern0pt}%
}%
\SNIP{lemma:K-thm-0}{%
\isacommand{lemma}\isamarkupfalse%
\ K{\isacharunderscore}{\kern0pt}thm{\isacharcolon}{\kern0pt}\ {\Cartouche{lemma:K-thm}{0}{0}}\isanewline%
}%
\SNIP{lemma:K-thm-1}{%
\isadelimproof
\endisadelimproof
\isatagproof
\isacommand{proof}\isamarkupfalse%
\ {\isacharminus}{\kern0pt}\isanewline%
}%
\SNIP{lemma:K-thm-2-0}{%
tautology\ {\isacharparenleft}{\kern0pt}p\ \isactrlbold {\isasymlongrightarrow}\ {\isacharparenleft}{\kern0pt}\isactrlbold {\isasymnot}{\isacharparenleft}{\kern0pt}p\ \isactrlbold {\isasymand}\ q{\isacharparenright}{\kern0pt}{\isacharparenright}{\kern0pt}\ \isactrlbold {\isasymlongrightarrow}\ \isactrlbold {\isasymnot}\ q{\isacharparenright}{\kern0pt}%
}%
\SNIP{lemma:K-thm-2}{%
\ \ \isacommand{have}\isamarkupfalse%
\ {\Cartouche{lemma:K-thm}{2}{0}}\isanewline%
}%
\SNIP{lemma:K-thm-3}{%
\ \ \ \ \isacommand{by}\isamarkupfalse%
\ force\isanewline%
}%
\SNIP{lemma:K-thm-4-0}{%
A\ {\isasymturnstile}\ {\isacharparenleft}{\kern0pt}p\ \isactrlbold {\isasymlongrightarrow}\ {\isacharparenleft}{\kern0pt}\isactrlbold {\isasymnot}{\isacharparenleft}{\kern0pt}p\ \isactrlbold {\isasymand}\ q{\isacharparenright}{\kern0pt}{\isacharparenright}{\kern0pt}\ \isactrlbold {\isasymlongrightarrow}\ \isactrlbold {\isasymnot}\ q{\isacharparenright}{\kern0pt}%
}%
\SNIP{lemma:K-thm-4}{%
\ \ \isacommand{then}\isamarkupfalse%
\ \isacommand{have}\isamarkupfalse%
\ {\Cartouche{lemma:K-thm}{4}{0}}\isanewline%
}%
\SNIP{lemma:K-thm-5}{%
\ \ \ \ \isacommand{by}\isamarkupfalse%
\ {\isacharparenleft}{\kern0pt}simp\ add{\isacharcolon}{\kern0pt}\ A{\isadigit{1}}{\isacharparenright}{\kern0pt}\isanewline%
}%
\SNIP{lemma:K-thm-6-0}{%
A\ {\isasymturnstile}\ {\isacharparenleft}{\kern0pt}{\isacharparenleft}{\kern0pt}K\ i\ p{\isacharparenright}{\kern0pt}\ \isactrlbold {\isasymlongrightarrow}\ K\ i\ {\isacharparenleft}{\kern0pt}{\isacharparenleft}{\kern0pt}\isactrlbold {\isasymnot}{\isacharparenleft}{\kern0pt}p\ \isactrlbold {\isasymand}\ q{\isacharparenright}{\kern0pt}{\isacharparenright}{\kern0pt}\ \isactrlbold {\isasymlongrightarrow}\ \isactrlbold {\isasymnot}\ q{\isacharparenright}{\kern0pt}{\isacharparenright}{\kern0pt}%
}%
\SNIP{lemma:K-thm-6}{%
\ \ \isacommand{then}\isamarkupfalse%
\ \isacommand{have}\isamarkupfalse%
\ {\Cartouche{lemma:K-thm}{6}{0}}\ \isanewline%
}%
\SNIP{lemma:K-thm-7-0}{%
A\ {\isasymturnstile}\ {\isacharparenleft}{\kern0pt}K\ i\ {\isacharparenleft}{\kern0pt}{\isacharparenleft}{\kern0pt}\isactrlbold {\isasymnot}{\isacharparenleft}{\kern0pt}p\ \isactrlbold {\isasymand}\ q{\isacharparenright}{\kern0pt}{\isacharparenright}{\kern0pt}\ \isactrlbold {\isasymlongrightarrow}\ \isactrlbold {\isasymnot}\ q{\isacharparenright}{\kern0pt}\ \isactrlbold {\isasymlongrightarrow}\ K\ i\ {\isacharparenleft}{\kern0pt}\isactrlbold {\isasymnot}{\isacharparenleft}{\kern0pt}p\ \isactrlbold {\isasymand}\ q{\isacharparenright}{\kern0pt}{\isacharparenright}{\kern0pt}\ \isactrlbold {\isasymlongrightarrow}\ K\ i\ {\isacharparenleft}{\kern0pt}\isactrlbold {\isasymnot}\ q{\isacharparenright}{\kern0pt}{\isacharparenright}{\kern0pt}%
}%
\SNIP{lemma:K-thm-7}{%
\ \ \ \ \ \ \ \ \isakeyword{and}\ {\Cartouche{lemma:K-thm}{7}{0}}\isanewline%
}%
\SNIP{lemma:K-thm-8}{%
\ \ \ \ \ \isacommand{apply}\isamarkupfalse%
\ {\isacharparenleft}{\kern0pt}simp\ add{\isacharcolon}{\kern0pt}\ K{\isacharunderscore}{\kern0pt}map{\isacharparenright}{\kern0pt}\isanewline%
}%
\SNIP{lemma:K-thm-9}{%
\ \ \ \ \isacommand{by}\isamarkupfalse%
\ {\isacharparenleft}{\kern0pt}meson\ K{\isacharunderscore}{\kern0pt}A{\isadigit{2}}{\isacharprime}{\kern0pt}{\isacharparenright}{\kern0pt}\isanewline%
}%
\SNIP{lemma:K-thm-10-0}{%
A\ {\isasymturnstile}\ {\isacharparenleft}{\kern0pt}{\isacharparenleft}{\kern0pt}K\ i\ p{\isacharparenright}{\kern0pt}\ \isactrlbold {\isasymlongrightarrow}\ K\ i\ {\isacharparenleft}{\kern0pt}\isactrlbold {\isasymnot}{\isacharparenleft}{\kern0pt}p\ \isactrlbold {\isasymand}\ q{\isacharparenright}{\kern0pt}{\isacharparenright}{\kern0pt}\ \isactrlbold {\isasymlongrightarrow}\ K\ i\ {\isacharparenleft}{\kern0pt}\isactrlbold {\isasymnot}\ q{\isacharparenright}{\kern0pt}{\isacharparenright}{\kern0pt}%
}%
\SNIP{lemma:K-thm-10}{%
\ \ \isacommand{then}\isamarkupfalse%
\ \isacommand{have}\isamarkupfalse%
\ {\Cartouche{lemma:K-thm}{10}{0}}\ \isanewline%
}%
\SNIP{lemma:K-thm-11}{%
\ \ \ \ \isacommand{using}\isamarkupfalse%
\ K{\isacharunderscore}{\kern0pt}imp{\isacharunderscore}{\kern0pt}trans\ \isacommand{by}\isamarkupfalse%
\ blast\isanewline%
}%
\SNIP{lemma:K-thm-12-0}{%
A\ {\isasymturnstile}\ {\isacharparenleft}{\kern0pt}{\isacharparenleft}{\kern0pt}K\ i\ p{\isacharparenright}{\kern0pt}\ \isactrlbold {\isasymlongrightarrow}\ L\ i\ {\isacharparenleft}{\kern0pt}q{\isacharparenright}{\kern0pt}\ \isactrlbold {\isasymlongrightarrow}\ L\ i\ {\isacharparenleft}{\kern0pt}p\ \isactrlbold {\isasymand}\ q{\isacharparenright}{\kern0pt}{\isacharparenright}{\kern0pt}%
}%
\SNIP{lemma:K-thm-12}{%
\ \ \isacommand{then}\isamarkupfalse%
\ \isacommand{have}\isamarkupfalse%
\ {\Cartouche{lemma:K-thm}{12}{0}}\ \isanewline%
}%
\SNIP{lemma:K-thm-13}{%
\ \ \ \ \isacommand{by}\isamarkupfalse%
\ {\isacharparenleft}{\kern0pt}metis\ AK{\isachardot}{\kern0pt}simps\ K{\isacharunderscore}{\kern0pt}imp{\isacharunderscore}{\kern0pt}trans\ duality{\isacharunderscore}{\kern0pt}taut{\isacharparenright}{\kern0pt}\isanewline%
}%
\SNIP{lemma:K-thm-14}{%
\ \ \isacommand{then}\isamarkupfalse%
\ \isacommand{show}\isamarkupfalse%
\ {\isacharquery}{\kern0pt}thesis\isanewline%
}%
\SNIP{lemma:K-thm-15}{%
\ \ \ \ \isacommand{by}\isamarkupfalse%
\ {\isacharparenleft}{\kern0pt}simp\ add{\isacharcolon}{\kern0pt}\ K{\isacharunderscore}{\kern0pt}multi{\isacharunderscore}{\kern0pt}imply{\isacharparenright}{\kern0pt}\isanewline%
}%
\SNIP{lemma:K-thm-16}{%
\isacommand{qed}\isamarkupfalse%
\endisatagproof
{\isafoldproof}%
\isadelimproof%
}%
\SNIP{primrec:conjunct-0-0}{%
{\isacharprime}{\kern0pt}i\ fm\ list\ {\isasymRightarrow}\ {\isacharprime}{\kern0pt}i\ fm%
}%
\SNIP{primrec:conjunct-0}{%
\isacommand{primrec}\isamarkupfalse%
\ conjunct\ {\isacharcolon}{\kern0pt}{\isacharcolon}{\kern0pt}\ {\Cartouche{primrec:conjunct}{0}{0}}\ \isakeyword{where}\isanewline%
}%
\SNIP{primrec:conjunct-1-0}{%
conjunct\ {\isacharbrackleft}{\kern0pt}{\isacharbrackright}{\kern0pt}\ {\isacharequal}{\kern0pt}\ \isactrlbold {\isasymtop}%
}%
\SNIP{primrec:conjunct-1}{%
\ \ {\Cartouche{primrec:conjunct}{1}{0}}\isanewline%
}%
\SNIP{primrec:conjunct-2-0}{%
conjunct\ {\isacharparenleft}{\kern0pt}p{\isacharhash}{\kern0pt}ps{\isacharparenright}{\kern0pt}\ {\isacharequal}{\kern0pt}\ {\isacharparenleft}{\kern0pt}p\ \isactrlbold {\isasymand}\ conjunct\ ps{\isacharparenright}{\kern0pt}%
}%
\SNIP{primrec:conjunct-2}{%
{\isacharbar}{\kern0pt}\ {\Cartouche{primrec:conjunct}{2}{0}}%
}%
\SNIP{lemma:imply-conjunct-0-0}{%
tautology\ {\isacharparenleft}{\kern0pt}{\isacharparenleft}{\kern0pt}imply\ G\ p{\isacharparenright}{\kern0pt}\ \isactrlbold {\isasymlongrightarrow}\ {\isacharparenleft}{\kern0pt}{\isacharparenleft}{\kern0pt}conjunct\ G{\isacharparenright}{\kern0pt}\ \isactrlbold {\isasymlongrightarrow}\ p{\isacharparenright}{\kern0pt}{\isacharparenright}{\kern0pt}%
}%
\SNIP{lemma:imply-conjunct-0}{%
\isacommand{lemma}\isamarkupfalse%
\ imply{\isacharunderscore}{\kern0pt}conjunct{\isacharcolon}{\kern0pt}\ {\Cartouche{lemma:imply-conjunct}{0}{0}}\isanewline%
}%
\SNIP{lemma:imply-conjunct-1}{%
\isadelimproof
\ \ %
\endisadelimproof
\isatagproof
\isacommand{apply}\isamarkupfalse%
{\isacharparenleft}{\kern0pt}induction\ G{\isacharparenright}{\kern0pt}\ \isanewline%
}%
\SNIP{lemma:imply-conjunct-2}{%
\ \ \isacommand{apply}\isamarkupfalse%
\ simp\isanewline%
}%
\SNIP{lemma:imply-conjunct-3}{%
\ \ \isacommand{by}\isamarkupfalse%
\ force%
\endisatagproof
{\isafoldproof}%
\isadelimproof%
}%
\SNIP{lemma:conjunct-imply-0-0}{%
tautology\ {\isacharparenleft}{\kern0pt}{\isacharparenleft}{\kern0pt}{\isacharparenleft}{\kern0pt}conjunct\ G{\isacharparenright}{\kern0pt}\ \isactrlbold {\isasymlongrightarrow}\ p{\isacharparenright}{\kern0pt}\ \isactrlbold {\isasymlongrightarrow}{\isacharparenleft}{\kern0pt}imply\ G\ p{\isacharparenright}{\kern0pt}{\isacharparenright}{\kern0pt}%
}%
\SNIP{lemma:conjunct-imply-0}{%
\isacommand{lemma}\isamarkupfalse%
\ conjunct{\isacharunderscore}{\kern0pt}imply{\isacharcolon}{\kern0pt}\ {\Cartouche{lemma:conjunct-imply}{0}{0}}\isanewline%
}%
\SNIP{lemma:conjunct-imply-1}{%
\isadelimproof
\ \ %
\endisadelimproof
\isatagproof
\isacommand{by}\isamarkupfalse%
\ {\isacharparenleft}{\kern0pt}induct\ G{\isacharparenright}{\kern0pt}\ simp{\isacharunderscore}{\kern0pt}all%
\endisatagproof
{\isafoldproof}%
\isadelimproof%
}%
\SNIP{lemma:K-imply-conjunct-0}{%
\isacommand{lemma}\isamarkupfalse%
\ K{\isacharunderscore}{\kern0pt}imply{\isacharunderscore}{\kern0pt}conjunct{\isacharcolon}{\kern0pt}\isanewline%
}%
\SNIP{lemma:K-imply-conjunct-1-0}{%
A\ {\isasymturnstile}\ imply\ G\ p%
}%
\SNIP{lemma:K-imply-conjunct-1}{%
\ \ \isakeyword{assumes}\ {\Cartouche{lemma:K-imply-conjunct}{1}{0}}\isanewline%
}%
\SNIP{lemma:K-imply-conjunct-2-0}{%
A\ {\isasymturnstile}\ {\isacharparenleft}{\kern0pt}{\isacharparenleft}{\kern0pt}conjunct\ G{\isacharparenright}{\kern0pt}\ \isactrlbold {\isasymlongrightarrow}\ p{\isacharparenright}{\kern0pt}%
}%
\SNIP{lemma:K-imply-conjunct-2}{%
\ \ \isakeyword{shows}\ {\Cartouche{lemma:K-imply-conjunct}{2}{0}}\isanewline%
}%
\SNIP{lemma:K-imply-conjunct-3}{%
\isadelimproof
\ \ %
\endisadelimproof
\isatagproof
\isacommand{using}\isamarkupfalse%
\ A{\isadigit{1}}\ R{\isadigit{1}}\ assms\ imply{\isacharunderscore}{\kern0pt}conjunct\ \isacommand{by}\isamarkupfalse%
\ blast%
\endisatagproof
{\isafoldproof}%
\isadelimproof%
}%
\SNIP{lemma:K-conjunct-imply-0}{%
\isacommand{lemma}\isamarkupfalse%
\ K{\isacharunderscore}{\kern0pt}conjunct{\isacharunderscore}{\kern0pt}imply{\isacharcolon}{\kern0pt}\isanewline%
}%
\SNIP{lemma:K-conjunct-imply-1-0}{%
A\ {\isasymturnstile}\ {\isacharparenleft}{\kern0pt}{\isacharparenleft}{\kern0pt}conjunct\ G{\isacharparenright}{\kern0pt}\ \isactrlbold {\isasymlongrightarrow}\ p{\isacharparenright}{\kern0pt}%
}%
\SNIP{lemma:K-conjunct-imply-1}{%
\ \ \isakeyword{assumes}\ {\Cartouche{lemma:K-conjunct-imply}{1}{0}}\isanewline%
}%
\SNIP{lemma:K-conjunct-imply-2-0}{%
A\ {\isasymturnstile}\ imply\ G\ p%
}%
\SNIP{lemma:K-conjunct-imply-2}{%
\ \ \isakeyword{shows}\ {\Cartouche{lemma:K-conjunct-imply}{2}{0}}\isanewline%
}%
\SNIP{lemma:K-conjunct-imply-3}{%
\isadelimproof
\ \ %
\endisadelimproof
\isatagproof
\isacommand{using}\isamarkupfalse%
\ A{\isadigit{1}}\ R{\isadigit{1}}\ assms\ conjunct{\isacharunderscore}{\kern0pt}imply\ \isacommand{by}\isamarkupfalse%
\ blast%
\endisatagproof
{\isafoldproof}%
\isadelimproof%
}%
\SNIP{lemma:K-conj-imply-factor-0}{%
\isacommand{lemma}\isamarkupfalse%
\ K{\isacharunderscore}{\kern0pt}conj{\isacharunderscore}{\kern0pt}imply{\isacharunderscore}{\kern0pt}factor{\isacharcolon}{\kern0pt}\ \isanewline%
}%
\SNIP{lemma:K-conj-imply-factor-1-0}{%
{\isacharparenleft}{\kern0pt}{\isacharparenleft}{\kern0pt}{\isacharprime}{\kern0pt}i\ {\isacharcolon}{\kern0pt}{\isacharcolon}{\kern0pt}\ countable{\isacharparenright}{\kern0pt}\ fm\ {\isasymRightarrow}\ bool{\isacharparenright}{\kern0pt}%
}%
\SNIP{lemma:K-conj-imply-factor-1}{%
\ \ \isakeyword{fixes}\ A\ {\isacharcolon}{\kern0pt}{\isacharcolon}{\kern0pt}\ {\Cartouche{lemma:K-conj-imply-factor}{1}{0}}\ \ \isanewline%
}%
\SNIP{lemma:K-conj-imply-factor-2-0}{%
A\ {\isasymturnstile}\ {\isacharparenleft}{\kern0pt}{\isacharparenleft}{\kern0pt}{\isacharparenleft}{\kern0pt}{\isacharparenleft}{\kern0pt}K\ i\ p{\isacharparenright}{\kern0pt}\ \isactrlbold {\isasymand}\ {\isacharparenleft}{\kern0pt}K\ i\ q{\isacharparenright}{\kern0pt}{\isacharparenright}{\kern0pt}\ \isactrlbold {\isasymlongrightarrow}\ r{\isacharparenright}{\kern0pt}\ \isactrlbold {\isasymlongrightarrow}{\isacharparenleft}{\kern0pt}{\isacharparenleft}{\kern0pt}K\ i\ {\isacharparenleft}{\kern0pt}p\ \isactrlbold {\isasymand}\ q{\isacharparenright}{\kern0pt}{\isacharparenright}{\kern0pt}\ \isactrlbold {\isasymlongrightarrow}\ r{\isacharparenright}{\kern0pt}{\isacharparenright}{\kern0pt}%
}%
\SNIP{lemma:K-conj-imply-factor-2}{%
\ \ \isakeyword{shows}\ {\Cartouche{lemma:K-conj-imply-factor}{2}{0}}\isanewline%
}%
\SNIP{lemma:K-conj-imply-factor-3}{%
\isadelimproof
\endisadelimproof
\isatagproof
\isacommand{proof}\isamarkupfalse%
\ {\isacharminus}{\kern0pt}\isanewline%
}%
\SNIP{lemma:K-conj-imply-factor-4-0}{%
A\ {\isasymturnstile}\ {\isacharparenleft}{\kern0pt}{\isacharparenleft}{\kern0pt}K\ i\ {\isacharparenleft}{\kern0pt}p\ \isactrlbold {\isasymand}\ q{\isacharparenright}{\kern0pt}{\isacharparenright}{\kern0pt}\ \isactrlbold {\isasymlongrightarrow}\ {\isacharparenleft}{\kern0pt}{\isacharparenleft}{\kern0pt}K\ i\ p{\isacharparenright}{\kern0pt}\ \isactrlbold {\isasymand}\ {\isacharparenleft}{\kern0pt}K\ i\ q{\isacharparenright}{\kern0pt}{\isacharparenright}{\kern0pt}{\isacharparenright}{\kern0pt}%
}%
\SNIP{lemma:K-conj-imply-factor-4}{%
\ \ \isacommand{have}\isamarkupfalse%
\ {\isacharasterisk}{\kern0pt}{\isacharcolon}{\kern0pt}\ {\Cartouche{lemma:K-conj-imply-factor}{4}{0}}\isanewline%
}%
\SNIP{lemma:K-conj-imply-factor-5}{%
\ \ \isacommand{proof}\isamarkupfalse%
\ {\isacharparenleft}{\kern0pt}rule\ ccontr{\isacharparenright}{\kern0pt}\ \isanewline%
}%
\SNIP{lemma:K-conj-imply-factor-6-0}{%
{\isasymnot}\ A\ {\isasymturnstile}\ {\isacharparenleft}{\kern0pt}{\isacharparenleft}{\kern0pt}K\ i\ {\isacharparenleft}{\kern0pt}p\ \isactrlbold {\isasymand}\ q{\isacharparenright}{\kern0pt}{\isacharparenright}{\kern0pt}\ \isactrlbold {\isasymlongrightarrow}\ {\isacharparenleft}{\kern0pt}{\isacharparenleft}{\kern0pt}K\ i\ p{\isacharparenright}{\kern0pt}\ \isactrlbold {\isasymand}\ {\isacharparenleft}{\kern0pt}K\ i\ q{\isacharparenright}{\kern0pt}{\isacharparenright}{\kern0pt}{\isacharparenright}{\kern0pt}%
}%
\SNIP{lemma:K-conj-imply-factor-6}{%
\ \ \ \ \isacommand{assume}\isamarkupfalse%
\ {\Cartouche{lemma:K-conj-imply-factor}{6}{0}}\isanewline%
}%
\SNIP{lemma:K-conj-imply-factor-7-0}{%
consistent\ A\ {\isacharbraceleft}{\kern0pt}\isactrlbold {\isasymnot}{\isacharparenleft}{\kern0pt}{\isacharparenleft}{\kern0pt}K\ i\ {\isacharparenleft}{\kern0pt}p\ \isactrlbold {\isasymand}\ q{\isacharparenright}{\kern0pt}{\isacharparenright}{\kern0pt}\ \isactrlbold {\isasymlongrightarrow}\ {\isacharparenleft}{\kern0pt}{\isacharparenleft}{\kern0pt}K\ i\ p{\isacharparenright}{\kern0pt}\ \isactrlbold {\isasymand}\ {\isacharparenleft}{\kern0pt}K\ i\ q{\isacharparenright}{\kern0pt}{\isacharparenright}{\kern0pt}{\isacharparenright}{\kern0pt}{\isacharbraceright}{\kern0pt}%
}%
\SNIP{lemma:K-conj-imply-factor-7}{%
\ \ \ \ \isacommand{then}\isamarkupfalse%
\ \isacommand{have}\isamarkupfalse%
\ {\Cartouche{lemma:K-conj-imply-factor}{7}{0}}\isanewline%
}%
\SNIP{lemma:K-conj-imply-factor-8}{%
\ \ \ \ \ \ \isacommand{by}\isamarkupfalse%
\ {\isacharparenleft}{\kern0pt}metis\ imply{\isachardot}{\kern0pt}simps{\isacharparenleft}{\kern0pt}{\isadigit{1}}{\isacharparenright}{\kern0pt}\ inconsistent{\isacharunderscore}{\kern0pt}imply\ insert{\isacharunderscore}{\kern0pt}is{\isacharunderscore}{\kern0pt}Un\ list{\isachardot}{\kern0pt}set{\isacharparenleft}{\kern0pt}{\isadigit{1}}{\isacharparenright}{\kern0pt}{\isacharparenright}{\kern0pt}\isanewline%
}%
\SNIP{lemma:K-conj-imply-factor-9-0}{%
Extend\ A\ {\isacharbraceleft}{\kern0pt}\isactrlbold {\isasymnot}{\isacharparenleft}{\kern0pt}{\isacharparenleft}{\kern0pt}K\ i\ {\isacharparenleft}{\kern0pt}p\ \isactrlbold {\isasymand}\ q{\isacharparenright}{\kern0pt}{\isacharparenright}{\kern0pt}\ \isactrlbold {\isasymlongrightarrow}\ {\isacharparenleft}{\kern0pt}{\isacharparenleft}{\kern0pt}K\ i\ p{\isacharparenright}{\kern0pt}\ \isactrlbold {\isasymand}\ {\isacharparenleft}{\kern0pt}K\ i\ q{\isacharparenright}{\kern0pt}{\isacharparenright}{\kern0pt}{\isacharparenright}{\kern0pt}{\isacharbraceright}{\kern0pt}\ from{\isacharunderscore}{\kern0pt}nat%
}%
\SNIP{lemma:K-conj-imply-factor-9}{%
\ \ \ \ \isacommand{let}\isamarkupfalse%
\ {\isacharquery}{\kern0pt}V\ {\isacharequal}{\kern0pt}\ {\Cartouche{lemma:K-conj-imply-factor}{9}{0}}\ \isanewline%
}%
\SNIP{lemma:K-conj-imply-factor-10-0}{%
Kripke\ {\isacharparenleft}{\kern0pt}mcss\ A{\isacharparenright}{\kern0pt}\ pi\ {\isacharparenleft}{\kern0pt}reach\ A{\isacharparenright}{\kern0pt}%
}%
\SNIP{lemma:K-conj-imply-factor-10}{%
\ \ \ \ \isacommand{let}\isamarkupfalse%
\ {\isacharquery}{\kern0pt}M\ {\isacharequal}{\kern0pt}\ {\Cartouche{lemma:K-conj-imply-factor}{10}{0}}\ \isanewline%
}%
\SNIP{lemma:K-conj-imply-factor-11-0}{%
{\isacharquery}{\kern0pt}V\ {\isasymin}\ {\isasymW}\ {\isacharquery}{\kern0pt}M\ {\isasymand}\ {\isacharquery}{\kern0pt}M{\isacharcomma}{\kern0pt}\ {\isacharquery}{\kern0pt}V\ {\isasymTurnstile}\ \isactrlbold {\isasymnot}{\isacharparenleft}{\kern0pt}{\isacharparenleft}{\kern0pt}K\ i\ {\isacharparenleft}{\kern0pt}p\ \isactrlbold {\isasymand}\ q{\isacharparenright}{\kern0pt}{\isacharparenright}{\kern0pt}\ \isactrlbold {\isasymlongrightarrow}\ {\isacharparenleft}{\kern0pt}{\isacharparenleft}{\kern0pt}K\ i\ p{\isacharparenright}{\kern0pt}\ \isactrlbold {\isasymand}\ {\isacharparenleft}{\kern0pt}K\ i\ q{\isacharparenright}{\kern0pt}{\isacharparenright}{\kern0pt}{\isacharparenright}{\kern0pt}%
}%
\SNIP{lemma:K-conj-imply-factor-11}{%
\ \ \ \ \isacommand{have}\isamarkupfalse%
\ {\Cartouche{lemma:K-conj-imply-factor}{11}{0}}\isanewline%
}%
\SNIP{lemma:K-conj-imply-factor-12-0}{%
consistent\ A\ {\isacharbraceleft}{\kern0pt}\isactrlbold {\isasymnot}\ {\isacharparenleft}{\kern0pt}K\ i\ {\isacharparenleft}{\kern0pt}p\ \isactrlbold {\isasymand}\ q{\isacharparenright}{\kern0pt}\ \isactrlbold {\isasymlongrightarrow}\ K\ i\ p\ \isactrlbold {\isasymand}\ K\ i\ q{\isacharparenright}{\kern0pt}{\isacharbraceright}{\kern0pt}%
}%
\SNIP{lemma:K-conj-imply-factor-12}{%
\ \ \ \ \ \ \isacommand{using}\isamarkupfalse%
\ canonical{\isacharunderscore}{\kern0pt}model\ {\Cartouche{lemma:K-conj-imply-factor}{12}{0}}\ \isanewline%
}%
\SNIP{lemma:K-conj-imply-factor-13}{%
\ \ \ \ \ \ \ \ \ \ \ \ insert{\isacharunderscore}{\kern0pt}iff\ kripke{\isachardot}{\kern0pt}sel{\isacharparenleft}{\kern0pt}{\isadigit{1}}{\isacharparenright}{\kern0pt}\ mem{\isacharunderscore}{\kern0pt}Collect{\isacharunderscore}{\kern0pt}eq\ \isacommand{by}\isamarkupfalse%
\ fastforce\isanewline%
}%
\SNIP{lemma:K-conj-imply-factor-14-0}{%
{\isacharquery}{\kern0pt}M{\isacharcomma}{\kern0pt}\ {\isacharquery}{\kern0pt}V\ {\isasymTurnstile}\ {\isacharparenleft}{\kern0pt}{\isacharparenleft}{\kern0pt}K\ i\ {\isacharparenleft}{\kern0pt}p\ \isactrlbold {\isasymand}\ q{\isacharparenright}{\kern0pt}{\isacharparenright}{\kern0pt}\ \isactrlbold {\isasymand}\ \isactrlbold {\isasymnot}{\isacharparenleft}{\kern0pt}{\isacharparenleft}{\kern0pt}K\ i\ p{\isacharparenright}{\kern0pt}\ \isactrlbold {\isasymand}\ {\isacharparenleft}{\kern0pt}K\ i\ q{\isacharparenright}{\kern0pt}{\isacharparenright}{\kern0pt}{\isacharparenright}{\kern0pt}%
}%
\SNIP{lemma:K-conj-imply-factor-14}{%
\ \ \ \ \isacommand{then}\isamarkupfalse%
\ \isacommand{have}\isamarkupfalse%
\ o{\isacharcolon}{\kern0pt}\ {\Cartouche{lemma:K-conj-imply-factor}{14}{0}}\isanewline%
}%
\SNIP{lemma:K-conj-imply-factor-15}{%
\ \ \ \ \ \ \isacommand{by}\isamarkupfalse%
\ auto\isanewline%
}%
\SNIP{lemma:K-conj-imply-factor-16-0}{%
{\isacharquery}{\kern0pt}M{\isacharcomma}{\kern0pt}\ {\isacharquery}{\kern0pt}V\ {\isasymTurnstile}\ {\isacharparenleft}{\kern0pt}K\ i\ {\isacharparenleft}{\kern0pt}p\ \isactrlbold {\isasymand}\ q{\isacharparenright}{\kern0pt}{\isacharparenright}{\kern0pt}%
}%
\SNIP{lemma:K-conj-imply-factor-16-1}{%
{\isacharparenleft}{\kern0pt}{\isacharquery}{\kern0pt}M{\isacharcomma}{\kern0pt}\ {\isacharquery}{\kern0pt}V\ {\isasymTurnstile}\ \isactrlbold {\isasymnot}{\isacharparenleft}{\kern0pt}K\ i\ p{\isacharparenright}{\kern0pt}{\isacharparenright}{\kern0pt}\ {\isasymor}\ {\isacharparenleft}{\kern0pt}{\isacharquery}{\kern0pt}M{\isacharcomma}{\kern0pt}\ {\isacharquery}{\kern0pt}V\ {\isasymTurnstile}\ \isactrlbold {\isasymnot}{\isacharparenleft}{\kern0pt}K\ i\ q{\isacharparenright}{\kern0pt}{\isacharparenright}{\kern0pt}%
}%
\SNIP{lemma:K-conj-imply-factor-16}{%
\ \ \ \ \isacommand{then}\isamarkupfalse%
\ \isacommand{have}\isamarkupfalse%
\ {\Cartouche{lemma:K-conj-imply-factor}{16}{0}}\ {\Cartouche{lemma:K-conj-imply-factor}{16}{1}}\isanewline%
}%
\SNIP{lemma:K-conj-imply-factor-17}{%
\ \ \ \ \ \ \isacommand{by}\isamarkupfalse%
\ auto\isanewline%
}%
\SNIP{lemma:K-conj-imply-factor-18-0}{%
{\isasymforall}\ U\ {\isasymin}\ {\isasymW}\ {\isacharquery}{\kern0pt}M\ {\isasyminter}\ {\isasymK}\ {\isacharquery}{\kern0pt}M\ i\ {\isacharquery}{\kern0pt}V{\isachardot}{\kern0pt}\ {\isacharquery}{\kern0pt}M{\isacharcomma}{\kern0pt}\ U\ {\isasymTurnstile}{\isacharparenleft}{\kern0pt}p\ \isactrlbold {\isasymand}\ q{\isacharparenright}{\kern0pt}%
}%
\SNIP{lemma:K-conj-imply-factor-18}{%
\ \ \ \ \isacommand{then}\isamarkupfalse%
\ \isacommand{have}\isamarkupfalse%
\ {\Cartouche{lemma:K-conj-imply-factor}{18}{0}}\isanewline%
}%
\SNIP{lemma:K-conj-imply-factor-19-0}{%
{\isasymexists}\ U\ {\isasymin}\ {\isasymW}\ {\isacharquery}{\kern0pt}M\ {\isasyminter}\ {\isasymK}\ {\isacharquery}{\kern0pt}M\ i\ {\isacharquery}{\kern0pt}V{\isachardot}{\kern0pt}\ {\isacharquery}{\kern0pt}M{\isacharcomma}{\kern0pt}\ U\ {\isasymTurnstile}\ {\isacharparenleft}{\kern0pt}{\isacharparenleft}{\kern0pt}\isactrlbold {\isasymnot}p{\isacharparenright}{\kern0pt}\ \isactrlbold {\isasymor}\ {\isacharparenleft}{\kern0pt}\isactrlbold {\isasymnot}q{\isacharparenright}{\kern0pt}{\isacharparenright}{\kern0pt}%
}%
\SNIP{lemma:K-conj-imply-factor-19}{%
\ \ \ \ \ \ \ \ \ \ \ \ \ \ {\Cartouche{lemma:K-conj-imply-factor}{19}{0}}\isanewline%
}%
\SNIP{lemma:K-conj-imply-factor-20}{%
\ \ \ \ \ \ \isacommand{apply}\isamarkupfalse%
\ {\isacharparenleft}{\kern0pt}meson\ semantics{\isachardot}{\kern0pt}simps{\isacharparenleft}{\kern0pt}{\isadigit{6}}{\isacharparenright}{\kern0pt}{\isacharparenright}{\kern0pt}\isanewline%
}%
\SNIP{lemma:K-conj-imply-factor-21}{%
\ \ \ \ \ \ \isacommand{using}\isamarkupfalse%
\ o\ \isacommand{by}\isamarkupfalse%
\ auto\isanewline%
}%
\SNIP{lemma:K-conj-imply-factor-22}{%
\ \ \ \ \isacommand{then}\isamarkupfalse%
\ \isacommand{show}\isamarkupfalse%
\ False\ \isanewline%
}%
\SNIP{lemma:K-conj-imply-factor-23}{%
\ \ \ \ \ \ \isacommand{by}\isamarkupfalse%
\ simp\isanewline%
}%
\SNIP{lemma:K-conj-imply-factor-24}{%
\ \ \isacommand{qed}\isamarkupfalse%
\isanewline%
}%
\SNIP{lemma:K-conj-imply-factor-25-0}{%
A\ {\isasymturnstile}\ {\isacharparenleft}{\kern0pt}{\isacharparenleft}{\kern0pt}{\isacharparenleft}{\kern0pt}K\ i\ {\isacharparenleft}{\kern0pt}p\ \isactrlbold {\isasymand}\ q{\isacharparenright}{\kern0pt}{\isacharparenright}{\kern0pt}\ \isactrlbold {\isasymlongrightarrow}\ {\isacharparenleft}{\kern0pt}{\isacharparenleft}{\kern0pt}K\ i\ p{\isacharparenright}{\kern0pt}\ \isactrlbold {\isasymand}\ {\isacharparenleft}{\kern0pt}K\ i\ q{\isacharparenright}{\kern0pt}{\isacharparenright}{\kern0pt}{\isacharparenright}{\kern0pt}\ \isactrlbold {\isasymlongrightarrow}\isanewline\ \ \ \ \ \ \ \ \ \ \ \ \ \ \ \ \ \ \ {\isacharparenleft}{\kern0pt}{\isacharparenleft}{\kern0pt}{\isacharparenleft}{\kern0pt}{\isacharparenleft}{\kern0pt}K\ i\ p{\isacharparenright}{\kern0pt}\ \isactrlbold {\isasymand}\ {\isacharparenleft}{\kern0pt}K\ i\ q{\isacharparenright}{\kern0pt}{\isacharparenright}{\kern0pt}\ \isactrlbold {\isasymlongrightarrow}\ r{\isacharparenright}{\kern0pt}\ \isactrlbold {\isasymlongrightarrow}\ {\isacharparenleft}{\kern0pt}{\isacharparenleft}{\kern0pt}K\ i\ {\isacharparenleft}{\kern0pt}p\ \isactrlbold {\isasymand}\ q{\isacharparenright}{\kern0pt}{\isacharparenright}{\kern0pt}\ \isactrlbold {\isasymlongrightarrow}\ r{\isacharparenright}{\kern0pt}{\isacharparenright}{\kern0pt}{\isacharparenright}{\kern0pt}%
}%
\SNIP{lemma:K-conj-imply-factor-25}{%
\ \ \isacommand{then}\isamarkupfalse%
\ \isacommand{have}\isamarkupfalse%
\ {\isacartoucheopen}A\ {\isasymturnstile}\ {\isacharparenleft}{\kern0pt}{\isacharparenleft}{\kern0pt}{\isacharparenleft}{\kern0pt}K\ i\ {\isacharparenleft}{\kern0pt}p\ \isactrlbold {\isasymand}\ q{\isacharparenright}{\kern0pt}{\isacharparenright}{\kern0pt}\ \isactrlbold {\isasymlongrightarrow}\ {\isacharparenleft}{\kern0pt}{\isacharparenleft}{\kern0pt}K\ i\ p{\isacharparenright}{\kern0pt}\ \isactrlbold {\isasymand}\ {\isacharparenleft}{\kern0pt}K\ i\ q{\isacharparenright}{\kern0pt}{\isacharparenright}{\kern0pt}{\isacharparenright}{\kern0pt}\ \isactrlbold {\isasymlongrightarrow}\isanewline%
}%
\SNIP{lemma:K-conj-imply-factor-26}{%
\ \ \ \ \ \ \ \ \ \ \ \ \ \ \ \ \ \ \ {\isacharparenleft}{\kern0pt}{\isacharparenleft}{\kern0pt}{\isacharparenleft}{\kern0pt}{\isacharparenleft}{\kern0pt}K\ i\ p{\isacharparenright}{\kern0pt}\ \isactrlbold {\isasymand}\ {\isacharparenleft}{\kern0pt}K\ i\ q{\isacharparenright}{\kern0pt}{\isacharparenright}{\kern0pt}\ \isactrlbold {\isasymlongrightarrow}\ r{\isacharparenright}{\kern0pt}\ \isactrlbold {\isasymlongrightarrow}\ {\isacharparenleft}{\kern0pt}{\isacharparenleft}{\kern0pt}K\ i\ {\isacharparenleft}{\kern0pt}p\ \isactrlbold {\isasymand}\ q{\isacharparenright}{\kern0pt}{\isacharparenright}{\kern0pt}\ \isactrlbold {\isasymlongrightarrow}\ r{\isacharparenright}{\kern0pt}{\isacharparenright}{\kern0pt}{\isacharparenright}{\kern0pt}{\isacartoucheclose}\isanewline%
}%
\SNIP{lemma:K-conj-imply-factor-27}{%
\ \ \ \ \isacommand{by}\isamarkupfalse%
\ {\isacharparenleft}{\kern0pt}simp\ add{\isacharcolon}{\kern0pt}\ A{\isadigit{1}}{\isacharparenright}{\kern0pt}\isanewline%
}%
\SNIP{lemma:K-conj-imply-factor-28}{%
\ \ \isacommand{then}\isamarkupfalse%
\ \isacommand{show}\isamarkupfalse%
\ {\isacharquery}{\kern0pt}thesis\ \isanewline%
}%
\SNIP{lemma:K-conj-imply-factor-29}{%
\ \ \ \ \isacommand{using}\isamarkupfalse%
\ {\isachardoublequoteopen}{\isacharasterisk}{\kern0pt}{\isachardoublequoteclose}\ R{\isadigit{1}}\ \isacommand{by}\isamarkupfalse%
\ blast\isanewline%
}%
\SNIP{lemma:K-conj-imply-factor-30}{%
\isacommand{qed}\isamarkupfalse%
\endisatagproof
{\isafoldproof}%
\isadelimproof%
}%
\SNIP{lemma:K-conjunction-in-0-0}{%
A\ {\isasymturnstile}\ {\isacharparenleft}{\kern0pt}K\ i\ {\isacharparenleft}{\kern0pt}p\ \isactrlbold {\isasymand}\ q{\isacharparenright}{\kern0pt}\ \isactrlbold {\isasymlongrightarrow}\ {\isacharparenleft}{\kern0pt}{\isacharparenleft}{\kern0pt}K\ i\ p{\isacharparenright}{\kern0pt}\ \isactrlbold {\isasymand}\ K\ i\ q{\isacharparenright}{\kern0pt}{\isacharparenright}{\kern0pt}%
}%
\SNIP{lemma:K-conjunction-in-0}{%
\isacommand{lemma}\isamarkupfalse%
\ K{\isacharunderscore}{\kern0pt}conjunction{\isacharunderscore}{\kern0pt}in{\isacharcolon}{\kern0pt}\ {\Cartouche{lemma:K-conjunction-in}{0}{0}}\isanewline%
}%
\SNIP{lemma:K-conjunction-in-1}{%
\isadelimproof
\endisadelimproof
\isatagproof
\isacommand{proof}\isamarkupfalse%
\ {\isacharminus}{\kern0pt}\isanewline%
}%
\SNIP{lemma:K-conjunction-in-2-0}{%
A\ {\isasymturnstile}\ {\isacharparenleft}{\kern0pt}{\isacharparenleft}{\kern0pt}p\isactrlbold {\isasymand}q{\isacharparenright}{\kern0pt}\ \isactrlbold {\isasymlongrightarrow}\ p{\isacharparenright}{\kern0pt}%
}%
\SNIP{lemma:K-conjunction-in-2-1}{%
A\ {\isasymturnstile}\ {\isacharparenleft}{\kern0pt}{\isacharparenleft}{\kern0pt}p\isactrlbold {\isasymand}q{\isacharparenright}{\kern0pt}\ \isactrlbold {\isasymlongrightarrow}\ q{\isacharparenright}{\kern0pt}%
}%
\SNIP{lemma:K-conjunction-in-2}{%
\ \ \isacommand{have}\isamarkupfalse%
\ o{\isadigit{1}}{\isacharcolon}{\kern0pt}\ {\Cartouche{lemma:K-conjunction-in}{2}{0}}\ \isakeyword{and}\ o{\isadigit{2}}{\isacharcolon}{\kern0pt}\ {\Cartouche{lemma:K-conjunction-in}{2}{1}}\isanewline%
}%
\SNIP{lemma:K-conjunction-in-3}{%
\ \ \ \ \isacommand{apply}\isamarkupfalse%
{\isacharparenleft}{\kern0pt}simp\ add{\isacharcolon}{\kern0pt}\ A{\isadigit{1}}{\isacharparenright}{\kern0pt}\isanewline%
}%
\SNIP{lemma:K-conjunction-in-4}{%
\ \ \ \ \isacommand{by}\isamarkupfalse%
\ {\isacharparenleft}{\kern0pt}simp\ add{\isacharcolon}{\kern0pt}\ A{\isadigit{1}}{\isacharparenright}{\kern0pt}\isanewline%
}%
\SNIP{lemma:K-conjunction-in-5-0}{%
A\ {\isasymturnstile}\ {\isacharparenleft}{\kern0pt}K\ i\ {\isacharparenleft}{\kern0pt}p\ \isactrlbold {\isasymand}\ q{\isacharparenright}{\kern0pt}\ \isactrlbold {\isasymlongrightarrow}\ K\ i\ q{\isacharparenright}{\kern0pt}%
}%
\SNIP{lemma:K-conjunction-in-5-1}{%
A\ {\isasymturnstile}\ {\isacharparenleft}{\kern0pt}K\ i\ {\isacharparenleft}{\kern0pt}p\ \isactrlbold {\isasymand}\ q{\isacharparenright}{\kern0pt}\ \isactrlbold {\isasymlongrightarrow}\ K\ i\ p{\isacharparenright}{\kern0pt}%
}%
\SNIP{lemma:K-conjunction-in-5}{%
\ \ \isacommand{then}\isamarkupfalse%
\ \isacommand{have}\isamarkupfalse%
\ c{\isadigit{1}}{\isacharcolon}{\kern0pt}\ {\Cartouche{lemma:K-conjunction-in}{5}{0}}\ \isakeyword{and}\ c{\isadigit{2}}{\isacharcolon}{\kern0pt}\ {\Cartouche{lemma:K-conjunction-in}{5}{1}}\ \isanewline%
}%
\SNIP{lemma:K-conjunction-in-6}{%
\ \ \ \ \isacommand{apply}\isamarkupfalse%
\ {\isacharparenleft}{\kern0pt}simp\ add{\isacharcolon}{\kern0pt}\ K{\isacharunderscore}{\kern0pt}map\ o{\isadigit{2}}{\isacharparenright}{\kern0pt}\isanewline%
}%
\SNIP{lemma:K-conjunction-in-7}{%
\ \ \ \ \isacommand{by}\isamarkupfalse%
\ {\isacharparenleft}{\kern0pt}simp\ add{\isacharcolon}{\kern0pt}\ K{\isacharunderscore}{\kern0pt}map\ o{\isadigit{1}}{\isacharparenright}{\kern0pt}\isanewline%
}%
\SNIP{lemma:K-conjunction-in-8}{%
\ \ \isacommand{then}\isamarkupfalse%
\ \isacommand{show}\isamarkupfalse%
\ {\isacharquery}{\kern0pt}thesis\isanewline%
}%
\SNIP{lemma:K-conjunction-in-9}{%
\ \ \ \ \isacommand{by}\isamarkupfalse%
\ {\isacharparenleft}{\kern0pt}simp\ add{\isacharcolon}{\kern0pt}\ K{\isacharunderscore}{\kern0pt}imply{\isacharunderscore}{\kern0pt}multi{\isacharparenright}{\kern0pt}\isanewline%
}%
\SNIP{lemma:K-conjunction-in-10}{%
\isacommand{qed}\isamarkupfalse%
\endisatagproof
{\isafoldproof}%
\isadelimproof%
}%
\SNIP{lemma:K-conjunction-in-mult-0-0}{%
A\ {\isasymturnstile}\ {\isacharparenleft}{\kern0pt}{\isacharparenleft}{\kern0pt}K\ i\ {\isacharparenleft}{\kern0pt}conjunct\ G{\isacharparenright}{\kern0pt}{\isacharparenright}{\kern0pt}\ \isactrlbold {\isasymlongrightarrow}\ conjunct\ {\isacharparenleft}{\kern0pt}map\ {\isacharparenleft}{\kern0pt}K\ i{\isacharparenright}{\kern0pt}\ G{\isacharparenright}{\kern0pt}{\isacharparenright}{\kern0pt}%
}%
\SNIP{lemma:K-conjunction-in-mult-0}{%
\isacommand{lemma}\isamarkupfalse%
\ K{\isacharunderscore}{\kern0pt}conjunction{\isacharunderscore}{\kern0pt}in{\isacharunderscore}{\kern0pt}mult{\isacharcolon}{\kern0pt}\ \ {\Cartouche{lemma:K-conjunction-in-mult}{0}{0}}\isanewline%
}%
\SNIP{lemma:K-conjunction-in-mult-1}{%
\isadelimproof
\endisadelimproof
\isatagproof
\isacommand{proof}\isamarkupfalse%
\ {\isacharparenleft}{\kern0pt}induct\ G{\isacharparenright}{\kern0pt}\isanewline%
}%
\SNIP{lemma:K-conjunction-in-mult-2}{%
\ \ \isacommand{case}\isamarkupfalse%
\ Nil\isanewline%
}%
\SNIP{lemma:K-conjunction-in-mult-3}{%
\ \ \isacommand{then}\isamarkupfalse%
\ \isacommand{show}\isamarkupfalse%
\ {\isacharquery}{\kern0pt}case\isanewline%
}%
\SNIP{lemma:K-conjunction-in-mult-4}{%
\ \ \ \ \isacommand{by}\isamarkupfalse%
\ {\isacharparenleft}{\kern0pt}simp\ add{\isacharcolon}{\kern0pt}\ A{\isadigit{1}}{\isacharparenright}{\kern0pt}\isanewline%
}%
\SNIP{lemma:K-conjunction-in-mult-5}{%
\ \ \isacommand{case}\isamarkupfalse%
\ {\isacharparenleft}{\kern0pt}Cons\ a\ G{\isacharparenright}{\kern0pt}\isanewline%
}%
\SNIP{lemma:K-conjunction-in-mult-6-0}{%
A\ {\isasymturnstile}\ {\isacharparenleft}{\kern0pt}{\isacharparenleft}{\kern0pt}K\ i\ {\isacharparenleft}{\kern0pt}conjunct\ {\isacharparenleft}{\kern0pt}a{\isacharhash}{\kern0pt}G{\isacharparenright}{\kern0pt}{\isacharparenright}{\kern0pt}{\isacharparenright}{\kern0pt}\ \isactrlbold {\isasymlongrightarrow}\ {\isacharparenleft}{\kern0pt}K\ i\ {\isacharparenleft}{\kern0pt}a\ \isactrlbold {\isasymand}\ conjunct\ G{\isacharparenright}{\kern0pt}{\isacharparenright}{\kern0pt}{\isacharparenright}{\kern0pt}%
}%
\SNIP{lemma:K-conjunction-in-mult-6}{%
\ \ \isacommand{then}\isamarkupfalse%
\ \isacommand{have}\isamarkupfalse%
\ {\Cartouche{lemma:K-conjunction-in-mult}{6}{0}}\isanewline%
}%
\SNIP{lemma:K-conjunction-in-mult-7-0}{%
A\ {\isasymturnstile}\ {\isacharparenleft}{\kern0pt}{\isacharparenleft}{\kern0pt}K\ i\ {\isacharparenleft}{\kern0pt}a\ \isactrlbold {\isasymand}\ conjunct\ G{\isacharparenright}{\kern0pt}{\isacharparenright}{\kern0pt}\ \isactrlbold {\isasymlongrightarrow}\ {\isacharparenleft}{\kern0pt}{\isacharparenleft}{\kern0pt}K\ i\ a{\isacharparenright}{\kern0pt}\ \isactrlbold {\isasymand}\ K\ i\ {\isacharparenleft}{\kern0pt}conjunct\ G{\isacharparenright}{\kern0pt}{\isacharparenright}{\kern0pt}{\isacharparenright}{\kern0pt}%
}%
\SNIP{lemma:K-conjunction-in-mult-7}{%
\ \ \ \ \ \ \ \ \isakeyword{and}\ \ {\Cartouche{lemma:K-conjunction-in-mult}{7}{0}}\isanewline%
}%
\SNIP{lemma:K-conjunction-in-mult-8}{%
\ \ \ \ \isacommand{apply}\isamarkupfalse%
\ {\isacharparenleft}{\kern0pt}simp\ add{\isacharcolon}{\kern0pt}\ A{\isadigit{1}}{\isacharparenright}{\kern0pt}\isanewline%
}%
\SNIP{lemma:K-conjunction-in-mult-9}{%
\ \ \ \ \isacommand{by}\isamarkupfalse%
\ {\isacharparenleft}{\kern0pt}metis\ K{\isacharunderscore}{\kern0pt}conjunction{\isacharunderscore}{\kern0pt}in{\isacharparenright}{\kern0pt}\isanewline%
}%
\SNIP{lemma:K-conjunction-in-mult-10-0}{%
A\ {\isasymturnstile}\ {\isacharparenleft}{\kern0pt}{\isacharparenleft}{\kern0pt}K\ i\ {\isacharparenleft}{\kern0pt}conjunct\ {\isacharparenleft}{\kern0pt}a{\isacharhash}{\kern0pt}G{\isacharparenright}{\kern0pt}{\isacharparenright}{\kern0pt}{\isacharparenright}{\kern0pt}\ \isactrlbold {\isasymlongrightarrow}\ {\isacharparenleft}{\kern0pt}{\isacharparenleft}{\kern0pt}K\ i\ a{\isacharparenright}{\kern0pt}\ \isactrlbold {\isasymand}\ K\ i\ {\isacharparenleft}{\kern0pt}conjunct\ G{\isacharparenright}{\kern0pt}{\isacharparenright}{\kern0pt}{\isacharparenright}{\kern0pt}%
}%
\SNIP{lemma:K-conjunction-in-mult-10}{%
\ \ \isacommand{then}\isamarkupfalse%
\ \isacommand{have}\isamarkupfalse%
\ o{\isadigit{1}}{\isacharcolon}{\kern0pt}\ {\Cartouche{lemma:K-conjunction-in-mult}{10}{0}}\isanewline%
}%
\SNIP{lemma:K-conjunction-in-mult-11}{%
\ \ \ \ \isacommand{using}\isamarkupfalse%
\ K{\isacharunderscore}{\kern0pt}imp{\isacharunderscore}{\kern0pt}trans\ \isacommand{by}\isamarkupfalse%
\ blast\isanewline%
}%
\SNIP{lemma:K-conjunction-in-mult-12-0}{%
A\ {\isasymturnstile}\ {\isacharparenleft}{\kern0pt}{\isacharparenleft}{\kern0pt}{\isacharparenleft}{\kern0pt}K\ i\ a{\isacharparenright}{\kern0pt}\ \isactrlbold {\isasymand}\ K\ i\ {\isacharparenleft}{\kern0pt}conjunct\ G{\isacharparenright}{\kern0pt}{\isacharparenright}{\kern0pt}\ \isactrlbold {\isasymlongrightarrow}\ K\ i\ a{\isacharparenright}{\kern0pt}%
}%
\SNIP{lemma:K-conjunction-in-mult-12}{%
\ \ \isacommand{then}\isamarkupfalse%
\ \isacommand{have}\isamarkupfalse%
\ {\Cartouche{lemma:K-conjunction-in-mult}{12}{0}}\ \isanewline%
}%
\SNIP{lemma:K-conjunction-in-mult-13-0}{%
A\ {\isasymturnstile}\ {\isacharparenleft}{\kern0pt}{\isacharparenleft}{\kern0pt}{\isacharparenleft}{\kern0pt}K\ i\ a{\isacharparenright}{\kern0pt}\ \isactrlbold {\isasymand}\ K\ i\ {\isacharparenleft}{\kern0pt}conjunct\ G{\isacharparenright}{\kern0pt}{\isacharparenright}{\kern0pt}\ \isactrlbold {\isasymlongrightarrow}\ conjunct\ {\isacharparenleft}{\kern0pt}map\ {\isacharparenleft}{\kern0pt}K\ i{\isacharparenright}{\kern0pt}\ G{\isacharparenright}{\kern0pt}{\isacharparenright}{\kern0pt}%
}%
\SNIP{lemma:K-conjunction-in-mult-13}{%
\ \ \ \ \ \ \ \ \isakeyword{and}\ o{\isadigit{2}}{\isacharcolon}{\kern0pt}\ {\Cartouche{lemma:K-conjunction-in-mult}{13}{0}}\isanewline%
}%
\SNIP{lemma:K-conjunction-in-mult-14}{%
\ \ \ \ \isacommand{apply}\isamarkupfalse%
\ {\isacharparenleft}{\kern0pt}simp\ add{\isacharcolon}{\kern0pt}\ A{\isadigit{1}}{\isacharparenright}{\kern0pt}\isanewline%
}%
\SNIP{lemma:K-conjunction-in-mult-15}{%
\ \ \ \ \isacommand{by}\isamarkupfalse%
\ {\isacharparenleft}{\kern0pt}metis\ Cons{\isachardot}{\kern0pt}hyps\ K{\isacharunderscore}{\kern0pt}imply{\isacharunderscore}{\kern0pt}Cons\ K{\isacharunderscore}{\kern0pt}multi{\isacharunderscore}{\kern0pt}imply\ imply{\isachardot}{\kern0pt}simps{\isacharparenleft}{\kern0pt}{\isadigit{1}}{\isacharparenright}{\kern0pt}\ imply{\isachardot}{\kern0pt}simps{\isacharparenleft}{\kern0pt}{\isadigit{2}}{\isacharparenright}{\kern0pt}{\isacharparenright}{\kern0pt}\isanewline%
}%
\SNIP{lemma:K-conjunction-in-mult-16-0}{%
A\ {\isasymturnstile}\ {\isacharparenleft}{\kern0pt}{\isacharparenleft}{\kern0pt}{\isacharparenleft}{\kern0pt}K\ i\ a{\isacharparenright}{\kern0pt}\ \isactrlbold {\isasymand}\ K\ i\ {\isacharparenleft}{\kern0pt}conjunct\ G{\isacharparenright}{\kern0pt}{\isacharparenright}{\kern0pt}\ \isactrlbold {\isasymlongrightarrow}\ {\isacharparenleft}{\kern0pt}K\ i\ a{\isacharparenright}{\kern0pt}\ \isactrlbold {\isasymand}\ conjunct\ {\isacharparenleft}{\kern0pt}map\ {\isacharparenleft}{\kern0pt}K\ i{\isacharparenright}{\kern0pt}\ G{\isacharparenright}{\kern0pt}{\isacharparenright}{\kern0pt}%
}%
\SNIP{lemma:K-conjunction-in-mult-16}{%
\ \ \isacommand{then}\isamarkupfalse%
\ \isacommand{have}\isamarkupfalse%
\ {\Cartouche{lemma:K-conjunction-in-mult}{16}{0}}\isanewline%
}%
\SNIP{lemma:K-conjunction-in-mult-17}{%
\ \ \ \ \isacommand{using}\isamarkupfalse%
\ K{\isacharunderscore}{\kern0pt}imply{\isacharunderscore}{\kern0pt}multi\ \isacommand{by}\isamarkupfalse%
\ blast\isanewline%
}%
\SNIP{lemma:K-conjunction-in-mult-18}{%
\ \ \isacommand{then}\isamarkupfalse%
\ \isacommand{show}\isamarkupfalse%
\ {\isacharquery}{\kern0pt}case\ \isanewline%
}%
\SNIP{lemma:K-conjunction-in-mult-19}{%
\ \ \ \ \isacommand{using}\isamarkupfalse%
\ K{\isacharunderscore}{\kern0pt}imp{\isacharunderscore}{\kern0pt}trans\ o{\isadigit{1}}\ \isacommand{by}\isamarkupfalse%
\ auto\isanewline%
}%
\SNIP{lemma:K-conjunction-in-mult-20}{%
\isacommand{qed}\isamarkupfalse%
\endisatagproof
{\isafoldproof}%
\isadelimproof%
}%
\SNIP{lemma:K-conjunction-out-0-0}{%
A\ {\isasymturnstile}\ {\isacharparenleft}{\kern0pt}{\isacharparenleft}{\kern0pt}K\ i\ p{\isacharparenright}{\kern0pt}\ \isactrlbold {\isasymand}\ {\isacharparenleft}{\kern0pt}K\ i\ q{\isacharparenright}{\kern0pt}\ \isactrlbold {\isasymlongrightarrow}\ K\ i\ {\isacharparenleft}{\kern0pt}p\ \isactrlbold {\isasymand}\ q{\isacharparenright}{\kern0pt}{\isacharparenright}{\kern0pt}%
}%
\SNIP{lemma:K-conjunction-out-0}{%
\isacommand{lemma}\isamarkupfalse%
\ K{\isacharunderscore}{\kern0pt}conjunction{\isacharunderscore}{\kern0pt}out{\isacharcolon}{\kern0pt}\ {\Cartouche{lemma:K-conjunction-out}{0}{0}}\isanewline%
}%
\SNIP{lemma:K-conjunction-out-1}{%
\isadelimproof
\endisadelimproof
\isatagproof
\isacommand{proof}\isamarkupfalse%
\ {\isacharminus}{\kern0pt}\isanewline%
}%
\SNIP{lemma:K-conjunction-out-2-0}{%
A\ {\isasymturnstile}\ {\isacharparenleft}{\kern0pt}p\ \isactrlbold {\isasymlongrightarrow}\ q\ \isactrlbold {\isasymlongrightarrow}\ p\isactrlbold {\isasymand}q{\isacharparenright}{\kern0pt}%
}%
\SNIP{lemma:K-conjunction-out-2}{%
\ \ \isacommand{have}\isamarkupfalse%
\ {\Cartouche{lemma:K-conjunction-out}{2}{0}}\isanewline%
}%
\SNIP{lemma:K-conjunction-out-3}{%
\ \ \ \ \isacommand{by}\isamarkupfalse%
\ {\isacharparenleft}{\kern0pt}simp\ add{\isacharcolon}{\kern0pt}\ A{\isadigit{1}}{\isacharparenright}{\kern0pt}\isanewline%
}%
\SNIP{lemma:K-conjunction-out-4-0}{%
A\ {\isasymturnstile}\ K\ i\ {\isacharparenleft}{\kern0pt}p\ \isactrlbold {\isasymlongrightarrow}\ q\ \isactrlbold {\isasymlongrightarrow}\ p\isactrlbold {\isasymand}q{\isacharparenright}{\kern0pt}%
}%
\SNIP{lemma:K-conjunction-out-4}{%
\ \ \isacommand{then}\isamarkupfalse%
\ \isacommand{have}\isamarkupfalse%
\ {\Cartouche{lemma:K-conjunction-out}{4}{0}}\isanewline%
}%
\SNIP{lemma:K-conjunction-out-5}{%
\ \ \ \ \isacommand{by}\isamarkupfalse%
\ {\isacharparenleft}{\kern0pt}simp\ add{\isacharcolon}{\kern0pt}\ R{\isadigit{2}}{\isacharparenright}{\kern0pt}\isanewline%
}%
\SNIP{lemma:K-conjunction-out-6-0}{%
A\ {\isasymturnstile}\ {\isacharparenleft}{\kern0pt}{\isacharparenleft}{\kern0pt}K\ i\ p{\isacharparenright}{\kern0pt}\ \isactrlbold {\isasymlongrightarrow}\ K\ i\ {\isacharparenleft}{\kern0pt}q\ \isactrlbold {\isasymlongrightarrow}\ p\isactrlbold {\isasymand}q{\isacharparenright}{\kern0pt}{\isacharparenright}{\kern0pt}%
}%
\SNIP{lemma:K-conjunction-out-6}{%
\ \ \isacommand{then}\isamarkupfalse%
\ \isacommand{have}\isamarkupfalse%
\ {\Cartouche{lemma:K-conjunction-out}{6}{0}}\isanewline%
}%
\SNIP{lemma:K-conjunction-out-7-0}{%
A\ {\isasymturnstile}\ {\isacharparenleft}{\kern0pt}p\ \isactrlbold {\isasymlongrightarrow}\ q\ \isactrlbold {\isasymlongrightarrow}\ p\ \isactrlbold {\isasymand}\ q{\isacharparenright}{\kern0pt}%
}%
\SNIP{lemma:K-conjunction-out-7}{%
\ \ \ \ \isacommand{by}\isamarkupfalse%
\ {\isacharparenleft}{\kern0pt}simp\ add{\isacharcolon}{\kern0pt}\ K{\isacharunderscore}{\kern0pt}map\ {\Cartouche{lemma:K-conjunction-out}{7}{0}}{\isacharparenright}{\kern0pt}\isanewline%
}%
\SNIP{lemma:K-conjunction-out-8-0}{%
A\ {\isasymturnstile}\ {\isacharparenleft}{\kern0pt}{\isacharparenleft}{\kern0pt}K\ i\ p{\isacharparenright}{\kern0pt}\ \isactrlbold {\isasymlongrightarrow}\ {\isacharparenleft}{\kern0pt}K\ i\ q{\isacharparenright}{\kern0pt}\ \isactrlbold {\isasymlongrightarrow}\ K\ i\ {\isacharparenleft}{\kern0pt}p\isactrlbold {\isasymand}q{\isacharparenright}{\kern0pt}{\isacharparenright}{\kern0pt}%
}%
\SNIP{lemma:K-conjunction-out-8}{%
\ \ \isacommand{then}\isamarkupfalse%
\ \isacommand{have}\isamarkupfalse%
\ {\Cartouche{lemma:K-conjunction-out}{8}{0}}\isanewline%
}%
\SNIP{lemma:K-conjunction-out-9}{%
\ \ \ \ \isacommand{by}\isamarkupfalse%
\ {\isacharparenleft}{\kern0pt}meson\ K{\isacharunderscore}{\kern0pt}A{\isadigit{2}}{\isacharprime}{\kern0pt}\ K{\isacharunderscore}{\kern0pt}imp{\isacharunderscore}{\kern0pt}trans{\isacharparenright}{\kern0pt}\isanewline%
}%
\SNIP{lemma:K-conjunction-out-10}{%
\ \ \isacommand{then}\isamarkupfalse%
\ \isacommand{show}\isamarkupfalse%
\ {\isacharquery}{\kern0pt}thesis\isanewline%
}%
\SNIP{lemma:K-conjunction-out-11}{%
\ \ \ \ \isacommand{using}\isamarkupfalse%
\ K{\isacharunderscore}{\kern0pt}multi{\isacharunderscore}{\kern0pt}imply\ \isacommand{by}\isamarkupfalse%
\ blast\isanewline%
}%
\SNIP{lemma:K-conjunction-out-12}{%
\isacommand{qed}\isamarkupfalse%
\endisatagproof
{\isafoldproof}%
\isadelimproof%
}%
\SNIP{lemma:K-conjunction-out-mult-0-0}{%
A\ {\isasymturnstile}\ {\isacharparenleft}{\kern0pt}conjunct\ {\isacharparenleft}{\kern0pt}map\ {\isacharparenleft}{\kern0pt}K\ i{\isacharparenright}{\kern0pt}\ G{\isacharparenright}{\kern0pt}\ \isactrlbold {\isasymlongrightarrow}\ {\isacharparenleft}{\kern0pt}K\ i\ {\isacharparenleft}{\kern0pt}conjunct\ G{\isacharparenright}{\kern0pt}{\isacharparenright}{\kern0pt}{\isacharparenright}{\kern0pt}%
}%
\SNIP{lemma:K-conjunction-out-mult-0}{%
\isacommand{lemma}\isamarkupfalse%
\ K{\isacharunderscore}{\kern0pt}conjunction{\isacharunderscore}{\kern0pt}out{\isacharunderscore}{\kern0pt}mult{\isacharcolon}{\kern0pt}\ {\Cartouche{lemma:K-conjunction-out-mult}{0}{0}}\isanewline%
}%
\SNIP{lemma:K-conjunction-out-mult-1}{%
\isadelimproof
\endisadelimproof
\isatagproof
\isacommand{proof}\isamarkupfalse%
\ {\isacharparenleft}{\kern0pt}induct\ G{\isacharparenright}{\kern0pt}\isanewline%
}%
\SNIP{lemma:K-conjunction-out-mult-2}{%
\ \ \isacommand{case}\isamarkupfalse%
\ Nil\isanewline%
}%
\SNIP{lemma:K-conjunction-out-mult-3}{%
\ \ \isacommand{then}\isamarkupfalse%
\ \isacommand{show}\isamarkupfalse%
\ {\isacharquery}{\kern0pt}case\isanewline%
}%
\SNIP{lemma:K-conjunction-out-mult-4}{%
\ \ \ \ \isacommand{by}\isamarkupfalse%
\ {\isacharparenleft}{\kern0pt}metis\ A{\isadigit{1}}\ K{\isacharunderscore}{\kern0pt}imply{\isacharunderscore}{\kern0pt}conjunct\ Nil{\isacharunderscore}{\kern0pt}is{\isacharunderscore}{\kern0pt}map{\isacharunderscore}{\kern0pt}conv\ R{\isadigit{2}}\ conjunct{\isachardot}{\kern0pt}simps{\isacharparenleft}{\kern0pt}{\isadigit{1}}{\isacharparenright}{\kern0pt}\ eval{\isachardot}{\kern0pt}simps{\isacharparenleft}{\kern0pt}{\isadigit{5}}{\isacharparenright}{\kern0pt}\ imply{\isachardot}{\kern0pt}simps{\isacharparenleft}{\kern0pt}{\isadigit{1}}{\isacharparenright}{\kern0pt}{\isacharparenright}{\kern0pt}\isanewline%
}%
\SNIP{lemma:K-conjunction-out-mult-5}{%
\ \ \isacommand{case}\isamarkupfalse%
\ {\isacharparenleft}{\kern0pt}Cons\ a\ G{\isacharparenright}{\kern0pt}\isanewline%
}%
\SNIP{lemma:K-conjunction-out-mult-6-0}{%
A\ {\isasymturnstile}\ {\isacharparenleft}{\kern0pt}{\isacharparenleft}{\kern0pt}conjunct\ {\isacharparenleft}{\kern0pt}map\ {\isacharparenleft}{\kern0pt}K\ i{\isacharparenright}{\kern0pt}\ {\isacharparenleft}{\kern0pt}a{\isacharhash}{\kern0pt}G{\isacharparenright}{\kern0pt}{\isacharparenright}{\kern0pt}{\isacharparenright}{\kern0pt}\ \isactrlbold {\isasymlongrightarrow}\ {\isacharparenleft}{\kern0pt}{\isacharparenleft}{\kern0pt}K\ i\ a{\isacharparenright}{\kern0pt}\ \isactrlbold {\isasymand}\ conjunct\ {\isacharparenleft}{\kern0pt}map\ {\isacharparenleft}{\kern0pt}K\ i{\isacharparenright}{\kern0pt}\ G{\isacharparenright}{\kern0pt}{\isacharparenright}{\kern0pt}{\isacharparenright}{\kern0pt}%
}%
\SNIP{lemma:K-conjunction-out-mult-6}{%
\ \ \isacommand{then}\isamarkupfalse%
\ \isacommand{have}\isamarkupfalse%
\ {\Cartouche{lemma:K-conjunction-out-mult}{6}{0}}\isanewline%
}%
\SNIP{lemma:K-conjunction-out-mult-7}{%
\ \ \ \ \isacommand{by}\isamarkupfalse%
\ {\isacharparenleft}{\kern0pt}simp\ add{\isacharcolon}{\kern0pt}\ A{\isadigit{1}}{\isacharparenright}{\kern0pt}\isanewline%
}%
\SNIP{lemma:K-conjunction-out-mult-8-0}{%
A\ {\isasymturnstile}\ {\isacharparenleft}{\kern0pt}{\isacharparenleft}{\kern0pt}{\isacharparenleft}{\kern0pt}K\ i\ a{\isacharparenright}{\kern0pt}\ \isactrlbold {\isasymand}\ conjunct\ {\isacharparenleft}{\kern0pt}map\ {\isacharparenleft}{\kern0pt}K\ i{\isacharparenright}{\kern0pt}\ G{\isacharparenright}{\kern0pt}{\isacharparenright}{\kern0pt}\isactrlbold {\isasymlongrightarrow}\ {\isacharparenleft}{\kern0pt}K\ i\ a{\isacharparenright}{\kern0pt}\ \isactrlbold {\isasymand}\ K\ i\ {\isacharparenleft}{\kern0pt}conjunct\ G{\isacharparenright}{\kern0pt}{\isacharparenright}{\kern0pt}%
}%
\SNIP{lemma:K-conjunction-out-mult-8}{%
\ \ \isacommand{then}\isamarkupfalse%
\ \isacommand{have}\isamarkupfalse%
\ {\isacharasterisk}{\kern0pt}{\isacharcolon}{\kern0pt}\ {\Cartouche{lemma:K-conjunction-out-mult}{8}{0}}\isanewline%
}%
\SNIP{lemma:K-conjunction-out-mult-9}{%
\ \ \ \ \isacommand{by}\isamarkupfalse%
\ {\isacharparenleft}{\kern0pt}metis\ Cons{\isachardot}{\kern0pt}hyps\ K{\isacharunderscore}{\kern0pt}imply{\isacharunderscore}{\kern0pt}Cons\ K{\isacharunderscore}{\kern0pt}imply{\isacharunderscore}{\kern0pt}head\ K{\isacharunderscore}{\kern0pt}imply{\isacharunderscore}{\kern0pt}multi\ K{\isacharunderscore}{\kern0pt}multi{\isacharunderscore}{\kern0pt}imply\ imply{\isachardot}{\kern0pt}simps{\isacharparenleft}{\kern0pt}{\isadigit{1}}{\isacharparenright}{\kern0pt}\ imply{\isachardot}{\kern0pt}simps{\isacharparenleft}{\kern0pt}{\isadigit{2}}{\isacharparenright}{\kern0pt}{\isacharparenright}{\kern0pt}\isanewline%
}%
\SNIP{lemma:K-conjunction-out-mult-10-0}{%
A\ {\isasymturnstile}\ {\isacharparenleft}{\kern0pt}{\isacharparenleft}{\kern0pt}{\isacharparenleft}{\kern0pt}K\ i\ a{\isacharparenright}{\kern0pt}\ \isactrlbold {\isasymand}\ K\ i\ {\isacharparenleft}{\kern0pt}conjunct\ G{\isacharparenright}{\kern0pt}{\isacharparenright}{\kern0pt}\isactrlbold {\isasymlongrightarrow}\ K\ i\ {\isacharparenleft}{\kern0pt}conjunct\ {\isacharparenleft}{\kern0pt}a{\isacharhash}{\kern0pt}G{\isacharparenright}{\kern0pt}{\isacharparenright}{\kern0pt}{\isacharparenright}{\kern0pt}%
}%
\SNIP{lemma:K-conjunction-out-mult-10}{%
\ \ \isacommand{then}\isamarkupfalse%
\ \isacommand{have}\isamarkupfalse%
\ {\Cartouche{lemma:K-conjunction-out-mult}{10}{0}}\isanewline%
}%
\SNIP{lemma:K-conjunction-out-mult-11}{%
\ \ \ \ \isacommand{by}\isamarkupfalse%
\ {\isacharparenleft}{\kern0pt}simp\ add{\isacharcolon}{\kern0pt}\ K{\isacharunderscore}{\kern0pt}conjunction{\isacharunderscore}{\kern0pt}out{\isacharparenright}{\kern0pt}\isanewline%
}%
\SNIP{lemma:K-conjunction-out-mult-12}{%
\ \ \isacommand{then}\isamarkupfalse%
\ \isacommand{show}\isamarkupfalse%
\ {\isacharquery}{\kern0pt}case\isanewline%
}%
\SNIP{lemma:K-conjunction-out-mult-13}{%
\ \ \ \ \isacommand{using}\isamarkupfalse%
\ {\isachardoublequoteopen}{\isacharasterisk}{\kern0pt}{\isachardoublequoteclose}\ K{\isacharunderscore}{\kern0pt}imp{\isacharunderscore}{\kern0pt}trans\ \isacommand{by}\isamarkupfalse%
\ auto\isanewline%
}%
\SNIP{lemma:K-conjunction-out-mult-14}{%
\isacommand{qed}\isamarkupfalse%
\endisatagproof
{\isafoldproof}%
\isadelimproof
\endisadelimproof
\isadelimdocument
\endisadelimdocument
\isatagdocument%
}%
\SNIP{lemma:mcs-conjunction-0}{%
\isacommand{lemma}\isamarkupfalse%
\ mcs{\isacharunderscore}{\kern0pt}conjunction{\isacharcolon}{\kern0pt}\ \isanewline%
}%
\SNIP{lemma:mcs-conjunction-1-0}{%
consistent\ A\ V%
}%
\SNIP{lemma:mcs-conjunction-1-1}{%
maximal\ A\ V%
}%
\SNIP{lemma:mcs-conjunction-1}{%
\ \ \isakeyword{assumes}\ {\Cartouche{lemma:mcs-conjunction}{1}{0}}\ \isakeyword{and}\ {\Cartouche{lemma:mcs-conjunction}{1}{1}}\ \isanewline%
}%
\SNIP{lemma:mcs-conjunction-2-0}{%
p\ {\isasymin}\ V\ {\isasymand}\ q\ {\isasymin}\ V\ {\isasymlongrightarrow}\ {\isacharparenleft}{\kern0pt}p\ \isactrlbold {\isasymand}\ q{\isacharparenright}{\kern0pt}\ {\isasymin}\ V%
}%
\SNIP{lemma:mcs-conjunction-2}{%
\ \ \isakeyword{shows}\ {\Cartouche{lemma:mcs-conjunction}{2}{0}}\isanewline%
}%
\SNIP{lemma:mcs-conjunction-3}{%
\isadelimproof
\endisadelimproof
\isatagproof
\isacommand{proof}\isamarkupfalse%
\ {\isacharminus}{\kern0pt}\ \isanewline%
}%
\SNIP{lemma:mcs-conjunction-4-0}{%
tautology\ {\isacharparenleft}{\kern0pt}p\ \isactrlbold {\isasymlongrightarrow}\ q\ \isactrlbold {\isasymlongrightarrow}\ {\isacharparenleft}{\kern0pt}p\isactrlbold {\isasymand}q{\isacharparenright}{\kern0pt}{\isacharparenright}{\kern0pt}%
}%
\SNIP{lemma:mcs-conjunction-4}{%
\ \ \isacommand{have}\isamarkupfalse%
\ {\Cartouche{lemma:mcs-conjunction}{4}{0}}\isanewline%
}%
\SNIP{lemma:mcs-conjunction-5}{%
\ \ \ \ \isacommand{by}\isamarkupfalse%
\ force\ \ \isanewline%
}%
\SNIP{lemma:mcs-conjunction-6-0}{%
{\isacharparenleft}{\kern0pt}p\ \isactrlbold {\isasymlongrightarrow}\ q\ \isactrlbold {\isasymlongrightarrow}\ {\isacharparenleft}{\kern0pt}p\isactrlbold {\isasymand}q{\isacharparenright}{\kern0pt}{\isacharparenright}{\kern0pt}\ {\isasymin}\ V%
}%
\SNIP{lemma:mcs-conjunction-6}{%
\ \ \isacommand{then}\isamarkupfalse%
\ \isacommand{have}\isamarkupfalse%
\ {\Cartouche{lemma:mcs-conjunction}{6}{0}}\isanewline%
}%
\SNIP{lemma:mcs-conjunction-7}{%
\ \ \ \ \isacommand{using}\isamarkupfalse%
\ A{\isadigit{1}}\ assms{\isacharparenleft}{\kern0pt}{\isadigit{1}}{\isacharparenright}{\kern0pt}\ assms{\isacharparenleft}{\kern0pt}{\isadigit{2}}{\isacharparenright}{\kern0pt}\ deriv{\isacharunderscore}{\kern0pt}in{\isacharunderscore}{\kern0pt}maximal\ \isacommand{by}\isamarkupfalse%
\ blast\isanewline%
}%
\SNIP{lemma:mcs-conjunction-8-0}{%
p\ {\isasymin}\ V\ {\isasymlongrightarrow}\ {\isacharparenleft}{\kern0pt}q\ \isactrlbold {\isasymlongrightarrow}\ {\isacharparenleft}{\kern0pt}p\isactrlbold {\isasymand}q{\isacharparenright}{\kern0pt}{\isacharparenright}{\kern0pt}\ {\isasymin}\ V%
}%
\SNIP{lemma:mcs-conjunction-8}{%
\ \ \isacommand{then}\isamarkupfalse%
\ \isacommand{have}\isamarkupfalse%
\ {\Cartouche{lemma:mcs-conjunction}{8}{0}}\isanewline%
}%
\SNIP{lemma:mcs-conjunction-9}{%
\ \ \ \ \isacommand{by}\isamarkupfalse%
\ {\isacharparenleft}{\kern0pt}meson\ assms{\isacharparenleft}{\kern0pt}{\isadigit{1}}{\isacharparenright}{\kern0pt}\ assms{\isacharparenleft}{\kern0pt}{\isadigit{2}}{\isacharparenright}{\kern0pt}\ consequent{\isacharunderscore}{\kern0pt}in{\isacharunderscore}{\kern0pt}maximal{\isacharparenright}{\kern0pt}\isanewline%
}%
\SNIP{lemma:mcs-conjunction-10}{%
\ \ \isacommand{then}\isamarkupfalse%
\ \isacommand{show}\isamarkupfalse%
\ {\isacharquery}{\kern0pt}thesis\isanewline%
}%
\SNIP{lemma:mcs-conjunction-11}{%
\ \ \ \ \isacommand{using}\isamarkupfalse%
\ assms{\isacharparenleft}{\kern0pt}{\isadigit{1}}{\isacharparenright}{\kern0pt}\ assms{\isacharparenleft}{\kern0pt}{\isadigit{2}}{\isacharparenright}{\kern0pt}\ consequent{\isacharunderscore}{\kern0pt}in{\isacharunderscore}{\kern0pt}maximal\ \isacommand{by}\isamarkupfalse%
\ blast\isanewline%
}%
\SNIP{lemma:mcs-conjunction-12}{%
\isacommand{qed}\isamarkupfalse%
\endisatagproof
{\isafoldproof}%
\isadelimproof%
}%
\SNIP{lemma:mcs-conjunction-mult-0}{%
\isacommand{lemma}\isamarkupfalse%
\ mcs{\isacharunderscore}{\kern0pt}conjunction{\isacharunderscore}{\kern0pt}mult{\isacharcolon}{\kern0pt}\ \isanewline%
}%
\SNIP{lemma:mcs-conjunction-mult-1-0}{%
consistent\ A\ V%
}%
\SNIP{lemma:mcs-conjunction-mult-1-1}{%
maximal\ A\ V%
}%
\SNIP{lemma:mcs-conjunction-mult-1}{%
\ \isakeyword{assumes}\ {\Cartouche{lemma:mcs-conjunction-mult}{1}{0}}\ \isakeyword{and}\ {\Cartouche{lemma:mcs-conjunction-mult}{1}{1}}\ \isanewline%
}%
\SNIP{lemma:mcs-conjunction-mult-2-0}{%
{\isacharparenleft}{\kern0pt}set\ {\isacharparenleft}{\kern0pt}S\ {\isacharcolon}{\kern0pt}{\isacharcolon}{\kern0pt}\ {\isacharparenleft}{\kern0pt}{\isacharprime}{\kern0pt}i\ fm\ list{\isacharparenright}{\kern0pt}{\isacharparenright}{\kern0pt}\ {\isasymsubseteq}\ V\ {\isasymand}\ finite\ {\isacharparenleft}{\kern0pt}set\ S{\isacharparenright}{\kern0pt}{\isacharparenright}{\kern0pt}\ {\isasymlongrightarrow}\ {\isacharparenleft}{\kern0pt}conjunct\ S{\isacharparenright}{\kern0pt}\ {\isasymin}\ V%
}%
\SNIP{lemma:mcs-conjunction-mult-2}{%
\ \ \isakeyword{shows}\ {\Cartouche{lemma:mcs-conjunction-mult}{2}{0}}\isanewline%
}%
\SNIP{lemma:mcs-conjunction-mult-3}{%
\isadelimproof
\endisadelimproof
\isatagproof
\isacommand{proof}\isamarkupfalse%
{\isacharparenleft}{\kern0pt}induct\ S{\isacharparenright}{\kern0pt}\isanewline%
}%
\SNIP{lemma:mcs-conjunction-mult-4}{%
\ \ \isacommand{case}\isamarkupfalse%
\ Nil\isanewline%
}%
\SNIP{lemma:mcs-conjunction-mult-5}{%
\ \ \isacommand{then}\isamarkupfalse%
\ \isacommand{show}\isamarkupfalse%
\ {\isacharquery}{\kern0pt}case\isanewline%
}%
\SNIP{lemma:mcs-conjunction-mult-6}{%
\ \ \ \ \isacommand{by}\isamarkupfalse%
\ {\isacharparenleft}{\kern0pt}metis\ K{\isacharunderscore}{\kern0pt}Boole\ assms{\isacharparenleft}{\kern0pt}{\isadigit{1}}{\isacharparenright}{\kern0pt}\ assms{\isacharparenleft}{\kern0pt}{\isadigit{2}}{\isacharparenright}{\kern0pt}\ conjunct{\isachardot}{\kern0pt}simps{\isacharparenleft}{\kern0pt}{\isadigit{1}}{\isacharparenright}{\kern0pt}\ consistent{\isacharunderscore}{\kern0pt}def\ inconsistent{\isacharunderscore}{\kern0pt}subset\ maximal{\isacharunderscore}{\kern0pt}def{\isacharparenright}{\kern0pt}\isanewline%
}%
\SNIP{lemma:mcs-conjunction-mult-7}{%
\ \ \isacommand{case}\isamarkupfalse%
\ {\isacharparenleft}{\kern0pt}Cons\ a\ S{\isacharparenright}{\kern0pt}\isanewline%
}%
\SNIP{lemma:mcs-conjunction-mult-8-0}{%
set\ S\ {\isasymsubseteq}\ set\ {\isacharparenleft}{\kern0pt}a{\isacharhash}{\kern0pt}S{\isacharparenright}{\kern0pt}%
}%
\SNIP{lemma:mcs-conjunction-mult-8}{%
\ \ \isacommand{then}\isamarkupfalse%
\ \isacommand{have}\isamarkupfalse%
\ {\Cartouche{lemma:mcs-conjunction-mult}{8}{0}}\isanewline%
}%
\SNIP{lemma:mcs-conjunction-mult-9}{%
\ \ \ \ \isacommand{by}\isamarkupfalse%
\ {\isacharparenleft}{\kern0pt}meson\ set{\isacharunderscore}{\kern0pt}subset{\isacharunderscore}{\kern0pt}Cons{\isacharparenright}{\kern0pt}\isanewline%
}%
\SNIP{lemma:mcs-conjunction-mult-10-0}{%
\ set\ {\isacharparenleft}{\kern0pt}a\ {\isacharhash}{\kern0pt}\ S{\isacharparenright}{\kern0pt}\ {\isasymsubseteq}\ V\ {\isasymand}\ finite\ {\isacharparenleft}{\kern0pt}set\ {\isacharparenleft}{\kern0pt}a\ {\isacharhash}{\kern0pt}\ S{\isacharparenright}{\kern0pt}{\isacharparenright}{\kern0pt}\ {\isasymlongrightarrow}\ conjunct\ {\isacharparenleft}{\kern0pt}S{\isacharparenright}{\kern0pt}\ {\isasymin}\ V\ {\isasymand}\ a\ {\isasymin}\ V%
}%
\SNIP{lemma:mcs-conjunction-mult-10}{%
\ \ \isacommand{then}\isamarkupfalse%
\ \isacommand{have}\isamarkupfalse%
\ c{\isadigit{1}}{\isacharcolon}{\kern0pt}\ {\Cartouche{lemma:mcs-conjunction-mult}{10}{0}}\isanewline%
}%
\SNIP{lemma:mcs-conjunction-mult-11}{%
\ \ \ \ \isacommand{using}\isamarkupfalse%
\ Cons\ \isacommand{by}\isamarkupfalse%
\ fastforce\isanewline%
}%
\SNIP{lemma:mcs-conjunction-mult-12-0}{%
\ conjunct\ {\isacharparenleft}{\kern0pt}S{\isacharparenright}{\kern0pt}\ {\isasymin}\ V\ {\isasymand}\ a\ {\isasymin}\ V\ {\isasymlongrightarrow}\ {\isacharparenleft}{\kern0pt}conjunct\ {\isacharparenleft}{\kern0pt}a{\isacharhash}{\kern0pt}S{\isacharparenright}{\kern0pt}{\isacharparenright}{\kern0pt}\ {\isasymin}\ V\ %
}%
\SNIP{lemma:mcs-conjunction-mult-12}{%
\ \ \isacommand{then}\isamarkupfalse%
\ \isacommand{have}\isamarkupfalse%
\ {\Cartouche{lemma:mcs-conjunction-mult}{12}{0}}\isanewline%
}%
\SNIP{lemma:mcs-conjunction-mult-13}{%
\ \ \ \ \isacommand{using}\isamarkupfalse%
\ assms{\isacharparenleft}{\kern0pt}{\isadigit{1}}{\isacharparenright}{\kern0pt}\ assms{\isacharparenleft}{\kern0pt}{\isadigit{2}}{\isacharparenright}{\kern0pt}\ mcs{\isacharunderscore}{\kern0pt}conjunction\ \isacommand{by}\isamarkupfalse%
\ auto\isanewline%
}%
\SNIP{lemma:mcs-conjunction-mult-14}{%
\ \ \isacommand{then}\isamarkupfalse%
\ \isacommand{show}\isamarkupfalse%
\ {\isacharquery}{\kern0pt}case\isanewline%
}%
\SNIP{lemma:mcs-conjunction-mult-15}{%
\ \ \ \ \isacommand{using}\isamarkupfalse%
\ c{\isadigit{1}}\ \isacommand{by}\isamarkupfalse%
\ fastforce\isanewline%
}%
\SNIP{lemma:mcs-conjunction-mult-16}{%
\isacommand{qed}\isamarkupfalse%
\endisatagproof
{\isafoldproof}%
\isadelimproof%
}%
\SNIP{lemma:reach-dualK-0}{%
\isacommand{lemma}\isamarkupfalse%
\ reach{\isacharunderscore}{\kern0pt}dualK{\isacharcolon}{\kern0pt}\ \isanewline%
}%
\SNIP{lemma:reach-dualK-1-0}{%
consistent\ A\ V%
}%
\SNIP{lemma:reach-dualK-1-1}{%
maximal\ A\ V%
}%
\SNIP{lemma:reach-dualK-1}{%
\ \ \isakeyword{assumes}\ {\Cartouche{lemma:reach-dualK}{1}{0}}\ {\Cartouche{lemma:reach-dualK}{1}{1}}\isanewline%
}%
\SNIP{lemma:reach-dualK-2-0}{%
consistent\ A\ W%
}%
\SNIP{lemma:reach-dualK-2-1}{%
maximal\ A\ W%
}%
\SNIP{lemma:reach-dualK-2-2}{%
W\ {\isasymin}\ reach\ A\ i\ V%
}%
\SNIP{lemma:reach-dualK-2}{%
\ \ \ \ \isakeyword{and}\ {\Cartouche{lemma:reach-dualK}{2}{0}}\ {\Cartouche{lemma:reach-dualK}{2}{1}}\ {\Cartouche{lemma:reach-dualK}{2}{2}}\isanewline%
}%
\SNIP{lemma:reach-dualK-3-0}{%
{\isasymforall}p{\isachardot}{\kern0pt}\ p\ {\isasymin}\ W\ {\isasymlongrightarrow}\ {\isacharparenleft}{\kern0pt}L\ i\ p{\isacharparenright}{\kern0pt}\ {\isasymin}\ V%
}%
\SNIP{lemma:reach-dualK-3}{%
\ \ \isakeyword{shows}\ {\Cartouche{lemma:reach-dualK}{3}{0}}\isanewline%
}%
\SNIP{lemma:reach-dualK-4}{%
\isadelimproof
\endisadelimproof
\isatagproof
\isacommand{proof}\isamarkupfalse%
\ {\isacharparenleft}{\kern0pt}rule\ ccontr{\isacharparenright}{\kern0pt}\isanewline%
}%
\SNIP{lemma:reach-dualK-5-0}{%
{\isasymnot}\ {\isacharparenleft}{\kern0pt}{\isasymforall}p{\isachardot}{\kern0pt}\ p\ {\isasymin}\ W\ {\isasymlongrightarrow}\ {\isacharparenleft}{\kern0pt}L\ i\ p{\isacharparenright}{\kern0pt}\ {\isasymin}\ V{\isacharparenright}{\kern0pt}%
}%
\SNIP{lemma:reach-dualK-5}{%
\ \ \isacommand{assume}\isamarkupfalse%
\ {\Cartouche{lemma:reach-dualK}{5}{0}}\isanewline%
}%
\SNIP{lemma:reach-dualK-6-0}{%
p{\isacharprime}{\kern0pt}\ {\isasymin}\ W\ {\isasymand}\ {\isacharparenleft}{\kern0pt}L\ i\ p{\isacharprime}{\kern0pt}{\isacharparenright}{\kern0pt}\ {\isasymnotin}\ V%
}%
\SNIP{lemma:reach-dualK-6}{%
\ \ \isacommand{then}\isamarkupfalse%
\ \isacommand{obtain}\isamarkupfalse%
\ p{\isacharprime}{\kern0pt}\ \isakeyword{where}\ {\isacharasterisk}{\kern0pt}{\isacharcolon}{\kern0pt}\ {\Cartouche{lemma:reach-dualK}{6}{0}}\isanewline%
}%
\SNIP{lemma:reach-dualK-7}{%
\ \ \ \ \isacommand{by}\isamarkupfalse%
\ presburger\isanewline%
}%
\SNIP{lemma:reach-dualK-8-0}{%
{\isacharparenleft}{\kern0pt}\isactrlbold {\isasymnot}\ L\ i\ p{\isacharprime}{\kern0pt}{\isacharparenright}{\kern0pt}\ {\isasymin}\ V%
}%
\SNIP{lemma:reach-dualK-8}{%
\ \ \isacommand{then}\isamarkupfalse%
\ \isacommand{have}\isamarkupfalse%
\ {\Cartouche{lemma:reach-dualK}{8}{0}}\isanewline%
}%
\SNIP{lemma:reach-dualK-9}{%
\ \ \ \ \isacommand{using}\isamarkupfalse%
\ assms{\isacharparenleft}{\kern0pt}{\isadigit{1}}{\isacharparenright}{\kern0pt}\ assms{\isacharparenleft}{\kern0pt}{\isadigit{2}}{\isacharparenright}{\kern0pt}\ assms{\isacharparenleft}{\kern0pt}{\isadigit{3}}{\isacharparenright}{\kern0pt}\ assms{\isacharparenleft}{\kern0pt}{\isadigit{4}}{\isacharparenright}{\kern0pt}\ assms{\isacharparenleft}{\kern0pt}{\isadigit{5}}{\isacharparenright}{\kern0pt}\ exactly{\isacharunderscore}{\kern0pt}one{\isacharunderscore}{\kern0pt}in{\isacharunderscore}{\kern0pt}maximal\ \isacommand{by}\isamarkupfalse%
\ blast\isanewline%
}%
\SNIP{lemma:reach-dualK-10-0}{%
K\ i\ {\isacharparenleft}{\kern0pt}\isactrlbold {\isasymnot}\ p{\isacharprime}{\kern0pt}{\isacharparenright}{\kern0pt}\ {\isasymin}\ V%
}%
\SNIP{lemma:reach-dualK-10}{%
\ \ \isacommand{then}\isamarkupfalse%
\ \isacommand{have}\isamarkupfalse%
\ {\Cartouche{lemma:reach-dualK}{10}{0}}\isanewline%
}%
\SNIP{lemma:reach-dualK-11}{%
\ \ \ \ \isacommand{using}\isamarkupfalse%
\ assms{\isacharparenleft}{\kern0pt}{\isadigit{1}}{\isacharparenright}{\kern0pt}\ assms{\isacharparenleft}{\kern0pt}{\isadigit{2}}{\isacharparenright}{\kern0pt}\ exactly{\isacharunderscore}{\kern0pt}one{\isacharunderscore}{\kern0pt}in{\isacharunderscore}{\kern0pt}maximal\ \isacommand{by}\isamarkupfalse%
\ blast\isanewline%
}%
\SNIP{lemma:reach-dualK-12-0}{%
{\isacharparenleft}{\kern0pt}\isactrlbold {\isasymnot}\ p{\isacharprime}{\kern0pt}{\isacharparenright}{\kern0pt}\ {\isasymin}\ W%
}%
\SNIP{lemma:reach-dualK-12}{%
\ \ \isacommand{then}\isamarkupfalse%
\ \isacommand{have}\isamarkupfalse%
\ {\Cartouche{lemma:reach-dualK}{12}{0}}\isanewline%
}%
\SNIP{lemma:reach-dualK-13}{%
\ \ \ \ \isacommand{using}\isamarkupfalse%
\ assms{\isacharparenleft}{\kern0pt}{\isadigit{5}}{\isacharparenright}{\kern0pt}\ \isacommand{by}\isamarkupfalse%
\ blast\isanewline%
}%
\SNIP{lemma:reach-dualK-14}{%
\ \ \isacommand{then}\isamarkupfalse%
\ \isacommand{show}\isamarkupfalse%
\ False\isanewline%
}%
\SNIP{lemma:reach-dualK-15}{%
\ \ \ \ \isacommand{by}\isamarkupfalse%
\ {\isacharparenleft}{\kern0pt}meson\ {\isachardoublequoteopen}{\isacharasterisk}{\kern0pt}{\isachardoublequoteclose}\ assms{\isacharparenleft}{\kern0pt}{\isadigit{3}}{\isacharparenright}{\kern0pt}\ assms{\isacharparenleft}{\kern0pt}{\isadigit{4}}{\isacharparenright}{\kern0pt}\ exactly{\isacharunderscore}{\kern0pt}one{\isacharunderscore}{\kern0pt}in{\isacharunderscore}{\kern0pt}maximal{\isacharparenright}{\kern0pt}\isanewline%
}%
\SNIP{lemma:reach-dualK-16}{%
\isacommand{qed}\isamarkupfalse%
\endisatagproof
{\isafoldproof}%
\isadelimproof
\%
}%
\SNIP{lemma:dual-reach-0}{%
\isacommand{lemma}\isamarkupfalse%
\ dual{\isacharunderscore}{\kern0pt}reach{\isacharcolon}{\kern0pt}\ \isanewline%
}%
\SNIP{lemma:dual-reach-1-0}{%
consistent\ A\ V%
}%
\SNIP{lemma:dual-reach-1-1}{%
maximal\ A\ V%
}%
\SNIP{lemma:dual-reach-1}{%
\ \ \isakeyword{assumes}\ {\Cartouche{lemma:dual-reach}{1}{0}}\ {\Cartouche{lemma:dual-reach}{1}{1}}\ \isanewline%
}%
\SNIP{lemma:dual-reach-2-0}{%
{\isacharparenleft}{\kern0pt}L\ i\ p{\isacharparenright}{\kern0pt}\ {\isasymin}\ V\ {\isasymlongrightarrow}\ {\isacharparenleft}{\kern0pt}{\isasymexists}\ W{\isachardot}{\kern0pt}\ W\ {\isasymin}\ reach\ A\ i\ V\ {\isasymand}\ p\ {\isasymin}\ W{\isacharparenright}{\kern0pt}%
}%
\SNIP{lemma:dual-reach-2}{%
\ \ \isakeyword{shows}\ {\Cartouche{lemma:dual-reach}{2}{0}}\isanewline%
}%
\SNIP{lemma:dual-reach-3}{%
\isadelimproof
\endisadelimproof
\isatagproof
\isacommand{proof}\isamarkupfalse%
\ {\isacharminus}{\kern0pt}\isanewline%
}%
\SNIP{lemma:dual-reach-4-0}{%
{\isacharparenleft}{\kern0pt}{\isasymnexists}W{\isachardot}{\kern0pt}\ W\ {\isasymin}\ {\isacharbraceleft}{\kern0pt}W{\isachardot}{\kern0pt}\ known\ V\ i\ {\isasymsubseteq}\ W{\isacharbraceright}{\kern0pt}\ {\isasymand}\ p\ {\isasymin}\ W{\isacharparenright}{\kern0pt}\ {\isasymlongrightarrow}\ {\isacharparenleft}{\kern0pt}{\isasymforall}W{\isachardot}{\kern0pt}\ W\ {\isasymin}\ {\isacharbraceleft}{\kern0pt}W{\isachardot}{\kern0pt}\ known\ V\ i\ {\isasymsubseteq}\ W{\isacharbraceright}{\kern0pt}\ {\isasymlongrightarrow}\ {\isacharparenleft}{\kern0pt}\isactrlbold {\isasymnot}p{\isacharparenright}{\kern0pt}\ {\isasymin}\ W{\isacharparenright}{\kern0pt}%
}%
\SNIP{lemma:dual-reach-4}{%
\ \ \isacommand{have}\isamarkupfalse%
\ {\Cartouche{lemma:dual-reach}{4}{0}}\isanewline%
}%
\SNIP{lemma:dual-reach-5}{%
\ \ \ \ \isacommand{by}\isamarkupfalse%
\ blast\isanewline%
}%
\SNIP{lemma:dual-reach-6-0}{%
{\isacharparenleft}{\kern0pt}{\isasymforall}W{\isachardot}{\kern0pt}\ W\ {\isasymin}\ {\isacharbraceleft}{\kern0pt}W{\isachardot}{\kern0pt}\ known\ V\ i\ {\isasymsubseteq}\ W{\isacharbraceright}{\kern0pt}\ {\isasymlongrightarrow}\ {\isacharparenleft}{\kern0pt}\isactrlbold {\isasymnot}p{\isacharparenright}{\kern0pt}\ {\isasymin}\ W{\isacharparenright}{\kern0pt}\ {\isasymlongrightarrow}\ {\isacharparenleft}{\kern0pt}{\isasymforall}\ W{\isachardot}{\kern0pt}\ W\ {\isasymin}\ reach\ A\ i\ V\ {\isasymlongrightarrow}\ {\isacharparenleft}{\kern0pt}\isactrlbold {\isasymnot}p{\isacharparenright}{\kern0pt}\ {\isasymin}\ W{\isacharparenright}{\kern0pt}%
}%
\SNIP{lemma:dual-reach-6}{%
\ \ \isacommand{then}\isamarkupfalse%
\ \isacommand{have}\isamarkupfalse%
\ {\Cartouche{lemma:dual-reach}{6}{0}}\isanewline%
}%
\SNIP{lemma:dual-reach-7}{%
\ \ \ \ \isacommand{by}\isamarkupfalse%
\ fastforce\isanewline%
}%
\SNIP{lemma:dual-reach-8-0}{%
{\isacharparenleft}{\kern0pt}{\isasymforall}\ W{\isachardot}{\kern0pt}\ W\ {\isasymin}\ reach\ A\ i\ V\ {\isasymlongrightarrow}\ {\isacharparenleft}{\kern0pt}\isactrlbold {\isasymnot}p{\isacharparenright}{\kern0pt}\ {\isasymin}\ W{\isacharparenright}{\kern0pt}\ {\isasymlongrightarrow}\ {\isacharparenleft}{\kern0pt}{\isacharparenleft}{\kern0pt}K\ i\ {\isacharparenleft}{\kern0pt}\isactrlbold {\isasymnot}p{\isacharparenright}{\kern0pt}{\isacharparenright}{\kern0pt}\ {\isasymin}\ V{\isacharparenright}{\kern0pt}%
}%
\SNIP{lemma:dual-reach-8}{%
\ \ \isacommand{then}\isamarkupfalse%
\ \isacommand{have}\isamarkupfalse%
\ {\Cartouche{lemma:dual-reach}{8}{0}}\isanewline%
}%
\SNIP{lemma:dual-reach-9}{%
\ \ \ \ \isacommand{by}\isamarkupfalse%
\ blast\isanewline%
}%
\SNIP{lemma:dual-reach-10-0}{%
{\isacharparenleft}{\kern0pt}{\isacharparenleft}{\kern0pt}K\ i\ {\isacharparenleft}{\kern0pt}\isactrlbold {\isasymnot}p{\isacharparenright}{\kern0pt}{\isacharparenright}{\kern0pt}\ {\isasymin}\ V{\isacharparenright}{\kern0pt}\ {\isasymlongrightarrow}\ {\isacharparenleft}{\kern0pt}{\isasymnot}{\isacharparenleft}{\kern0pt}{\isacharparenleft}{\kern0pt}L\ i\ p{\isacharparenright}{\kern0pt}\ {\isasymin}\ V{\isacharparenright}{\kern0pt}{\isacharparenright}{\kern0pt}%
}%
\SNIP{lemma:dual-reach-10}{%
\ \ \isacommand{then}\isamarkupfalse%
\ \isacommand{have}\isamarkupfalse%
\ {\Cartouche{lemma:dual-reach}{10}{0}}\isanewline%
}%
\SNIP{lemma:dual-reach-11}{%
\ \ \ \ \isacommand{using}\isamarkupfalse%
\ assms{\isacharparenleft}{\kern0pt}{\isadigit{1}}{\isacharparenright}{\kern0pt}\ assms{\isacharparenleft}{\kern0pt}{\isadigit{2}}{\isacharparenright}{\kern0pt}\ exactly{\isacharunderscore}{\kern0pt}one{\isacharunderscore}{\kern0pt}in{\isacharunderscore}{\kern0pt}maximal\ \isacommand{by}\isamarkupfalse%
\ blast\isanewline%
}%
\SNIP{lemma:dual-reach-12-0}{%
{\isacharparenleft}{\kern0pt}{\isasymnexists}W{\isachardot}{\kern0pt}\ W\ {\isasymin}\ {\isacharbraceleft}{\kern0pt}W{\isachardot}{\kern0pt}\ known\ V\ i\ {\isasymsubseteq}\ W{\isacharbraceright}{\kern0pt}\ {\isasymand}\ p\ {\isasymin}\ W{\isacharparenright}{\kern0pt}\ {\isasymlongrightarrow}\ {\isasymnot}{\isacharparenleft}{\kern0pt}{\isacharparenleft}{\kern0pt}L\ i\ p{\isacharparenright}{\kern0pt}\ {\isasymin}\ V{\isacharparenright}{\kern0pt}%
}%
\SNIP{lemma:dual-reach-12}{%
\ \ \isacommand{then}\isamarkupfalse%
\ \isacommand{have}\isamarkupfalse%
\ {\Cartouche{lemma:dual-reach}{12}{0}}\ \isanewline%
}%
\SNIP{lemma:dual-reach-13}{%
\ \ \ \ \isacommand{by}\isamarkupfalse%
\ blast\isanewline%
}%
\SNIP{lemma:dual-reach-14}{%
\ \ \isacommand{then}\isamarkupfalse%
\ \isacommand{show}\isamarkupfalse%
\ {\isacharquery}{\kern0pt}thesis\isanewline%
}%
\SNIP{lemma:dual-reach-15}{%
\ \ \ \ \isacommand{by}\isamarkupfalse%
\ blast\isanewline%
}%
\SNIP{lemma:dual-reach-16}{%
\isacommand{qed}\isamarkupfalse%
\endisatagproof
{\isafoldproof}%
\isadelimproof
\endisadelimproof
\isadelimdocument
\endisadelimdocument
\isatagdocument%
}%
\SNIP{definition:weakly-directed-0-0}{%
{\isacharparenleft}{\kern0pt}{\isacharprime}{\kern0pt}i{\isacharcomma}{\kern0pt}\ {\isacharprime}{\kern0pt}s{\isacharparenright}{\kern0pt}\ kripke\ {\isasymRightarrow}\ bool%
}%
\SNIP{definition:weakly-directed-0}{%
\isacommand{definition}\isamarkupfalse%
\ weakly{\isacharunderscore}{\kern0pt}directed\ {\isacharcolon}{\kern0pt}{\isacharcolon}{\kern0pt}\ {\Cartouche{definition:weakly-directed}{0}{0}}\ \isakeyword{where}\isanewline%
}%
\SNIP{definition:weakly-directed-1-0}{%
weakly{\isacharunderscore}{\kern0pt}directed\ M\ {\isasymequiv}\ {\isasymforall}i{\isachardot}{\kern0pt}\ {\isasymforall}s\ {\isasymin}\ {\isasymW}\ M{\isachardot}{\kern0pt}\ {\isasymforall}t\ {\isasymin}\ {\isasymW}\ M{\isachardot}{\kern0pt}\ {\isasymforall}r\ {\isasymin}\ {\isasymW}\ M{\isachardot}{\kern0pt}\isanewline\ \ \ \ \ {\isacharparenleft}{\kern0pt}r\ {\isasymin}\ {\isasymK}\ M\ i\ s\ {\isasymand}\ t\ {\isasymin}\ {\isasymK}\ M\ i\ s{\isacharparenright}{\kern0pt}{\isasymlongrightarrow}{\isacharparenleft}{\kern0pt}{\isasymexists}\ u\ {\isasymin}\ {\isasymW}\ M{\isachardot}{\kern0pt}\ {\isacharparenleft}{\kern0pt}u\ {\isasymin}\ {\isasymK}\ M\ i\ r\ {\isasymand}\ u\ {\isasymin}\ {\isasymK}\ M\ i\ t{\isacharparenright}{\kern0pt}{\isacharparenright}{\kern0pt}%
}%
\SNIP{definition:weakly-directed-1}{%
\ \ {\isacartoucheopen}weakly{\isacharunderscore}{\kern0pt}directed\ M\ {\isasymequiv}\ {\isasymforall}i{\isachardot}{\kern0pt}\ {\isasymforall}s\ {\isasymin}\ {\isasymW}\ M{\isachardot}{\kern0pt}\ {\isasymforall}t\ {\isasymin}\ {\isasymW}\ M{\isachardot}{\kern0pt}\ {\isasymforall}r\ {\isasymin}\ {\isasymW}\ M{\isachardot}{\kern0pt}\isanewline%
}%
\SNIP{definition:weakly-directed-2}{%
\ \ \ \ \ {\isacharparenleft}{\kern0pt}r\ {\isasymin}\ {\isasymK}\ M\ i\ s\ {\isasymand}\ t\ {\isasymin}\ {\isasymK}\ M\ i\ s{\isacharparenright}{\kern0pt}{\isasymlongrightarrow}{\isacharparenleft}{\kern0pt}{\isasymexists}\ u\ {\isasymin}\ {\isasymW}\ M{\isachardot}{\kern0pt}\ {\isacharparenleft}{\kern0pt}u\ {\isasymin}\ {\isasymK}\ M\ i\ r\ {\isasymand}\ u\ {\isasymin}\ {\isasymK}\ M\ i\ t{\isacharparenright}{\kern0pt}{\isacharparenright}{\kern0pt}{\isacartoucheclose}%
}%
\SNIP{inductive:Ax-2-0-0}{%
{\isacharparenleft}{\kern0pt}{\isacharprime}{\kern0pt}i\ {\isacharcolon}{\kern0pt}{\isacharcolon}{\kern0pt}\ countable{\isacharparenright}{\kern0pt}\ fm\ {\isasymRightarrow}\ bool%
}%
\SNIP{inductive:Ax-2-0}{%
\isacommand{inductive}\isamarkupfalse%
\ Ax{\isacharunderscore}{\kern0pt}{\isadigit{2}}\ {\isacharcolon}{\kern0pt}{\isacharcolon}{\kern0pt}\ {\Cartouche{inductive:Ax-2}{0}{0}}\ \isakeyword{where}\ \isanewline%
}%
\SNIP{inductive:Ax-2-1-0}{%
Ax{\isacharunderscore}{\kern0pt}{\isadigit{2}}\ {\isacharparenleft}{\kern0pt}\isactrlbold {\isasymnot}\ K\ i\ {\isacharparenleft}{\kern0pt}\isactrlbold {\isasymnot}\ K\ i\ p{\isacharparenright}{\kern0pt}\ \isactrlbold {\isasymlongrightarrow}\ K\ i\ {\isacharparenleft}{\kern0pt}\isactrlbold {\isasymnot}\ K\ i\ {\isacharparenleft}{\kern0pt}\isactrlbold {\isasymnot}\ p{\isacharparenright}{\kern0pt}{\isacharparenright}{\kern0pt}{\isacharparenright}{\kern0pt}%
}%
\SNIP{inductive:Ax-2-1}{%
\ \ {\Cartouche{inductive:Ax-2}{1}{0}}%
\isadelimdocument
\endisadelimdocument
\isatagdocument%
}%
\SNIP{theorem:weakly-directed-0}{%
\isacommand{theorem}\isamarkupfalse%
\ weakly{\isacharunderscore}{\kern0pt}directed{\isacharcolon}{\kern0pt}\isanewline%
}%
\SNIP{theorem:weakly-directed-1-0}{%
weakly{\isacharunderscore}{\kern0pt}directed\ M%
}%
\SNIP{theorem:weakly-directed-1-1}{%
w\ {\isasymin}\ {\isasymW}\ M%
}%
\SNIP{theorem:weakly-directed-1}{%
\ \ \isakeyword{assumes}\ {\Cartouche{theorem:weakly-directed}{1}{0}}\ {\Cartouche{theorem:weakly-directed}{1}{1}}\isanewline%
}%
\SNIP{theorem:weakly-directed-2-0}{%
M{\isacharcomma}{\kern0pt}\ w\ {\isasymTurnstile}\ {\isacharparenleft}{\kern0pt}L\ i\ {\isacharparenleft}{\kern0pt}K\ i\ p{\isacharparenright}{\kern0pt}\ \isactrlbold {\isasymlongrightarrow}\ K\ i\ {\isacharparenleft}{\kern0pt}L\ i\ p{\isacharparenright}{\kern0pt}{\isacharparenright}{\kern0pt}%
}%
\SNIP{theorem:weakly-directed-2}{%
\ \ \isakeyword{shows}\ {\Cartouche{theorem:weakly-directed}{2}{0}}\isanewline%
}%
\SNIP{theorem:weakly-directed-3}{%
\isadelimproof
\endisadelimproof
\isatagproof
\isacommand{proof}\isamarkupfalse%
\ \isanewline%
}%
\SNIP{theorem:weakly-directed-4-0}{%
M{\isacharcomma}{\kern0pt}\ w\ {\isasymTurnstile}\ {\isacharparenleft}{\kern0pt}L\ i\ {\isacharparenleft}{\kern0pt}K\ i\ p{\isacharparenright}{\kern0pt}{\isacharparenright}{\kern0pt}%
}%
\SNIP{theorem:weakly-directed-4}{%
\ \ \isacommand{assume}\isamarkupfalse%
\ {\Cartouche{theorem:weakly-directed}{4}{0}}\isanewline%
}%
\SNIP{theorem:weakly-directed-5-0}{%
{\isasymexists}v\ {\isasymin}\ {\isasymW}\ M\ {\isasyminter}\ {\isasymK}\ M\ i\ w{\isachardot}{\kern0pt}\ M{\isacharcomma}{\kern0pt}\ v\ {\isasymTurnstile}\ K\ i\ p%
}%
\SNIP{theorem:weakly-directed-5}{%
\ \ \isacommand{then}\isamarkupfalse%
\ \isacommand{have}\isamarkupfalse%
\ {\Cartouche{theorem:weakly-directed}{5}{0}}\ \isanewline%
}%
\SNIP{theorem:weakly-directed-6}{%
\ \ \ \ \isacommand{by}\isamarkupfalse%
\ simp\isanewline%
}%
\SNIP{theorem:weakly-directed-7-0}{%
{\isasymforall}v\ {\isasymin}\ {\isasymW}\ M\ {\isasyminter}\ {\isasymK}\ M\ i\ w{\isachardot}{\kern0pt}\ {\isasymexists}u\ {\isasymin}\ {\isasymW}\ M\ {\isasyminter}\ {\isasymK}\ M\ i\ v{\isachardot}{\kern0pt}\ M{\isacharcomma}{\kern0pt}\ u\ {\isasymTurnstile}\ p%
}%
\SNIP{theorem:weakly-directed-7}{%
\ \ \isacommand{then}\isamarkupfalse%
\ \isacommand{have}\isamarkupfalse%
\ {\Cartouche{theorem:weakly-directed}{7}{0}}\isanewline%
}%
\SNIP{theorem:weakly-directed-8-0}{%
weakly{\isacharunderscore}{\kern0pt}directed\ M%
}%
\SNIP{theorem:weakly-directed-8-1}{%
w\ {\isasymin}\ {\isasymW}\ M%
}%
\SNIP{theorem:weakly-directed-8}{%
\ \ \ \ \isacommand{using}\isamarkupfalse%
\ {\Cartouche{theorem:weakly-directed}{8}{0}}\ {\Cartouche{theorem:weakly-directed}{8}{1}}\ \isacommand{unfolding}\isamarkupfalse%
\ weakly{\isacharunderscore}{\kern0pt}directed{\isacharunderscore}{\kern0pt}def\isanewline%
}%
\SNIP{theorem:weakly-directed-9}{%
\ \ \ \ \isacommand{by}\isamarkupfalse%
\ {\isacharparenleft}{\kern0pt}metis\ IntE\ IntI\ semantics{\isachardot}{\kern0pt}simps{\isacharparenleft}{\kern0pt}{\isadigit{6}}{\isacharparenright}{\kern0pt}{\isacharparenright}{\kern0pt}\ \isanewline%
}%
\SNIP{theorem:weakly-directed-10-0}{%
{\isasymforall}v\ {\isasymin}\ {\isasymW}\ M\ {\isasyminter}\ {\isasymK}\ M\ i\ w{\isachardot}{\kern0pt}\ M{\isacharcomma}{\kern0pt}\ v\ {\isasymTurnstile}\ L\ i\ p%
}%
\SNIP{theorem:weakly-directed-10}{%
\ \ \isacommand{then}\isamarkupfalse%
\ \isacommand{have}\isamarkupfalse%
\ {\Cartouche{theorem:weakly-directed}{10}{0}}\isanewline%
}%
\SNIP{theorem:weakly-directed-11}{%
\ \ \ \ \isacommand{by}\isamarkupfalse%
\ simp\isanewline%
}%
\SNIP{theorem:weakly-directed-12-0}{%
M{\isacharcomma}{\kern0pt}\ w\ {\isasymTurnstile}\ K\ i\ {\isacharparenleft}{\kern0pt}L\ i\ p{\isacharparenright}{\kern0pt}%
}%
\SNIP{theorem:weakly-directed-12}{%
\ \ \isacommand{then}\isamarkupfalse%
\ \isacommand{show}\isamarkupfalse%
\ {\Cartouche{theorem:weakly-directed}{12}{0}}\isanewline%
}%
\SNIP{theorem:weakly-directed-13}{%
\ \ \ \ \isacommand{by}\isamarkupfalse%
\ simp\isanewline%
}%
\SNIP{theorem:weakly-directed-14}{%
\isacommand{qed}\isamarkupfalse%
\endisatagproof
{\isafoldproof}%
\isadelimproof%
}%
\SNIP{lemma:soundness-Ax-2-0-0}{%
Ax{\isacharunderscore}{\kern0pt}{\isadigit{2}}\ p\ {\isasymLongrightarrow}\ weakly{\isacharunderscore}{\kern0pt}directed\ M\ {\isasymLongrightarrow}\ w\ {\isasymin}\ {\isasymW}\ M\ {\isasymLongrightarrow}\ M{\isacharcomma}{\kern0pt}\ w\ {\isasymTurnstile}\ p%
}%
\SNIP{lemma:soundness-Ax-2-0}{%
\isacommand{lemma}\isamarkupfalse%
\ soundness{\isacharunderscore}{\kern0pt}Ax{\isacharunderscore}{\kern0pt}{\isadigit{2}}{\isacharcolon}{\kern0pt}\ {\Cartouche{lemma:soundness-Ax-2}{0}{0}}\isanewline%
}%
\SNIP{lemma:soundness-Ax-2-1}{%
\isadelimproof
\ \ %
\endisadelimproof
\isatagproof
\isacommand{by}\isamarkupfalse%
\ {\isacharparenleft}{\kern0pt}induct\ p\ rule{\isacharcolon}{\kern0pt}\ Ax{\isacharunderscore}{\kern0pt}{\isadigit{2}}{\isachardot}{\kern0pt}induct{\isacharparenright}{\kern0pt}\ {\isacharparenleft}{\kern0pt}meson\ weakly{\isacharunderscore}{\kern0pt}directed{\isacharparenright}{\kern0pt}%
\endisatagproof
{\isafoldproof}%
\isadelimproof
\endisadelimproof
\isadelimdocument
\endisadelimdocument
\isatagdocument%
}%
\SNIP{lemma:Ax-2-weakly-directed-0}{%
\isacommand{lemma}\isamarkupfalse%
\ Ax{\isacharunderscore}{\kern0pt}{\isadigit{2}}{\isacharunderscore}{\kern0pt}weakly{\isacharunderscore}{\kern0pt}directed{\isacharcolon}{\kern0pt}\isanewline%
}%
\SNIP{lemma:Ax-2-weakly-directed-1-0}{%
{\isacharparenleft}{\kern0pt}{\isacharparenleft}{\kern0pt}{\isacharprime}{\kern0pt}i\ {\isacharcolon}{\kern0pt}{\isacharcolon}{\kern0pt}\ countable{\isacharparenright}{\kern0pt}\ fm\ {\isasymRightarrow}\ bool{\isacharparenright}{\kern0pt}%
}%
\SNIP{lemma:Ax-2-weakly-directed-1}{%
\ \ \isakeyword{fixes}\ A\ {\isacharcolon}{\kern0pt}{\isacharcolon}{\kern0pt}\ {\Cartouche{lemma:Ax-2-weakly-directed}{1}{0}}\isanewline%
}%
\SNIP{lemma:Ax-2-weakly-directed-2-0}{%
{\isasymforall}p{\isachardot}{\kern0pt}\ Ax{\isacharunderscore}{\kern0pt}{\isadigit{2}}\ p\ {\isasymlongrightarrow}\ A\ p%
}%
\SNIP{lemma:Ax-2-weakly-directed-2-1}{%
consistent\ A\ V%
}%
\SNIP{lemma:Ax-2-weakly-directed-2-2}{%
maximal\ A\ V%
}%
\SNIP{lemma:Ax-2-weakly-directed-2}{%
\ \ \isakeyword{assumes}\ {\Cartouche{lemma:Ax-2-weakly-directed}{2}{0}}\ {\Cartouche{lemma:Ax-2-weakly-directed}{2}{1}}\ {\Cartouche{lemma:Ax-2-weakly-directed}{2}{2}}\isanewline%
}%
\SNIP{lemma:Ax-2-weakly-directed-3-0}{%
consistent\ A\ W%
}%
\SNIP{lemma:Ax-2-weakly-directed-3-1}{%
maximal\ A\ W%
}%
\SNIP{lemma:Ax-2-weakly-directed-3-2}{%
consistent\ A\ U%
}%
\SNIP{lemma:Ax-2-weakly-directed-3-3}{%
maximal\ A\ U%
}%
\SNIP{lemma:Ax-2-weakly-directed-3}{%
\ \ \ \ \isakeyword{and}\ {\Cartouche{lemma:Ax-2-weakly-directed}{3}{0}}\ {\Cartouche{lemma:Ax-2-weakly-directed}{3}{1}}\ {\Cartouche{lemma:Ax-2-weakly-directed}{3}{2}}\ {\Cartouche{lemma:Ax-2-weakly-directed}{3}{3}}\isanewline%
}%
\SNIP{lemma:Ax-2-weakly-directed-4-0}{%
W\ {\isasymin}\ reach\ A\ i\ V%
}%
\SNIP{lemma:Ax-2-weakly-directed-4-1}{%
U\ {\isasymin}\ reach\ A\ i\ V%
}%
\SNIP{lemma:Ax-2-weakly-directed-4}{%
\ \ \ \ \isakeyword{and}\ {\Cartouche{lemma:Ax-2-weakly-directed}{4}{0}}\ {\Cartouche{lemma:Ax-2-weakly-directed}{4}{1}}\isanewline%
}%
\SNIP{lemma:Ax-2-weakly-directed-5-0}{%
{\isasymexists}X{\isachardot}{\kern0pt}\ {\isacharparenleft}{\kern0pt}consistent\ A\ X{\isacharparenright}{\kern0pt}\ {\isasymand}\ {\isacharparenleft}{\kern0pt}maximal\ A\ X{\isacharparenright}{\kern0pt}\ {\isasymand}\ X\ {\isasymin}\ {\isacharparenleft}{\kern0pt}reach\ A\ i\ W{\isacharparenright}{\kern0pt}\ {\isasyminter}\ {\isacharparenleft}{\kern0pt}reach\ A\ i\ U{\isacharparenright}{\kern0pt}%
}%
\SNIP{lemma:Ax-2-weakly-directed-5}{%
\ \ \isakeyword{shows}\ {\Cartouche{lemma:Ax-2-weakly-directed}{5}{0}}\isanewline%
}%
\SNIP{lemma:Ax-2-weakly-directed-6}{%
\isadelimproof
\endisadelimproof
\isatagproof
\isacommand{proof}\isamarkupfalse%
\ {\isacharparenleft}{\kern0pt}rule\ ccontr{\isacharparenright}{\kern0pt}\isanewline%
}%
\SNIP{lemma:Ax-2-weakly-directed-7-0}{%
{\isasymnot}\ {\isacharquery}{\kern0pt}thesis%
}%
\SNIP{lemma:Ax-2-weakly-directed-7}{%
\ \ \isacommand{assume}\isamarkupfalse%
\ {\Cartouche{lemma:Ax-2-weakly-directed}{7}{0}}\isanewline%
}%
\SNIP{lemma:Ax-2-weakly-directed-8-0}{%
{\isacharparenleft}{\kern0pt}known\ W\ i{\isacharparenright}{\kern0pt}\ {\isasymunion}\ {\isacharparenleft}{\kern0pt}known\ U\ i{\isacharparenright}{\kern0pt}%
}%
\SNIP{lemma:Ax-2-weakly-directed-8}{%
\ \ \isacommand{let}\isamarkupfalse%
\ {\isacharquery}{\kern0pt}S\ {\isacharequal}{\kern0pt}\ {\Cartouche{lemma:Ax-2-weakly-directed}{8}{0}}\isanewline%
}%
\SNIP{lemma:Ax-2-weakly-directed-9-0}{%
{\isasymnot}\ consistent\ A\ {\isacharquery}{\kern0pt}S%
}%
\SNIP{lemma:Ax-2-weakly-directed-9}{%
\ \ \isacommand{have}\isamarkupfalse%
\ {\Cartouche{lemma:Ax-2-weakly-directed}{9}{0}}\ \isanewline%
}%
\SNIP{lemma:Ax-2-weakly-directed-10-0}{%
{\isasymnexists}X{\isachardot}{\kern0pt}\ consistent\ A\ X\ {\isasymand}\ maximal\ A\ X\ {\isasymand}\ X\ {\isasymin}\ {\isacharbraceleft}{\kern0pt}Wa{\isachardot}{\kern0pt}\ known\ W\ i\ {\isasymsubseteq}\ Wa{\isacharbraceright}{\kern0pt}\ {\isasyminter}\ {\isacharbraceleft}{\kern0pt}W{\isachardot}{\kern0pt}\ known\ U\ i\ {\isasymsubseteq}\ W{\isacharbraceright}{\kern0pt}%
}%
\SNIP{lemma:Ax-2-weakly-directed-10}{%
\ \ \ \ \isacommand{by}\isamarkupfalse%
\ {\isacharparenleft}{\kern0pt}smt\ {\isacharparenleft}{\kern0pt}verit{\isacharcomma}{\kern0pt}\ best{\isacharparenright}{\kern0pt}\ Int{\isacharunderscore}{\kern0pt}Collect\ {\Cartouche{lemma:Ax-2-weakly-directed}{10}{0}}\ maximal{\isacharunderscore}{\kern0pt}extension\ mem{\isacharunderscore}{\kern0pt}Collect{\isacharunderscore}{\kern0pt}eq\ sup{\isachardot}{\kern0pt}bounded{\isacharunderscore}{\kern0pt}iff{\isacharparenright}{\kern0pt}\isanewline%
}%
\SNIP{lemma:Ax-2-weakly-directed-11-0}{%
A\ {\isasymturnstile}\ imply\ S{\isacharprime}{\kern0pt}\ \isactrlbold {\isasymbottom}%
}%
\SNIP{lemma:Ax-2-weakly-directed-11-1}{%
set\ S{\isacharprime}{\kern0pt}\ {\isasymsubseteq}\ {\isacharquery}{\kern0pt}S%
}%
\SNIP{lemma:Ax-2-weakly-directed-11-2}{%
finite\ {\isacharparenleft}{\kern0pt}set\ S{\isacharprime}{\kern0pt}{\isacharparenright}{\kern0pt}%
}%
\SNIP{lemma:Ax-2-weakly-directed-11}{%
\ \ \isacommand{then}\isamarkupfalse%
\ \isacommand{obtain}\isamarkupfalse%
\ S{\isacharprime}{\kern0pt}\ \isakeyword{where}\ {\isacharasterisk}{\kern0pt}{\isacharcolon}{\kern0pt}\ {\Cartouche{lemma:Ax-2-weakly-directed}{11}{0}}\ {\Cartouche{lemma:Ax-2-weakly-directed}{11}{1}}\ {\Cartouche{lemma:Ax-2-weakly-directed}{11}{2}}\isanewline%
}%
\SNIP{lemma:Ax-2-weakly-directed-12}{%
\ \ \ \ \isacommand{unfolding}\isamarkupfalse%
\ consistent{\isacharunderscore}{\kern0pt}def\ \isacommand{by}\isamarkupfalse%
\ blast\ \isanewline%
}%
\SNIP{lemma:Ax-2-weakly-directed-13-0}{%
filter\ {\isacharparenleft}{\kern0pt}{\isasymlambda}p{\isachardot}{\kern0pt}\ p\ {\isasymin}\ {\isacharparenleft}{\kern0pt}known\ U\ i{\isacharparenright}{\kern0pt}{\isacharparenright}{\kern0pt}\ S{\isacharprime}{\kern0pt}%
}%
\SNIP{lemma:Ax-2-weakly-directed-13}{%
\ \ \isacommand{let}\isamarkupfalse%
\ {\isacharquery}{\kern0pt}U\ {\isacharequal}{\kern0pt}\ {\Cartouche{lemma:Ax-2-weakly-directed}{13}{0}}\isanewline%
}%
\SNIP{lemma:Ax-2-weakly-directed-14-0}{%
filter\ {\isacharparenleft}{\kern0pt}{\isasymlambda}p{\isachardot}{\kern0pt}\ p\ {\isasymin}\ {\isacharparenleft}{\kern0pt}known\ W\ i{\isacharparenright}{\kern0pt}{\isacharparenright}{\kern0pt}\ S{\isacharprime}{\kern0pt}%
}%
\SNIP{lemma:Ax-2-weakly-directed-14}{%
\ \ \isacommand{let}\isamarkupfalse%
\ {\isacharquery}{\kern0pt}W\ {\isacharequal}{\kern0pt}\ {\Cartouche{lemma:Ax-2-weakly-directed}{14}{0}}\isanewline%
}%
\SNIP{lemma:Ax-2-weakly-directed-15-0}{%
conjunct\ {\isacharquery}{\kern0pt}U%
}%
\SNIP{lemma:Ax-2-weakly-directed-15-1}{%
conjunct\ {\isacharquery}{\kern0pt}W%
}%
\SNIP{lemma:Ax-2-weakly-directed-15}{%
\ \ \isacommand{let}\isamarkupfalse%
\ {\isacharquery}{\kern0pt}p\ {\isacharequal}{\kern0pt}\ {\Cartouche{lemma:Ax-2-weakly-directed}{15}{0}}\ \isakeyword{and}\ {\isacharquery}{\kern0pt}q\ {\isacharequal}{\kern0pt}\ {\Cartouche{lemma:Ax-2-weakly-directed}{15}{1}}\isanewline%
}%
\SNIP{lemma:Ax-2-weakly-directed-16-0}{%
{\isacharparenleft}{\kern0pt}set\ {\isacharquery}{\kern0pt}U{\isacharparenright}{\kern0pt}\ {\isasymunion}\ {\isacharparenleft}{\kern0pt}set\ {\isacharquery}{\kern0pt}W{\isacharparenright}{\kern0pt}\ {\isacharequal}{\kern0pt}\ {\isacharparenleft}{\kern0pt}set\ S{\isacharprime}{\kern0pt}{\isacharparenright}{\kern0pt}%
}%
\SNIP{lemma:Ax-2-weakly-directed-16}{%
\ \ \isacommand{have}\isamarkupfalse%
\ {\Cartouche{lemma:Ax-2-weakly-directed}{16}{0}}\isanewline%
}%
\SNIP{lemma:Ax-2-weakly-directed-17}{%
\ \ \ \ \isacommand{using}\isamarkupfalse%
\ {\isacharasterisk}{\kern0pt}\ \isacommand{by}\isamarkupfalse%
\ auto\isanewline%
}%
\SNIP{lemma:Ax-2-weakly-directed-18-0}{%
A\ {\isasymturnstile}\ imply\ {\isacharquery}{\kern0pt}U\ {\isacharparenleft}{\kern0pt}imply\ {\isacharquery}{\kern0pt}W\ \isactrlbold {\isasymbottom}{\isacharparenright}{\kern0pt}%
}%
\SNIP{lemma:Ax-2-weakly-directed-18}{%
\ \ \isacommand{then}\isamarkupfalse%
\ \isacommand{have}\isamarkupfalse%
\ {\Cartouche{lemma:Ax-2-weakly-directed}{18}{0}}\isanewline%
}%
\SNIP{lemma:Ax-2-weakly-directed-19}{%
\ \ \ \ \isacommand{using}\isamarkupfalse%
\ K{\isacharunderscore}{\kern0pt}imply{\isacharunderscore}{\kern0pt}weaken\ imply{\isacharunderscore}{\kern0pt}append\ \ \isanewline%
}%
\SNIP{lemma:Ax-2-weakly-directed-20}{%
\ \ \ \ \isacommand{by}\isamarkupfalse%
\ {\isacharparenleft}{\kern0pt}metis\ {\isacharparenleft}{\kern0pt}mono{\isacharunderscore}{\kern0pt}tags{\isacharcomma}{\kern0pt}\ lifting{\isacharparenright}{\kern0pt}\ {\isachardoublequoteopen}{\isacharasterisk}{\kern0pt}{\isachardoublequoteclose}{\isacharparenleft}{\kern0pt}{\isadigit{1}}{\isacharparenright}{\kern0pt}\ set{\isacharunderscore}{\kern0pt}append\ subset{\isacharunderscore}{\kern0pt}refl{\isacharparenright}{\kern0pt}\ \isanewline%
}%
\SNIP{lemma:Ax-2-weakly-directed-21-0}{%
A\ {\isasymturnstile}\ {\isacharparenleft}{\kern0pt}{\isacharquery}{\kern0pt}p\ \isactrlbold {\isasymlongrightarrow}\ {\isacharparenleft}{\kern0pt}imply\ {\isacharquery}{\kern0pt}W\ \isactrlbold {\isasymbottom}{\isacharparenright}{\kern0pt}{\isacharparenright}{\kern0pt}%
}%
\SNIP{lemma:Ax-2-weakly-directed-21}{%
\ \ \isacommand{then}\isamarkupfalse%
\ \isacommand{have}\isamarkupfalse%
\ {\Cartouche{lemma:Ax-2-weakly-directed}{21}{0}}\isanewline%
}%
\SNIP{lemma:Ax-2-weakly-directed-22}{%
\ \ \ \ \isacommand{using}\isamarkupfalse%
\ K{\isacharunderscore}{\kern0pt}imply{\isacharunderscore}{\kern0pt}conjunct\ \isacommand{by}\isamarkupfalse%
\ blast\isanewline%
}%
\SNIP{lemma:Ax-2-weakly-directed-23-0}{%
tautology\ {\isacharparenleft}{\kern0pt}{\isacharparenleft}{\kern0pt}imply\ {\isacharquery}{\kern0pt}W\ \isactrlbold {\isasymbottom}{\isacharparenright}{\kern0pt}\ \isactrlbold {\isasymlongrightarrow}\ {\isacharparenleft}{\kern0pt}{\isacharquery}{\kern0pt}q\ \isactrlbold {\isasymlongrightarrow}\ \isactrlbold {\isasymbottom}{\isacharparenright}{\kern0pt}{\isacharparenright}{\kern0pt}%
}%
\SNIP{lemma:Ax-2-weakly-directed-23}{%
\ \ \isacommand{then}\isamarkupfalse%
\ \isacommand{have}\isamarkupfalse%
\ {\Cartouche{lemma:Ax-2-weakly-directed}{23}{0}}\isanewline%
}%
\SNIP{lemma:Ax-2-weakly-directed-24}{%
\ \ \ \ \isacommand{using}\isamarkupfalse%
\ imply{\isacharunderscore}{\kern0pt}conjunct\ \isacommand{by}\isamarkupfalse%
\ blast\isanewline%
}%
\SNIP{lemma:Ax-2-weakly-directed-25-0}{%
A\ {\isasymturnstile}\ {\isacharparenleft}{\kern0pt}{\isacharparenleft}{\kern0pt}imply\ {\isacharquery}{\kern0pt}W\ \isactrlbold {\isasymbottom}{\isacharparenright}{\kern0pt}\ \isactrlbold {\isasymlongrightarrow}\ {\isacharparenleft}{\kern0pt}{\isacharquery}{\kern0pt}q\ \isactrlbold {\isasymlongrightarrow}\ \isactrlbold {\isasymbottom}{\isacharparenright}{\kern0pt}{\isacharparenright}{\kern0pt}%
}%
\SNIP{lemma:Ax-2-weakly-directed-25}{%
\ \ \isacommand{then}\isamarkupfalse%
\ \isacommand{have}\isamarkupfalse%
\ {\Cartouche{lemma:Ax-2-weakly-directed}{25}{0}}\isanewline%
}%
\SNIP{lemma:Ax-2-weakly-directed-26}{%
\ \ \ \ \isacommand{using}\isamarkupfalse%
\ A{\isadigit{1}}\ \isacommand{by}\isamarkupfalse%
\ blast\isanewline%
}%
\SNIP{lemma:Ax-2-weakly-directed-27-0}{%
A\ {\isasymturnstile}\ {\isacharparenleft}{\kern0pt}{\isacharquery}{\kern0pt}p\ \isactrlbold {\isasymlongrightarrow}\ {\isacharparenleft}{\kern0pt}{\isacharquery}{\kern0pt}q\ \isactrlbold {\isasymlongrightarrow}\ \isactrlbold {\isasymbottom}{\isacharparenright}{\kern0pt}{\isacharparenright}{\kern0pt}%
}%
\SNIP{lemma:Ax-2-weakly-directed-27}{%
\ \ \isacommand{then}\isamarkupfalse%
\ \isacommand{have}\isamarkupfalse%
\ {\Cartouche{lemma:Ax-2-weakly-directed}{27}{0}}\isanewline%
}%
\SNIP{lemma:Ax-2-weakly-directed-28-0}{%
A\ {\isasymturnstile}\ {\isacharparenleft}{\kern0pt}conjunct\ {\isacharparenleft}{\kern0pt}filter\ {\isacharparenleft}{\kern0pt}{\isasymlambda}p{\isachardot}{\kern0pt}\ p\ {\isasymin}\ known\ U\ i{\isacharparenright}{\kern0pt}\ S{\isacharprime}{\kern0pt}{\isacharparenright}{\kern0pt}\ \isactrlbold {\isasymlongrightarrow}\ imply\ {\isacharparenleft}{\kern0pt}filter\ {\isacharparenleft}{\kern0pt}{\isasymlambda}p{\isachardot}{\kern0pt}\ p\ {\isasymin}\ known\ W\ i{\isacharparenright}{\kern0pt}\ S{\isacharprime}{\kern0pt}{\isacharparenright}{\kern0pt}\ \isactrlbold {\isasymbottom}{\isacharparenright}{\kern0pt}%
}%
\SNIP{lemma:Ax-2-weakly-directed-28}{%
\ \ \ \ \isacommand{using}\isamarkupfalse%
\ K{\isacharunderscore}{\kern0pt}imp{\isacharunderscore}{\kern0pt}trans\ {\Cartouche{lemma:Ax-2-weakly-directed}{28}{0}}\isanewline%
}%
\SNIP{lemma:Ax-2-weakly-directed-29}{%
\ \ \ \ \isacommand{by}\isamarkupfalse%
\ blast\ \isanewline%
}%
\SNIP{lemma:Ax-2-weakly-directed-30-0}{%
A\ {\isasymturnstile}\ {\isacharparenleft}{\kern0pt}{\isacharparenleft}{\kern0pt}{\isacharquery}{\kern0pt}p\ \isactrlbold {\isasymand}\ {\isacharquery}{\kern0pt}q{\isacharparenright}{\kern0pt}\ \isactrlbold {\isasymlongrightarrow}\ \isactrlbold {\isasymbottom}{\isacharparenright}{\kern0pt}%
}%
\SNIP{lemma:Ax-2-weakly-directed-30}{%
\ \ \isacommand{then}\isamarkupfalse%
\ \isacommand{have}\isamarkupfalse%
\ o{\isadigit{1}}{\isacharcolon}{\kern0pt}\ {\Cartouche{lemma:Ax-2-weakly-directed}{30}{0}}\isanewline%
}%
\SNIP{lemma:Ax-2-weakly-directed-31}{%
\ \ \ \ \isacommand{by}\isamarkupfalse%
\ {\isacharparenleft}{\kern0pt}meson\ K{\isacharunderscore}{\kern0pt}multi{\isacharunderscore}{\kern0pt}imply{\isacharparenright}{\kern0pt}\isanewline%
}%
\SNIP{lemma:Ax-2-weakly-directed-32-0}{%
set\ {\isacharquery}{\kern0pt}U\ {\isasymsubseteq}\ {\isacharparenleft}{\kern0pt}known\ U\ i{\isacharparenright}{\kern0pt}%
}%
\SNIP{lemma:Ax-2-weakly-directed-32-1}{%
set\ {\isacharquery}{\kern0pt}W\ {\isasymsubseteq}\ {\isacharparenleft}{\kern0pt}known\ W\ i{\isacharparenright}{\kern0pt}%
}%
\SNIP{lemma:Ax-2-weakly-directed-32}{%
\ \ \isacommand{moreover}\isamarkupfalse%
\ \isacommand{have}\isamarkupfalse%
\ {\Cartouche{lemma:Ax-2-weakly-directed}{32}{0}}\ \isakeyword{and}\ {\Cartouche{lemma:Ax-2-weakly-directed}{32}{1}}\ \isanewline%
}%
\SNIP{lemma:Ax-2-weakly-directed-33-0}{%
{\isasymforall}\ p{\isachardot}{\kern0pt}\ p\ {\isasymin}\ set\ {\isacharquery}{\kern0pt}U\ {\isasymlongrightarrow}\ {\isacharparenleft}{\kern0pt}K\ i\ p{\isacharparenright}{\kern0pt}\ {\isasymin}\ U%
}%
\SNIP{lemma:Ax-2-weakly-directed-33-1}{%
{\isasymforall}\ p{\isachardot}{\kern0pt}\ p\ {\isasymin}\ set\ {\isacharquery}{\kern0pt}W\ {\isasymlongrightarrow}\ {\isacharparenleft}{\kern0pt}K\ i\ p{\isacharparenright}{\kern0pt}\ {\isasymin}\ W%
}%
\SNIP{lemma:Ax-2-weakly-directed-33}{%
\ \ \ \ \ \ \ \ \ \ \ \ \ \ \ \ \isakeyword{and}\ {\Cartouche{lemma:Ax-2-weakly-directed}{33}{0}}\ \isakeyword{and}\ {\Cartouche{lemma:Ax-2-weakly-directed}{33}{1}}\ \isanewline%
}%
\SNIP{lemma:Ax-2-weakly-directed-34}{%
\ \ \ \ \isacommand{by}\isamarkupfalse%
\ auto\isanewline%
}%
\SNIP{lemma:Ax-2-weakly-directed-35-0}{%
set\ {\isacharparenleft}{\kern0pt}map\ {\isacharparenleft}{\kern0pt}K\ i{\isacharparenright}{\kern0pt}\ {\isacharquery}{\kern0pt}U{\isacharparenright}{\kern0pt}\ {\isasymsubseteq}\ U%
}%
\SNIP{lemma:Ax-2-weakly-directed-35-1}{%
set\ {\isacharparenleft}{\kern0pt}map\ {\isacharparenleft}{\kern0pt}K\ i{\isacharparenright}{\kern0pt}\ {\isacharquery}{\kern0pt}W{\isacharparenright}{\kern0pt}\ {\isasymsubseteq}\ W%
}%
\SNIP{lemma:Ax-2-weakly-directed-35}{%
\ \ \isacommand{then}\isamarkupfalse%
\ \isacommand{have}\isamarkupfalse%
\ {\Cartouche{lemma:Ax-2-weakly-directed}{35}{0}}\ \isakeyword{and}\ c{\isadigit{1}}{\isacharcolon}{\kern0pt}\ {\Cartouche{lemma:Ax-2-weakly-directed}{35}{1}}\isanewline%
}%
\SNIP{lemma:Ax-2-weakly-directed-36}{%
\ \ \ \ \ \isacommand{apply}\isamarkupfalse%
\ {\isacharparenleft}{\kern0pt}metis\ {\isacharparenleft}{\kern0pt}mono{\isacharunderscore}{\kern0pt}tags{\isacharcomma}{\kern0pt}\ lifting{\isacharparenright}{\kern0pt}\ imageE\ set{\isacharunderscore}{\kern0pt}map\ subsetI{\isacharparenright}{\kern0pt}\ \isanewline%
}%
\SNIP{lemma:Ax-2-weakly-directed-37}{%
\ \ \ \ \isacommand{by}\isamarkupfalse%
\ auto\isanewline%
}%
\SNIP{lemma:Ax-2-weakly-directed-38-0}{%
conjunct\ {\isacharparenleft}{\kern0pt}map\ {\isacharparenleft}{\kern0pt}K\ i{\isacharparenright}{\kern0pt}\ {\isacharquery}{\kern0pt}U{\isacharparenright}{\kern0pt}\ {\isasymin}\ U%
}%
\SNIP{lemma:Ax-2-weakly-directed-38-1}{%
conjunct\ {\isacharparenleft}{\kern0pt}map\ {\isacharparenleft}{\kern0pt}K\ i{\isacharparenright}{\kern0pt}\ {\isacharquery}{\kern0pt}W{\isacharparenright}{\kern0pt}\ {\isasymin}\ W%
}%
\SNIP{lemma:Ax-2-weakly-directed-38}{%
\ \ \isacommand{then}\isamarkupfalse%
\ \isacommand{have}\isamarkupfalse%
\ c{\isadigit{2}}{\isacharcolon}{\kern0pt}\ {\Cartouche{lemma:Ax-2-weakly-directed}{38}{0}}\ \isakeyword{and}\ c{\isadigit{2}}{\isacharprime}{\kern0pt}{\isacharcolon}{\kern0pt}\ {\Cartouche{lemma:Ax-2-weakly-directed}{38}{1}}\isanewline%
}%
\SNIP{lemma:Ax-2-weakly-directed-39}{%
\ \ \ \ \isacommand{using}\isamarkupfalse%
\ assms{\isacharparenleft}{\kern0pt}{\isadigit{6}}{\isacharparenright}{\kern0pt}\ assms{\isacharparenleft}{\kern0pt}{\isadigit{7}}{\isacharparenright}{\kern0pt}\ mcs{\isacharunderscore}{\kern0pt}conjunction{\isacharunderscore}{\kern0pt}mult\ \isacommand{apply}\isamarkupfalse%
\ blast\isanewline%
}%
\SNIP{lemma:Ax-2-weakly-directed-40}{%
\ \ \ \ \isacommand{using}\isamarkupfalse%
\ assms{\isacharparenleft}{\kern0pt}{\isadigit{4}}{\isacharparenright}{\kern0pt}\ assms{\isacharparenleft}{\kern0pt}{\isadigit{5}}{\isacharparenright}{\kern0pt}\ c{\isadigit{1}}\ mcs{\isacharunderscore}{\kern0pt}conjunction{\isacharunderscore}{\kern0pt}mult\ \isacommand{by}\isamarkupfalse%
\ blast\isanewline%
}%
\SNIP{lemma:Ax-2-weakly-directed-41-0}{%
{\isacharparenleft}{\kern0pt}{\isacharparenleft}{\kern0pt}conjunct\ {\isacharparenleft}{\kern0pt}map\ {\isacharparenleft}{\kern0pt}K\ i{\isacharparenright}{\kern0pt}\ {\isacharquery}{\kern0pt}U{\isacharparenright}{\kern0pt}{\isacharparenright}{\kern0pt}\ \isactrlbold {\isasymlongrightarrow}\ {\isacharparenleft}{\kern0pt}K\ i\ {\isacharquery}{\kern0pt}p{\isacharparenright}{\kern0pt}{\isacharparenright}{\kern0pt}\ {\isasymin}\ U%
}%
\SNIP{lemma:Ax-2-weakly-directed-41}{%
\ \ \isacommand{then}\isamarkupfalse%
\ \isacommand{have}\isamarkupfalse%
\ {\Cartouche{lemma:Ax-2-weakly-directed}{41}{0}}\ \isanewline%
}%
\SNIP{lemma:Ax-2-weakly-directed-42-0}{%
{\isacharparenleft}{\kern0pt}{\isacharparenleft}{\kern0pt}conjunct\ {\isacharparenleft}{\kern0pt}map\ {\isacharparenleft}{\kern0pt}K\ i{\isacharparenright}{\kern0pt}\ {\isacharquery}{\kern0pt}W{\isacharparenright}{\kern0pt}{\isacharparenright}{\kern0pt}\ \isactrlbold {\isasymlongrightarrow}\ {\isacharparenleft}{\kern0pt}K\ i\ {\isacharquery}{\kern0pt}q{\isacharparenright}{\kern0pt}{\isacharparenright}{\kern0pt}\ {\isasymin}\ W%
}%
\SNIP{lemma:Ax-2-weakly-directed-42}{%
\ \ \ \ \ \ \ \ \isakeyword{and}\ c{\isadigit{3}}{\isacharcolon}{\kern0pt}\ {\Cartouche{lemma:Ax-2-weakly-directed}{42}{0}}\ \isanewline%
}%
\SNIP{lemma:Ax-2-weakly-directed-43}{%
\ \ \ \ \isacommand{apply}\isamarkupfalse%
\ {\isacharparenleft}{\kern0pt}meson\ K{\isacharunderscore}{\kern0pt}conjunction{\isacharunderscore}{\kern0pt}out{\isacharunderscore}{\kern0pt}mult\ assms{\isacharparenleft}{\kern0pt}{\isadigit{6}}{\isacharparenright}{\kern0pt}\ assms{\isacharparenleft}{\kern0pt}{\isadigit{7}}{\isacharparenright}{\kern0pt}\ deriv{\isacharunderscore}{\kern0pt}in{\isacharunderscore}{\kern0pt}maximal{\isacharparenright}{\kern0pt}\isanewline%
}%
\SNIP{lemma:Ax-2-weakly-directed-44}{%
\ \ \ \ \isacommand{by}\isamarkupfalse%
\ {\isacharparenleft}{\kern0pt}meson\ K{\isacharunderscore}{\kern0pt}conjunction{\isacharunderscore}{\kern0pt}out{\isacharunderscore}{\kern0pt}mult\ assms{\isacharparenleft}{\kern0pt}{\isadigit{4}}{\isacharparenright}{\kern0pt}\ assms{\isacharparenleft}{\kern0pt}{\isadigit{5}}{\isacharparenright}{\kern0pt}\ deriv{\isacharunderscore}{\kern0pt}in{\isacharunderscore}{\kern0pt}maximal{\isacharparenright}{\kern0pt}\isanewline%
}%
\SNIP{lemma:Ax-2-weakly-directed-45-0}{%
{\isacharparenleft}{\kern0pt}K\ i\ {\isacharquery}{\kern0pt}p{\isacharparenright}{\kern0pt}\ {\isasymin}\ U%
}%
\SNIP{lemma:Ax-2-weakly-directed-45-1}{%
{\isacharparenleft}{\kern0pt}K\ i\ {\isacharquery}{\kern0pt}q{\isacharparenright}{\kern0pt}\ {\isasymin}\ W%
}%
\SNIP{lemma:Ax-2-weakly-directed-45}{%
\ \ \isacommand{then}\isamarkupfalse%
\ \isacommand{have}\isamarkupfalse%
\ c{\isadigit{4}}{\isacharcolon}{\kern0pt}\ {\Cartouche{lemma:Ax-2-weakly-directed}{45}{0}}\ \isakeyword{and}\ c{\isadigit{4}}{\isacharprime}{\kern0pt}{\isacharcolon}{\kern0pt}\ {\Cartouche{lemma:Ax-2-weakly-directed}{45}{1}}\isanewline%
}%
\SNIP{lemma:Ax-2-weakly-directed-46}{%
\ \ \ \ \isacommand{using}\isamarkupfalse%
\ assms{\isacharparenleft}{\kern0pt}{\isadigit{6}}{\isacharparenright}{\kern0pt}\ assms{\isacharparenleft}{\kern0pt}{\isadigit{7}}{\isacharparenright}{\kern0pt}\ c{\isadigit{2}}\ consequent{\isacharunderscore}{\kern0pt}in{\isacharunderscore}{\kern0pt}maximal\ \isacommand{apply}\isamarkupfalse%
\ blast\isanewline%
}%
\SNIP{lemma:Ax-2-weakly-directed-47}{%
\ \ \ \ \isacommand{using}\isamarkupfalse%
\ assms{\isacharparenleft}{\kern0pt}{\isadigit{4}}{\isacharparenright}{\kern0pt}\ assms{\isacharparenleft}{\kern0pt}{\isadigit{5}}{\isacharparenright}{\kern0pt}\ c{\isadigit{2}}{\isacharprime}{\kern0pt}\ c{\isadigit{3}}\ consequent{\isacharunderscore}{\kern0pt}in{\isacharunderscore}{\kern0pt}maximal\ \isacommand{by}\isamarkupfalse%
\ blast\isanewline%
}%
\SNIP{lemma:Ax-2-weakly-directed-48-0}{%
{\isacharparenleft}{\kern0pt}L\ i\ {\isacharparenleft}{\kern0pt}K\ i\ {\isacharquery}{\kern0pt}p{\isacharparenright}{\kern0pt}{\isacharparenright}{\kern0pt}\ {\isasymin}\ V%
}%
\SNIP{lemma:Ax-2-weakly-directed-48-1}{%
{\isacharparenleft}{\kern0pt}L\ i\ {\isacharparenleft}{\kern0pt}K\ i\ {\isacharquery}{\kern0pt}q{\isacharparenright}{\kern0pt}{\isacharparenright}{\kern0pt}\ {\isasymin}\ V%
}%
\SNIP{lemma:Ax-2-weakly-directed-48}{%
\ \ \isacommand{then}\isamarkupfalse%
\ \isacommand{have}\isamarkupfalse%
\ {\Cartouche{lemma:Ax-2-weakly-directed}{48}{0}}\ \isakeyword{and}\ c{\isadigit{5}}{\isacharcolon}{\kern0pt}\ {\Cartouche{lemma:Ax-2-weakly-directed}{48}{1}}\ \isanewline%
}%
\SNIP{lemma:Ax-2-weakly-directed-49}{%
\ \ \ \ \isacommand{using}\isamarkupfalse%
\ assms{\isacharparenleft}{\kern0pt}{\isadigit{2}}{\isacharparenright}{\kern0pt}\ assms{\isacharparenleft}{\kern0pt}{\isadigit{3}}{\isacharparenright}{\kern0pt}\ assms{\isacharparenleft}{\kern0pt}{\isadigit{6}}{\isacharparenright}{\kern0pt}\ assms{\isacharparenleft}{\kern0pt}{\isadigit{7}}{\isacharparenright}{\kern0pt}\ assms{\isacharparenleft}{\kern0pt}{\isadigit{9}}{\isacharparenright}{\kern0pt}\ exactly{\isacharunderscore}{\kern0pt}one{\isacharunderscore}{\kern0pt}in{\isacharunderscore}{\kern0pt}maximal\ \isacommand{apply}\isamarkupfalse%
\ blast\isanewline%
}%
\SNIP{lemma:Ax-2-weakly-directed-50}{%
\ \ \ \ \isacommand{using}\isamarkupfalse%
\ assms{\isacharparenleft}{\kern0pt}{\isadigit{2}}{\isacharparenright}{\kern0pt}\ assms{\isacharparenleft}{\kern0pt}{\isadigit{3}}{\isacharparenright}{\kern0pt}\ assms{\isacharparenleft}{\kern0pt}{\isadigit{4}}{\isacharparenright}{\kern0pt}\ assms{\isacharparenleft}{\kern0pt}{\isadigit{5}}{\isacharparenright}{\kern0pt}\ assms{\isacharparenleft}{\kern0pt}{\isadigit{8}}{\isacharparenright}{\kern0pt}\ c{\isadigit{4}}{\isacharprime}{\kern0pt}\ exactly{\isacharunderscore}{\kern0pt}one{\isacharunderscore}{\kern0pt}in{\isacharunderscore}{\kern0pt}maximal\ \isacommand{by}\isamarkupfalse%
\ blast\isanewline%
}%
\SNIP{lemma:Ax-2-weakly-directed-51-0}{%
{\isacharparenleft}{\kern0pt}K\ i\ {\isacharparenleft}{\kern0pt}L\ i\ {\isacharquery}{\kern0pt}p{\isacharparenright}{\kern0pt}{\isacharparenright}{\kern0pt}\ {\isasymin}\ V%
}%
\SNIP{lemma:Ax-2-weakly-directed-51}{%
\ \ \isacommand{then}\isamarkupfalse%
\ \isacommand{have}\isamarkupfalse%
\ {\Cartouche{lemma:Ax-2-weakly-directed}{51}{0}}\isanewline%
}%
\SNIP{lemma:Ax-2-weakly-directed-52}{%
\ \ \ \ \isacommand{by}\isamarkupfalse%
\ {\isacharparenleft}{\kern0pt}meson\ Ax{\isacharunderscore}{\kern0pt}{\isadigit{2}}{\isachardot}{\kern0pt}intros\ assms{\isacharparenleft}{\kern0pt}{\isadigit{1}}{\isacharparenright}{\kern0pt}\ assms{\isacharparenleft}{\kern0pt}{\isadigit{2}}{\isacharparenright}{\kern0pt}\ assms{\isacharparenleft}{\kern0pt}{\isadigit{3}}{\isacharparenright}{\kern0pt}\ ax{\isacharunderscore}{\kern0pt}in{\isacharunderscore}{\kern0pt}maximal\ consequent{\isacharunderscore}{\kern0pt}in{\isacharunderscore}{\kern0pt}maximal{\isacharparenright}{\kern0pt}\isanewline%
}%
\SNIP{lemma:Ax-2-weakly-directed-53-0}{%
{\isacharparenleft}{\kern0pt}{\isacharparenleft}{\kern0pt}K\ i\ {\isacharparenleft}{\kern0pt}L\ i\ {\isacharquery}{\kern0pt}p{\isacharparenright}{\kern0pt}{\isacharparenright}{\kern0pt}\ \isactrlbold {\isasymand}\ {\isacharparenleft}{\kern0pt}L\ i\ {\isacharparenleft}{\kern0pt}K\ i\ {\isacharquery}{\kern0pt}q{\isacharparenright}{\kern0pt}{\isacharparenright}{\kern0pt}{\isacharparenright}{\kern0pt}\ {\isasymin}\ V%
}%
\SNIP{lemma:Ax-2-weakly-directed-53}{%
\ \ \isacommand{then}\isamarkupfalse%
\ \isacommand{have}\isamarkupfalse%
\ {\Cartouche{lemma:Ax-2-weakly-directed}{53}{0}}\isanewline%
}%
\SNIP{lemma:Ax-2-weakly-directed-54}{%
\ \ \ \ \isacommand{using}\isamarkupfalse%
\ assms{\isacharparenleft}{\kern0pt}{\isadigit{2}}{\isacharparenright}{\kern0pt}\ assms{\isacharparenleft}{\kern0pt}{\isadigit{3}}{\isacharparenright}{\kern0pt}\ c{\isadigit{5}}\ mcs{\isacharunderscore}{\kern0pt}conjunction\ \isacommand{by}\isamarkupfalse%
\ blast\isanewline%
}%
\SNIP{lemma:Ax-2-weakly-directed-55-0}{%
{\isacharparenleft}{\kern0pt}L\ i\ {\isacharparenleft}{\kern0pt}{\isacharparenleft}{\kern0pt}L\ i\ {\isacharquery}{\kern0pt}p{\isacharparenright}{\kern0pt}\ \isactrlbold {\isasymand}\ K\ i\ {\isacharquery}{\kern0pt}q{\isacharparenright}{\kern0pt}{\isacharparenright}{\kern0pt}\ {\isasymin}\ V%
}%
\SNIP{lemma:Ax-2-weakly-directed-55}{%
\ \ \isacommand{then}\isamarkupfalse%
\ \isacommand{have}\isamarkupfalse%
\ {\Cartouche{lemma:Ax-2-weakly-directed}{55}{0}}\isanewline%
}%
\SNIP{lemma:Ax-2-weakly-directed-56}{%
\ \ \ \ \isacommand{by}\isamarkupfalse%
\ {\isacharparenleft}{\kern0pt}meson\ K{\isacharunderscore}{\kern0pt}thm\ assms{\isacharparenleft}{\kern0pt}{\isadigit{2}}{\isacharparenright}{\kern0pt}\ assms{\isacharparenleft}{\kern0pt}{\isadigit{3}}{\isacharparenright}{\kern0pt}\ consequent{\isacharunderscore}{\kern0pt}in{\isacharunderscore}{\kern0pt}maximal\ deriv{\isacharunderscore}{\kern0pt}in{\isacharunderscore}{\kern0pt}maximal{\isacharparenright}{\kern0pt}\isanewline%
}%
\SNIP{lemma:Ax-2-weakly-directed-57-0}{%
{\isacharparenleft}{\kern0pt}L\ i\ {\isacharparenleft}{\kern0pt}{\isacharparenleft}{\kern0pt}K\ i\ {\isacharquery}{\kern0pt}q{\isacharparenright}{\kern0pt}\ \isactrlbold {\isasymand}\ L\ i\ {\isacharquery}{\kern0pt}p{\isacharparenright}{\kern0pt}{\isacharparenright}{\kern0pt}\ {\isasymin}\ V%
}%
\SNIP{lemma:Ax-2-weakly-directed-57}{%
\ \ \isacommand{then}\isamarkupfalse%
\ \isacommand{have}\isamarkupfalse%
\ {\Cartouche{lemma:Ax-2-weakly-directed}{57}{0}}\isanewline%
}%
\SNIP{lemma:Ax-2-weakly-directed-58-0}{%
K\ i\ {\isacharparenleft}{\kern0pt}L\ i\ {\isacharparenleft}{\kern0pt}conjunct\ {\isacharparenleft}{\kern0pt}filter\ {\isacharparenleft}{\kern0pt}{\isasymlambda}p{\isachardot}{\kern0pt}\ p\ {\isasymin}\ known\ U\ i{\isacharparenright}{\kern0pt}\ S{\isacharprime}{\kern0pt}{\isacharparenright}{\kern0pt}{\isacharparenright}{\kern0pt}{\isacharparenright}{\kern0pt}\ {\isasymin}\ V%
}%
\SNIP{lemma:Ax-2-weakly-directed-58}{%
\ \ \ \ \isacommand{by}\isamarkupfalse%
\ {\isacharparenleft}{\kern0pt}smt\ {\isacharparenleft}{\kern0pt}z{\isadigit{3}}{\isacharparenright}{\kern0pt}\ {\Cartouche{lemma:Ax-2-weakly-directed}{58}{0}}\ assms{\isacharparenleft}{\kern0pt}{\isadigit{2}}{\isacharparenright}{\kern0pt}\ assms{\isacharparenleft}{\kern0pt}{\isadigit{3}}{\isacharparenright}{\kern0pt}\ assms{\isacharparenleft}{\kern0pt}{\isadigit{4}}{\isacharparenright}{\kern0pt}\ assms{\isacharparenleft}{\kern0pt}{\isadigit{5}}{\isacharparenright}{\kern0pt}\ assms{\isacharparenleft}{\kern0pt}{\isadigit{8}}{\isacharparenright}{\kern0pt}\ c{\isadigit{4}}{\isacharprime}{\kern0pt}\ exactly{\isacharunderscore}{\kern0pt}one{\isacharunderscore}{\kern0pt}in{\isacharunderscore}{\kern0pt}maximal\ mcs{\isacharunderscore}{\kern0pt}conjunction\ mem{\isacharunderscore}{\kern0pt}Collect{\isacharunderscore}{\kern0pt}eq\ subset{\isacharunderscore}{\kern0pt}iff{\isacharparenright}{\kern0pt}\isanewline%
}%
\SNIP{lemma:Ax-2-weakly-directed-59-0}{%
{\isacharparenleft}{\kern0pt}consistent\ A\ Z{\isacharprime}{\kern0pt}{\isacharparenright}{\kern0pt}\ {\isasymand}\ {\isacharparenleft}{\kern0pt}maximal\ A\ Z{\isacharprime}{\kern0pt}{\isacharparenright}{\kern0pt}%
}%
\SNIP{lemma:Ax-2-weakly-directed-59-1}{%
Z{\isacharprime}{\kern0pt}\ {\isasymin}\ {\isacharparenleft}{\kern0pt}reach\ A\ i\ V{\isacharparenright}{\kern0pt}%
}%
\SNIP{lemma:Ax-2-weakly-directed-59}{%
\ \ \isacommand{then}\isamarkupfalse%
\ \isacommand{obtain}\isamarkupfalse%
\ Z{\isacharprime}{\kern0pt}\ \isakeyword{where}\ z{\isadigit{1}}{\isacharcolon}{\kern0pt}{\Cartouche{lemma:Ax-2-weakly-directed}{59}{0}}\ \isakeyword{and}\ z{\isadigit{2}}{\isacharcolon}{\kern0pt}{\Cartouche{lemma:Ax-2-weakly-directed}{59}{1}}\ \isanewline%
}%
\SNIP{lemma:Ax-2-weakly-directed-60-0}{%
{\isacharparenleft}{\kern0pt}{\isacharparenleft}{\kern0pt}K\ i\ {\isacharquery}{\kern0pt}q{\isacharparenright}{\kern0pt}\ \isactrlbold {\isasymand}\ L\ i\ {\isacharquery}{\kern0pt}p{\isacharparenright}{\kern0pt}\ {\isasymin}\ Z{\isacharprime}{\kern0pt}%
}%
\SNIP{lemma:Ax-2-weakly-directed-60}{%
\ \ \ \ \ \ \ \ \ \ \ \ \ \ \ \ \ \ \isakeyword{and}\ z{\isadigit{3}}{\isacharcolon}{\kern0pt}\ {\Cartouche{lemma:Ax-2-weakly-directed}{60}{0}}\isanewline%
}%
\SNIP{lemma:Ax-2-weakly-directed-61-0}{%
K\ i\ {\isacharparenleft}{\kern0pt}L\ i\ {\isacharparenleft}{\kern0pt}conjunct\ {\isacharparenleft}{\kern0pt}filter\ {\isacharparenleft}{\kern0pt}{\isasymlambda}p{\isachardot}{\kern0pt}\ p\ {\isasymin}\ known\ U\ i{\isacharparenright}{\kern0pt}\ S{\isacharprime}{\kern0pt}{\isacharparenright}{\kern0pt}{\isacharparenright}{\kern0pt}{\isacharparenright}{\kern0pt}\ {\isasymin}\ V%
}%
\SNIP{lemma:Ax-2-weakly-directed-61}{%
\ \ \ \ \isacommand{using}\isamarkupfalse%
\ {\Cartouche{lemma:Ax-2-weakly-directed}{61}{0}}\ assms{\isacharparenleft}{\kern0pt}{\isadigit{4}}{\isacharparenright}{\kern0pt}\ assms{\isacharparenleft}{\kern0pt}{\isadigit{5}}{\isacharparenright}{\kern0pt}\ assms{\isacharparenleft}{\kern0pt}{\isadigit{8}}{\isacharparenright}{\kern0pt}\ c{\isadigit{4}}{\isacharprime}{\kern0pt}\ mcs{\isacharunderscore}{\kern0pt}conjunction\ \isacommand{by}\isamarkupfalse%
\ blast\isanewline%
}%
\SNIP{lemma:Ax-2-weakly-directed-62-0}{%
{\isacharparenleft}{\kern0pt}L\ i\ {\isacharparenleft}{\kern0pt}{\isacharquery}{\kern0pt}q\ \isactrlbold {\isasymand}\ {\isacharquery}{\kern0pt}p{\isacharparenright}{\kern0pt}{\isacharparenright}{\kern0pt}\ {\isasymin}\ Z{\isacharprime}{\kern0pt}%
}%
\SNIP{lemma:Ax-2-weakly-directed-62}{%
\ \ \isacommand{then}\isamarkupfalse%
\ \isacommand{have}\isamarkupfalse%
\ z{\isadigit{4}}{\isacharcolon}{\kern0pt}\ {\Cartouche{lemma:Ax-2-weakly-directed}{62}{0}}\isanewline%
}%
\SNIP{lemma:Ax-2-weakly-directed-63}{%
\ \ \ \ \isacommand{by}\isamarkupfalse%
\ {\isacharparenleft}{\kern0pt}metis\ K{\isacharunderscore}{\kern0pt}thm\ consequent{\isacharunderscore}{\kern0pt}in{\isacharunderscore}{\kern0pt}maximal\ deriv{\isacharunderscore}{\kern0pt}in{\isacharunderscore}{\kern0pt}maximal{\isacharparenright}{\kern0pt}\isanewline%
}%
\SNIP{lemma:Ax-2-weakly-directed-64-0}{%
{\isacharparenleft}{\kern0pt}L\ i\ {\isacharparenleft}{\kern0pt}L\ i\ {\isacharparenleft}{\kern0pt}{\isacharquery}{\kern0pt}q\ \isactrlbold {\isasymand}\ {\isacharquery}{\kern0pt}p{\isacharparenright}{\kern0pt}{\isacharparenright}{\kern0pt}{\isacharparenright}{\kern0pt}\ {\isasymin}\ V%
}%
\SNIP{lemma:Ax-2-weakly-directed-64}{%
\ \ \isacommand{then}\isamarkupfalse%
\ \isacommand{have}\isamarkupfalse%
\ o{\isadigit{2}}{\isacharcolon}{\kern0pt}{\Cartouche{lemma:Ax-2-weakly-directed}{64}{0}}\isanewline%
}%
\SNIP{lemma:Ax-2-weakly-directed-65}{%
\ \ \ \ \isacommand{using}\isamarkupfalse%
\ assms{\isacharparenleft}{\kern0pt}{\isadigit{2}}{\isacharparenright}{\kern0pt}\ assms{\isacharparenleft}{\kern0pt}{\isadigit{3}}{\isacharparenright}{\kern0pt}\ mcs{\isacharunderscore}{\kern0pt}properties{\isacharparenleft}{\kern0pt}{\isadigit{2}}{\isacharparenright}{\kern0pt}\ z{\isadigit{1}}\ z{\isadigit{2}}\ \isacommand{by}\isamarkupfalse%
\ blast\isanewline%
}%
\SNIP{lemma:Ax-2-weakly-directed-66-0}{%
A\ {\isasymturnstile}\ K\ i\ {\isacharparenleft}{\kern0pt}K\ i\ {\isacharparenleft}{\kern0pt}{\isacharparenleft}{\kern0pt}{\isacharparenleft}{\kern0pt}{\isacharquery}{\kern0pt}p\ \isactrlbold {\isasymand}\ {\isacharquery}{\kern0pt}q{\isacharparenright}{\kern0pt}\ \isactrlbold {\isasymlongrightarrow}\ \isactrlbold {\isasymbottom}{\isacharparenright}{\kern0pt}{\isacharparenright}{\kern0pt}{\isacharparenright}{\kern0pt}%
}%
\SNIP{lemma:Ax-2-weakly-directed-66}{%
\ \ \isacommand{then}\isamarkupfalse%
\ \isacommand{have}\isamarkupfalse%
\ {\Cartouche{lemma:Ax-2-weakly-directed}{66}{0}}\ \isanewline%
}%
\SNIP{lemma:Ax-2-weakly-directed-67}{%
\ \ \ \ \isacommand{by}\isamarkupfalse%
\ {\isacharparenleft}{\kern0pt}metis\ R{\isadigit{2}}\ o{\isadigit{1}}{\isacharparenright}{\kern0pt}\isanewline%
}%
\SNIP{lemma:Ax-2-weakly-directed-68-0}{%
K\ i\ {\isacharparenleft}{\kern0pt}K\ i\ {\isacharparenleft}{\kern0pt}{\isacharparenleft}{\kern0pt}{\isacharparenleft}{\kern0pt}{\isacharquery}{\kern0pt}p\ \isactrlbold {\isasymand}\ {\isacharquery}{\kern0pt}q{\isacharparenright}{\kern0pt}\ \isactrlbold {\isasymlongrightarrow}\ \isactrlbold {\isasymbottom}{\isacharparenright}{\kern0pt}{\isacharparenright}{\kern0pt}{\isacharparenright}{\kern0pt}\ {\isasymin}\ V%
}%
\SNIP{lemma:Ax-2-weakly-directed-68}{%
\ \ \isacommand{then}\isamarkupfalse%
\ \isacommand{have}\isamarkupfalse%
\ o{\isadigit{3}}{\isacharcolon}{\kern0pt}{\Cartouche{lemma:Ax-2-weakly-directed}{68}{0}}\isanewline%
}%
\SNIP{lemma:Ax-2-weakly-directed-69}{%
\ \ \ \ \isacommand{using}\isamarkupfalse%
\ assms{\isacharparenleft}{\kern0pt}{\isadigit{2}}{\isacharparenright}{\kern0pt}\ assms{\isacharparenleft}{\kern0pt}{\isadigit{3}}{\isacharparenright}{\kern0pt}\ deriv{\isacharunderscore}{\kern0pt}in{\isacharunderscore}{\kern0pt}maximal\ \isacommand{by}\isamarkupfalse%
\ blast\isanewline%
}%
\SNIP{lemma:Ax-2-weakly-directed-70-0}{%
{\isacharparenleft}{\kern0pt}consistent\ A\ X{\isadigit{1}}{\isacharparenright}{\kern0pt}\ {\isasymand}\ {\isacharparenleft}{\kern0pt}maximal\ A\ X{\isadigit{1}}{\isacharparenright}{\kern0pt}%
}%
\SNIP{lemma:Ax-2-weakly-directed-70-1}{%
X{\isadigit{1}}\ {\isasymin}\ {\isacharparenleft}{\kern0pt}reach\ A\ i\ V{\isacharparenright}{\kern0pt}%
}%
\SNIP{lemma:Ax-2-weakly-directed-70}{%
\ \ \isacommand{then}\isamarkupfalse%
\ \isacommand{obtain}\isamarkupfalse%
\ X{\isadigit{1}}\ \isakeyword{where}\ x{\isadigit{1}}{\isacharcolon}{\kern0pt}{\Cartouche{lemma:Ax-2-weakly-directed}{70}{0}}\ \isakeyword{and}\ x{\isadigit{2}}{\isacharcolon}{\kern0pt}{\Cartouche{lemma:Ax-2-weakly-directed}{70}{1}}\ \isanewline%
}%
\SNIP{lemma:Ax-2-weakly-directed-71-0}{%
{\isacharparenleft}{\kern0pt}L\ i\ {\isacharparenleft}{\kern0pt}{\isacharquery}{\kern0pt}q\ \isactrlbold {\isasymand}\ {\isacharquery}{\kern0pt}p{\isacharparenright}{\kern0pt}{\isacharparenright}{\kern0pt}\ {\isasymin}\ X{\isadigit{1}}%
}%
\SNIP{lemma:Ax-2-weakly-directed-71}{%
\ \ \ \ \ \ \ \ \ \ \ \ \ \ \ \ \ \ \isakeyword{and}\ x{\isadigit{3}}{\isacharcolon}{\kern0pt}\ {\Cartouche{lemma:Ax-2-weakly-directed}{71}{0}}\isanewline%
}%
\SNIP{lemma:Ax-2-weakly-directed-72}{%
\ \ \ \ \isacommand{using}\isamarkupfalse%
\ z{\isadigit{1}}\ z{\isadigit{2}}\ z{\isadigit{4}}\ \isacommand{by}\isamarkupfalse%
\ blast\isanewline%
}%
\SNIP{lemma:Ax-2-weakly-directed-73-0}{%
{\isacharparenleft}{\kern0pt}K\ i\ {\isacharparenleft}{\kern0pt}{\isacharparenleft}{\kern0pt}{\isacharparenleft}{\kern0pt}{\isacharquery}{\kern0pt}p\ \isactrlbold {\isasymand}\ {\isacharquery}{\kern0pt}q{\isacharparenright}{\kern0pt}\ \isactrlbold {\isasymlongrightarrow}\ \isactrlbold {\isasymbottom}{\isacharparenright}{\kern0pt}{\isacharparenright}{\kern0pt}{\isacharparenright}{\kern0pt}\ {\isasymin}\ X{\isadigit{1}}%
}%
\SNIP{lemma:Ax-2-weakly-directed-73}{%
\ \ \isacommand{then}\isamarkupfalse%
\ \isacommand{have}\isamarkupfalse%
\ x{\isadigit{4}}{\isacharcolon}{\kern0pt}{\Cartouche{lemma:Ax-2-weakly-directed}{73}{0}}\isanewline%
}%
\SNIP{lemma:Ax-2-weakly-directed-74}{%
\ \ \ \ \isacommand{using}\isamarkupfalse%
\ o{\isadigit{3}}\ \isacommand{by}\isamarkupfalse%
\ blast\isanewline%
}%
\SNIP{lemma:Ax-2-weakly-directed-75-0}{%
{\isasymforall}x{\isachardot}{\kern0pt}\ {\isasymforall}y{\isachardot}{\kern0pt}\ tautology\ {\isacharparenleft}{\kern0pt}{\isacharparenleft}{\kern0pt}{\isacharparenleft}{\kern0pt}x\isactrlbold {\isasymand}y{\isacharparenright}{\kern0pt}\isactrlbold {\isasymlongrightarrow}\isactrlbold {\isasymbottom}{\isacharparenright}{\kern0pt}\isactrlbold {\isasymlongrightarrow}\isactrlbold {\isasymnot}{\isacharparenleft}{\kern0pt}y\isactrlbold {\isasymand}x{\isacharparenright}{\kern0pt}{\isacharparenright}{\kern0pt}%
}%
\SNIP{lemma:Ax-2-weakly-directed-75}{%
\ \ \isacommand{then}\isamarkupfalse%
\ \isacommand{have}\isamarkupfalse%
\ t{\isacharcolon}{\kern0pt}{\Cartouche{lemma:Ax-2-weakly-directed}{75}{0}}\isanewline%
}%
\SNIP{lemma:Ax-2-weakly-directed-76}{%
\ \ \ \ \isacommand{by}\isamarkupfalse%
\ {\isacharparenleft}{\kern0pt}metis\ eval{\isachardot}{\kern0pt}simps{\isacharparenleft}{\kern0pt}{\isadigit{4}}{\isacharparenright}{\kern0pt}\ eval{\isachardot}{\kern0pt}simps{\isacharparenleft}{\kern0pt}{\isadigit{5}}{\isacharparenright}{\kern0pt}{\isacharparenright}{\kern0pt}\isanewline%
}%
\SNIP{lemma:Ax-2-weakly-directed-77-0}{%
{\isacharparenleft}{\kern0pt}{\isacharparenleft}{\kern0pt}{\isacharparenleft}{\kern0pt}{\isacharquery}{\kern0pt}p\isactrlbold {\isasymand}{\isacharquery}{\kern0pt}q{\isacharparenright}{\kern0pt}\isactrlbold {\isasymlongrightarrow}\isactrlbold {\isasymbottom}{\isacharparenright}{\kern0pt}\isactrlbold {\isasymlongrightarrow}\isactrlbold {\isasymnot}{\isacharparenleft}{\kern0pt}{\isacharquery}{\kern0pt}q\isactrlbold {\isasymand}{\isacharquery}{\kern0pt}p{\isacharparenright}{\kern0pt}{\isacharparenright}{\kern0pt}\ {\isasymin}\ X{\isadigit{1}}%
}%
\SNIP{lemma:Ax-2-weakly-directed-77}{%
\ \ \isacommand{then}\isamarkupfalse%
\ \isacommand{have}\isamarkupfalse%
\ {\Cartouche{lemma:Ax-2-weakly-directed}{77}{0}}\isanewline%
}%
\SNIP{lemma:Ax-2-weakly-directed-78}{%
\ \ \ \ \isacommand{using}\isamarkupfalse%
\ A{\isadigit{1}}\ deriv{\isacharunderscore}{\kern0pt}in{\isacharunderscore}{\kern0pt}maximal\ x{\isadigit{1}}\ \isacommand{by}\isamarkupfalse%
\ blast\isanewline%
}%
\SNIP{lemma:Ax-2-weakly-directed-79-0}{%
K\ i\ {\isacharparenleft}{\kern0pt}{\isacharparenleft}{\kern0pt}{\isacharparenleft}{\kern0pt}{\isacharquery}{\kern0pt}p\isactrlbold {\isasymand}{\isacharquery}{\kern0pt}q{\isacharparenright}{\kern0pt}\isactrlbold {\isasymlongrightarrow}\isactrlbold {\isasymbottom}{\isacharparenright}{\kern0pt}\isactrlbold {\isasymlongrightarrow}\isactrlbold {\isasymnot}{\isacharparenleft}{\kern0pt}{\isacharquery}{\kern0pt}q\isactrlbold {\isasymand}{\isacharquery}{\kern0pt}p{\isacharparenright}{\kern0pt}{\isacharparenright}{\kern0pt}\ {\isasymin}\ X{\isadigit{1}}%
}%
\SNIP{lemma:Ax-2-weakly-directed-79}{%
\ \ \isacommand{then}\isamarkupfalse%
\ \isacommand{have}\isamarkupfalse%
\ {\Cartouche{lemma:Ax-2-weakly-directed}{79}{0}}\isanewline%
}%
\SNIP{lemma:Ax-2-weakly-directed-80}{%
\ \ \ \ \isacommand{by}\isamarkupfalse%
\ {\isacharparenleft}{\kern0pt}meson\ A{\isadigit{1}}\ R{\isadigit{2}}\ deriv{\isacharunderscore}{\kern0pt}in{\isacharunderscore}{\kern0pt}maximal\ t\ x{\isadigit{1}}{\isacharparenright}{\kern0pt}\isanewline%
}%
\SNIP{lemma:Ax-2-weakly-directed-81-0}{%
{\isacharparenleft}{\kern0pt}K\ i\ {\isacharparenleft}{\kern0pt}{\isacharparenleft}{\kern0pt}{\isacharquery}{\kern0pt}p\isactrlbold {\isasymand}{\isacharquery}{\kern0pt}q{\isacharparenright}{\kern0pt}\isactrlbold {\isasymlongrightarrow}\isactrlbold {\isasymbottom}{\isacharparenright}{\kern0pt}\isactrlbold {\isasymlongrightarrow}\ K\ i\ {\isacharparenleft}{\kern0pt}\isactrlbold {\isasymnot}{\isacharparenleft}{\kern0pt}{\isacharquery}{\kern0pt}q\isactrlbold {\isasymand}{\isacharquery}{\kern0pt}p{\isacharparenright}{\kern0pt}{\isacharparenright}{\kern0pt}{\isacharparenright}{\kern0pt}\ {\isasymin}\ X{\isadigit{1}}%
}%
\SNIP{lemma:Ax-2-weakly-directed-81}{%
\ \ \isacommand{then}\isamarkupfalse%
\ \isacommand{have}\isamarkupfalse%
\ {\Cartouche{lemma:Ax-2-weakly-directed}{81}{0}}\isanewline%
}%
\SNIP{lemma:Ax-2-weakly-directed-82}{%
\ \ \ \ \isacommand{by}\isamarkupfalse%
\ {\isacharparenleft}{\kern0pt}meson\ K{\isacharunderscore}{\kern0pt}A{\isadigit{2}}{\isacharprime}{\kern0pt}\ consequent{\isacharunderscore}{\kern0pt}in{\isacharunderscore}{\kern0pt}maximal\ deriv{\isacharunderscore}{\kern0pt}in{\isacharunderscore}{\kern0pt}maximal\ x{\isadigit{1}}{\isacharparenright}{\kern0pt}\isanewline%
}%
\SNIP{lemma:Ax-2-weakly-directed-83-0}{%
{\isacharparenleft}{\kern0pt}K\ i\ {\isacharparenleft}{\kern0pt}{\isacharparenleft}{\kern0pt}{\isacharquery}{\kern0pt}p\isactrlbold {\isasymand}{\isacharquery}{\kern0pt}q{\isacharparenright}{\kern0pt}\isactrlbold {\isasymlongrightarrow}\isactrlbold {\isasymbottom}{\isacharparenright}{\kern0pt}\isactrlbold {\isasymlongrightarrow}\ {\isacharparenleft}{\kern0pt}\isactrlbold {\isasymnot}\ L\ i{\isacharparenleft}{\kern0pt}{\isacharquery}{\kern0pt}q\isactrlbold {\isasymand}{\isacharquery}{\kern0pt}p{\isacharparenright}{\kern0pt}{\isacharparenright}{\kern0pt}{\isacharparenright}{\kern0pt}\ {\isasymin}\ X{\isadigit{1}}%
}%
\SNIP{lemma:Ax-2-weakly-directed-83}{%
\ \ \isacommand{then}\isamarkupfalse%
\ \isacommand{have}\isamarkupfalse%
\ {\Cartouche{lemma:Ax-2-weakly-directed}{83}{0}}\isanewline%
}%
\SNIP{lemma:Ax-2-weakly-directed-84}{%
\ \ \ \ \isacommand{using}\isamarkupfalse%
\ consequent{\isacharunderscore}{\kern0pt}in{\isacharunderscore}{\kern0pt}maximal\ exactly{\isacharunderscore}{\kern0pt}one{\isacharunderscore}{\kern0pt}in{\isacharunderscore}{\kern0pt}maximal\ x{\isadigit{1}}\ x{\isadigit{3}}\ x{\isadigit{4}}\ \isacommand{by}\isamarkupfalse%
\ blast\isanewline%
}%
\SNIP{lemma:Ax-2-weakly-directed-85-0}{%
{\isacharparenleft}{\kern0pt}\isactrlbold {\isasymnot}\ L\ i{\isacharparenleft}{\kern0pt}{\isacharquery}{\kern0pt}q\isactrlbold {\isasymand}{\isacharquery}{\kern0pt}p{\isacharparenright}{\kern0pt}{\isacharparenright}{\kern0pt}\ {\isasymin}\ X{\isadigit{1}}\ {\isasymand}\ {\isacharparenleft}{\kern0pt}L\ i{\isacharparenleft}{\kern0pt}{\isacharquery}{\kern0pt}q\isactrlbold {\isasymand}{\isacharquery}{\kern0pt}p{\isacharparenright}{\kern0pt}{\isacharparenright}{\kern0pt}\ {\isasymin}\ X{\isadigit{1}}%
}%
\SNIP{lemma:Ax-2-weakly-directed-85}{%
\ \ \isacommand{then}\isamarkupfalse%
\ \isacommand{have}\isamarkupfalse%
\ {\Cartouche{lemma:Ax-2-weakly-directed}{85}{0}}\isanewline%
}%
\SNIP{lemma:Ax-2-weakly-directed-86}{%
\ \ \ \ \isacommand{using}\isamarkupfalse%
\ consequent{\isacharunderscore}{\kern0pt}in{\isacharunderscore}{\kern0pt}maximal\ x{\isadigit{1}}\ x{\isadigit{4}}\ x{\isadigit{3}}\ \isacommand{by}\isamarkupfalse%
\ blast\isanewline%
}%
\SNIP{lemma:Ax-2-weakly-directed-87}{%
\ \ \isacommand{then}\isamarkupfalse%
\ \isacommand{show}\isamarkupfalse%
\ False\isanewline%
}%
\SNIP{lemma:Ax-2-weakly-directed-88}{%
\ \ \ \ \isacommand{using}\isamarkupfalse%
\ exactly{\isacharunderscore}{\kern0pt}one{\isacharunderscore}{\kern0pt}in{\isacharunderscore}{\kern0pt}maximal\ x{\isadigit{1}}\ \isacommand{by}\isamarkupfalse%
\ blast\isanewline%
}%
\SNIP{lemma:Ax-2-weakly-directed-89}{%
\isacommand{qed}\isamarkupfalse%
\endisatagproof
{\isafoldproof}%
\isadelimproof%
}%
\SNIP{lemma:mcs-2-weakly-directed-0}{%
\isacommand{lemma}\isamarkupfalse%
\ mcs\isactrlsub {\isacharunderscore}{\kern0pt}\isactrlsub {\isadigit{2}}{\isacharunderscore}{\kern0pt}weakly{\isacharunderscore}{\kern0pt}directed{\isacharcolon}{\kern0pt}\isanewline%
}%
\SNIP{lemma:mcs-2-weakly-directed-1-0}{%
{\isacharparenleft}{\kern0pt}{\isacharparenleft}{\kern0pt}{\isacharprime}{\kern0pt}i\ {\isacharcolon}{\kern0pt}{\isacharcolon}{\kern0pt}\ countable{\isacharparenright}{\kern0pt}\ fm\ {\isasymRightarrow}\ bool{\isacharparenright}{\kern0pt}%
}%
\SNIP{lemma:mcs-2-weakly-directed-1}{%
\ \ \isakeyword{fixes}\ A\ {\isacharcolon}{\kern0pt}{\isacharcolon}{\kern0pt}\ {\Cartouche{lemma:mcs-2-weakly-directed}{1}{0}}\isanewline%
}%
\SNIP{lemma:mcs-2-weakly-directed-2-0}{%
{\isasymforall}p{\isachardot}{\kern0pt}\ Ax{\isacharunderscore}{\kern0pt}{\isadigit{2}}\ p\ {\isasymlongrightarrow}\ A\ p%
}%
\SNIP{lemma:mcs-2-weakly-directed-2}{%
\ \ \isakeyword{assumes}\ {\Cartouche{lemma:mcs-2-weakly-directed}{2}{0}}\isanewline%
}%
\SNIP{lemma:mcs-2-weakly-directed-3-0}{%
weakly{\isacharunderscore}{\kern0pt}directed\ {\isacharparenleft}{\kern0pt}Kripke\ {\isacharparenleft}{\kern0pt}mcss\ A{\isacharparenright}{\kern0pt}\ pi\ {\isacharparenleft}{\kern0pt}reach\ A{\isacharparenright}{\kern0pt}{\isacharparenright}{\kern0pt}%
}%
\SNIP{lemma:mcs-2-weakly-directed-3}{%
\ \ \isakeyword{shows}\ {\Cartouche{lemma:mcs-2-weakly-directed}{3}{0}}\isanewline%
}%
\SNIP{lemma:mcs-2-weakly-directed-4}{%
\isadelimproof
\ \ %
\endisadelimproof
\isatagproof
\isacommand{unfolding}\isamarkupfalse%
\ weakly{\isacharunderscore}{\kern0pt}directed{\isacharunderscore}{\kern0pt}def\isanewline%
}%
\SNIP{lemma:mcs-2-weakly-directed-5}{%
\isacommand{proof}\isamarkupfalse%
\ {\isacharparenleft}{\kern0pt}intro\ allI\ ballI{\isacharcomma}{\kern0pt}\ auto{\isacharparenright}{\kern0pt}\isanewline%
}%
\SNIP{lemma:mcs-2-weakly-directed-6}{%
\ \ \isacommand{fix}\isamarkupfalse%
\ i\ V\ U\ W\isanewline%
}%
\SNIP{lemma:mcs-2-weakly-directed-7-0}{%
consistent\ A\ V%
}%
\SNIP{lemma:mcs-2-weakly-directed-7-1}{%
maximal\ A\ V%
}%
\SNIP{lemma:mcs-2-weakly-directed-7-2}{%
consistent\ A\ U%
}%
\SNIP{lemma:mcs-2-weakly-directed-7-3}{%
maximal\ A\ U%
}%
\SNIP{lemma:mcs-2-weakly-directed-7-4}{%
consistent\ A\ W%
}%
\SNIP{lemma:mcs-2-weakly-directed-7}{%
\ \ \isacommand{assume}\isamarkupfalse%
\ {\Cartouche{lemma:mcs-2-weakly-directed}{7}{0}}\ {\Cartouche{lemma:mcs-2-weakly-directed}{7}{1}}\ {\Cartouche{lemma:mcs-2-weakly-directed}{7}{2}}\ {\Cartouche{lemma:mcs-2-weakly-directed}{7}{3}}\ {\Cartouche{lemma:mcs-2-weakly-directed}{7}{4}}\ \isanewline%
}%
\SNIP{lemma:mcs-2-weakly-directed-8-0}{%
maximal\ A\ W%
}%
\SNIP{lemma:mcs-2-weakly-directed-8-1}{%
known\ V\ i\ {\isasymsubseteq}\ U%
}%
\SNIP{lemma:mcs-2-weakly-directed-8-2}{%
known\ V\ i\ {\isasymsubseteq}\ W%
}%
\SNIP{lemma:mcs-2-weakly-directed-8}{%
\ \ \ \ \ \ {\Cartouche{lemma:mcs-2-weakly-directed}{8}{0}}\ {\Cartouche{lemma:mcs-2-weakly-directed}{8}{1}}\ {\Cartouche{lemma:mcs-2-weakly-directed}{8}{2}}\isanewline%
}%
\SNIP{lemma:mcs-2-weakly-directed-9-0}{%
{\isasymexists}X{\isachardot}{\kern0pt}\ {\isacharparenleft}{\kern0pt}consistent\ A\ X{\isacharparenright}{\kern0pt}\ {\isasymand}\ {\isacharparenleft}{\kern0pt}maximal\ A\ X{\isacharparenright}{\kern0pt}\ {\isasymand}\ X\ {\isasymin}\ {\isacharparenleft}{\kern0pt}reach\ A\ i\ W{\isacharparenright}{\kern0pt}\ {\isasyminter}\ {\isacharparenleft}{\kern0pt}reach\ A\ i\ U{\isacharparenright}{\kern0pt}%
}%
\SNIP{lemma:mcs-2-weakly-directed-9}{%
\ \ \isacommand{then}\isamarkupfalse%
\ \isacommand{have}\isamarkupfalse%
\ {\Cartouche{lemma:mcs-2-weakly-directed}{9}{0}}\isanewline%
}%
\SNIP{lemma:mcs-2-weakly-directed-10}{%
\ \ \ \ \isacommand{using}\isamarkupfalse%
\ Ax{\isacharunderscore}{\kern0pt}{\isadigit{2}}{\isacharunderscore}{\kern0pt}weakly{\isacharunderscore}{\kern0pt}directed\ {\isacharbrackleft}{\kern0pt}\isakeyword{where}\ A{\isacharequal}{\kern0pt}A\ \isakeyword{and}\ V{\isacharequal}{\kern0pt}V\ \isakeyword{and}\ W{\isacharequal}{\kern0pt}W\ \isakeyword{and}\ U{\isacharequal}{\kern0pt}U{\isacharbrackright}{\kern0pt}\ assms\ IntD{\isadigit{2}}\ \isanewline%
}%
\SNIP{lemma:mcs-2-weakly-directed-11}{%
\ \ \ \ \ \ \isacommand{by}\isamarkupfalse%
\ simp\isanewline%
}%
\SNIP{lemma:mcs-2-weakly-directed-12-0}{%
{\isasymexists}X{\isachardot}{\kern0pt}\ {\isacharparenleft}{\kern0pt}consistent\ A\ X{\isacharparenright}{\kern0pt}\ {\isasymand}\ {\isacharparenleft}{\kern0pt}maximal\ A\ X{\isacharparenright}{\kern0pt}\ {\isasymand}\ X\ {\isasymin}\ {\isacharparenleft}{\kern0pt}reach\ A\ i\ W{\isacharparenright}{\kern0pt}\ {\isasymand}\ X\ {\isasymin}\ {\isacharparenleft}{\kern0pt}reach\ A\ i\ U{\isacharparenright}{\kern0pt}%
}%
\SNIP{lemma:mcs-2-weakly-directed-12}{%
\ \ \isacommand{then}\isamarkupfalse%
\ \isacommand{have}\isamarkupfalse%
\ {\Cartouche{lemma:mcs-2-weakly-directed}{12}{0}}\isanewline%
}%
\SNIP{lemma:mcs-2-weakly-directed-13}{%
\ \ \ \ \isacommand{by}\isamarkupfalse%
\ simp\isanewline%
}%
\SNIP{lemma:mcs-2-weakly-directed-14-0}{%
{\isasymexists}X{\isachardot}{\kern0pt}\ {\isacharparenleft}{\kern0pt}consistent\ A\ X{\isacharparenright}{\kern0pt}\ {\isasymand}\ {\isacharparenleft}{\kern0pt}maximal\ A\ X{\isacharparenright}{\kern0pt}\ {\isasymand}\ known\ W\ i\ {\isasymsubseteq}\ X\ {\isasymand}\ known\ U\ i\ {\isasymsubseteq}\ X%
}%
\SNIP{lemma:mcs-2-weakly-directed-14}{%
\ \ \isacommand{then}\isamarkupfalse%
\ \isacommand{show}\isamarkupfalse%
\ {\Cartouche{lemma:mcs-2-weakly-directed}{14}{0}}\isanewline%
}%
\SNIP{lemma:mcs-2-weakly-directed-15}{%
\ \ \ \ \isacommand{by}\isamarkupfalse%
\ auto\isanewline%
}%
\SNIP{lemma:mcs-2-weakly-directed-16}{%
\isacommand{qed}\isamarkupfalse%
\endisatagproof
{\isafoldproof}%
\isadelimproof%
}%
\SNIP{lemma:imply-completeness-K-2-0}{%
\isacommand{lemma}\isamarkupfalse%
\ imply{\isacharunderscore}{\kern0pt}completeness{\isacharunderscore}{\kern0pt}K{\isacharunderscore}{\kern0pt}{\isadigit{2}}{\isacharcolon}{\kern0pt}\isanewline%
}%
\SNIP{lemma:imply-completeness-K-2-1-0}{%
{\isasymforall}{\isacharparenleft}{\kern0pt}M\ {\isacharcolon}{\kern0pt}{\isacharcolon}{\kern0pt}\ {\isacharparenleft}{\kern0pt}{\isacharprime}{\kern0pt}i\ {\isacharcolon}{\kern0pt}{\isacharcolon}{\kern0pt}\ countable{\isacharcomma}{\kern0pt}\ {\isacharprime}{\kern0pt}i\ fm\ set{\isacharparenright}{\kern0pt}\ kripke{\isacharparenright}{\kern0pt}{\isachardot}{\kern0pt}\ {\isasymforall}w\ {\isasymin}\ {\isasymW}\ M{\isachardot}{\kern0pt}\isanewline\ \ \ \ weakly{\isacharunderscore}{\kern0pt}directed\ M\ {\isasymlongrightarrow}\ {\isacharparenleft}{\kern0pt}{\isasymforall}q\ {\isasymin}\ G{\isachardot}{\kern0pt}\ M{\isacharcomma}{\kern0pt}\ w\ {\isasymTurnstile}\ q{\isacharparenright}{\kern0pt}\ {\isasymlongrightarrow}\ M{\isacharcomma}{\kern0pt}\ w\ {\isasymTurnstile}\ p%
}%
\SNIP{lemma:imply-completeness-K-2-1}{%
\ \ \isakeyword{assumes}\ valid{\isacharcolon}{\kern0pt}\ {\isacartoucheopen}{\isasymforall}{\isacharparenleft}{\kern0pt}M\ {\isacharcolon}{\kern0pt}{\isacharcolon}{\kern0pt}\ {\isacharparenleft}{\kern0pt}{\isacharprime}{\kern0pt}i\ {\isacharcolon}{\kern0pt}{\isacharcolon}{\kern0pt}\ countable{\isacharcomma}{\kern0pt}\ {\isacharprime}{\kern0pt}i\ fm\ set{\isacharparenright}{\kern0pt}\ kripke{\isacharparenright}{\kern0pt}{\isachardot}{\kern0pt}\ {\isasymforall}w\ {\isasymin}\ {\isasymW}\ M{\isachardot}{\kern0pt}\isanewline%
}%
\SNIP{lemma:imply-completeness-K-2-2}{%
\ \ \ \ weakly{\isacharunderscore}{\kern0pt}directed\ M\ {\isasymlongrightarrow}\ {\isacharparenleft}{\kern0pt}{\isasymforall}q\ {\isasymin}\ G{\isachardot}{\kern0pt}\ M{\isacharcomma}{\kern0pt}\ w\ {\isasymTurnstile}\ q{\isacharparenright}{\kern0pt}\ {\isasymlongrightarrow}\ M{\isacharcomma}{\kern0pt}\ w\ {\isasymTurnstile}\ p{\isacartoucheclose}\isanewline%
}%
\SNIP{lemma:imply-completeness-K-2-3-0}{%
{\isasymexists}qs{\isachardot}{\kern0pt}\ set\ qs\ {\isasymsubseteq}\ G\ {\isasymand}\ {\isacharparenleft}{\kern0pt}Ax{\isacharunderscore}{\kern0pt}{\isadigit{2}}\ {\isasymturnstile}\ imply\ qs\ p{\isacharparenright}{\kern0pt}%
}%
\SNIP{lemma:imply-completeness-K-2-3}{%
\ \ \isakeyword{shows}\ {\Cartouche{lemma:imply-completeness-K-2}{3}{0}}\isanewline%
}%
\SNIP{lemma:imply-completeness-K-2-4}{%
\isadelimproof
\endisadelimproof
\isatagproof
\isacommand{proof}\isamarkupfalse%
\ {\isacharparenleft}{\kern0pt}rule\ ccontr{\isacharparenright}{\kern0pt}\isanewline%
}%
\SNIP{lemma:imply-completeness-K-2-5-0}{%
{\isasymnexists}qs{\isachardot}{\kern0pt}\ set\ qs\ {\isasymsubseteq}\ G\ {\isasymand}\ Ax{\isacharunderscore}{\kern0pt}{\isadigit{2}}\ {\isasymturnstile}\ imply\ qs\ p%
}%
\SNIP{lemma:imply-completeness-K-2-5}{%
\isacommand{assume}\isamarkupfalse%
\ {\Cartouche{lemma:imply-completeness-K-2}{5}{0}}\isanewline%
}%
\SNIP{lemma:imply-completeness-K-2-6-0}{%
{\isasymforall}qs{\isachardot}{\kern0pt}\ set\ qs\ {\isasymsubseteq}\ G\ {\isasymlongrightarrow}\ {\isasymnot}\ Ax{\isacharunderscore}{\kern0pt}{\isadigit{2}}\ {\isasymturnstile}\ imply\ {\isacharparenleft}{\kern0pt}{\isacharparenleft}{\kern0pt}\isactrlbold {\isasymnot}\ p{\isacharparenright}{\kern0pt}\ {\isacharhash}{\kern0pt}\ qs{\isacharparenright}{\kern0pt}\ \isactrlbold {\isasymbottom}%
}%
\SNIP{lemma:imply-completeness-K-2-6}{%
\ \ \isacommand{then}\isamarkupfalse%
\ \isacommand{have}\isamarkupfalse%
\ {\isacharasterisk}{\kern0pt}{\isacharcolon}{\kern0pt}\ {\Cartouche{lemma:imply-completeness-K-2}{6}{0}}\isanewline%
}%
\SNIP{lemma:imply-completeness-K-2-7}{%
\ \ \ \ \isacommand{using}\isamarkupfalse%
\ K{\isacharunderscore}{\kern0pt}Boole\ \isacommand{by}\isamarkupfalse%
\ blast\isanewline%
}%
\SNIP{lemma:imply-completeness-K-2-8}{%
\isanewline%
}%
\SNIP{lemma:imply-completeness-K-2-9-0}{%
{\isacharbraceleft}{\kern0pt}\isactrlbold {\isasymnot}\ p{\isacharbraceright}{\kern0pt}\ {\isasymunion}\ G%
}%
\SNIP{lemma:imply-completeness-K-2-9}{%
\ \ \isacommand{let}\isamarkupfalse%
\ {\isacharquery}{\kern0pt}S\ {\isacharequal}{\kern0pt}\ {\Cartouche{lemma:imply-completeness-K-2}{9}{0}}\isanewline%
}%
\SNIP{lemma:imply-completeness-K-2-10-0}{%
Extend\ Ax{\isacharunderscore}{\kern0pt}{\isadigit{2}}\ {\isacharquery}{\kern0pt}S\ from{\isacharunderscore}{\kern0pt}nat%
}%
\SNIP{lemma:imply-completeness-K-2-10}{%
\ \ \isacommand{let}\isamarkupfalse%
\ {\isacharquery}{\kern0pt}V\ {\isacharequal}{\kern0pt}\ {\Cartouche{lemma:imply-completeness-K-2}{10}{0}}\isanewline%
}%
\SNIP{lemma:imply-completeness-K-2-11-0}{%
Kripke\ {\isacharparenleft}{\kern0pt}mcss\ Ax{\isacharunderscore}{\kern0pt}{\isadigit{2}}{\isacharparenright}{\kern0pt}\ pi\ {\isacharparenleft}{\kern0pt}reach\ Ax{\isacharunderscore}{\kern0pt}{\isadigit{2}}{\isacharparenright}{\kern0pt}%
}%
\SNIP{lemma:imply-completeness-K-2-11}{%
\ \ \isacommand{let}\isamarkupfalse%
\ {\isacharquery}{\kern0pt}M\ {\isacharequal}{\kern0pt}\ {\Cartouche{lemma:imply-completeness-K-2}{11}{0}}\isanewline%
}%
\SNIP{lemma:imply-completeness-K-2-12}{%
\isanewline%
}%
\SNIP{lemma:imply-completeness-K-2-13-0}{%
consistent\ Ax{\isacharunderscore}{\kern0pt}{\isadigit{2}}\ {\isacharquery}{\kern0pt}S%
}%
\SNIP{lemma:imply-completeness-K-2-13}{%
\ \ \isacommand{have}\isamarkupfalse%
\ {\Cartouche{lemma:imply-completeness-K-2}{13}{0}}\isanewline%
}%
\SNIP{lemma:imply-completeness-K-2-14}{%
\ \ \ \ \isacommand{using}\isamarkupfalse%
\ {\isacharasterisk}{\kern0pt}\ \isacommand{by}\isamarkupfalse%
\ {\isacharparenleft}{\kern0pt}metis\ K{\isacharunderscore}{\kern0pt}imply{\isacharunderscore}{\kern0pt}Cons\ consistent{\isacharunderscore}{\kern0pt}def\ inconsistent{\isacharunderscore}{\kern0pt}subset{\isacharparenright}{\kern0pt}\isanewline%
}%
\SNIP{lemma:imply-completeness-K-2-15-0}{%
{\isacharquery}{\kern0pt}M{\isacharcomma}{\kern0pt}\ {\isacharquery}{\kern0pt}V\ {\isasymTurnstile}\ {\isacharparenleft}{\kern0pt}\isactrlbold {\isasymnot}\ p{\isacharparenright}{\kern0pt}%
}%
\SNIP{lemma:imply-completeness-K-2-15-1}{%
{\isasymforall}q\ {\isasymin}\ G{\isachardot}{\kern0pt}\ {\isacharquery}{\kern0pt}M{\isacharcomma}{\kern0pt}\ {\isacharquery}{\kern0pt}V\ {\isasymTurnstile}\ q%
}%
\SNIP{lemma:imply-completeness-K-2-15-2}{%
consistent\ Ax{\isacharunderscore}{\kern0pt}{\isadigit{2}}\ {\isacharquery}{\kern0pt}V%
}%
\SNIP{lemma:imply-completeness-K-2-15-3}{%
maximal\ Ax{\isacharunderscore}{\kern0pt}{\isadigit{2}}\ {\isacharquery}{\kern0pt}V%
}%
\SNIP{lemma:imply-completeness-K-2-15}{%
\ \ \isacommand{then}\isamarkupfalse%
\ \isacommand{have}\isamarkupfalse%
\ {\Cartouche{lemma:imply-completeness-K-2}{15}{0}}\ {\Cartouche{lemma:imply-completeness-K-2}{15}{1}}\ {\Cartouche{lemma:imply-completeness-K-2}{15}{2}}\ {\Cartouche{lemma:imply-completeness-K-2}{15}{3}}\isanewline%
}%
\SNIP{lemma:imply-completeness-K-2-16}{%
\ \ \ \ \isacommand{using}\isamarkupfalse%
\ canonical{\isacharunderscore}{\kern0pt}model\ \isacommand{unfolding}\isamarkupfalse%
\ list{\isacharunderscore}{\kern0pt}all{\isacharunderscore}{\kern0pt}def\ \isacommand{by}\isamarkupfalse%
\ fastforce{\isacharplus}{\kern0pt}\isanewline%
}%
\SNIP{lemma:imply-completeness-K-2-17-0}{%
weakly{\isacharunderscore}{\kern0pt}directed\ {\isacharquery}{\kern0pt}M%
}%
\SNIP{lemma:imply-completeness-K-2-17}{%
\ \ \isacommand{moreover}\isamarkupfalse%
\ \isacommand{have}\isamarkupfalse%
\ {\Cartouche{lemma:imply-completeness-K-2}{17}{0}}\isanewline%
}%
\SNIP{lemma:imply-completeness-K-2-18}{%
\ \ \ \ \isacommand{using}\isamarkupfalse%
\ mcs\isactrlsub {\isacharunderscore}{\kern0pt}\isactrlsub {\isadigit{2}}{\isacharunderscore}{\kern0pt}weakly{\isacharunderscore}{\kern0pt}directed\ {\isacharbrackleft}{\kern0pt}\isakeyword{where}\ A{\isacharequal}{\kern0pt}Ax{\isacharunderscore}{\kern0pt}{\isadigit{2}}{\isacharbrackright}{\kern0pt}\ \isacommand{by}\isamarkupfalse%
\ fast\isanewline%
}%
\SNIP{lemma:imply-completeness-K-2-19-0}{%
{\isacharquery}{\kern0pt}M{\isacharcomma}{\kern0pt}\ {\isacharquery}{\kern0pt}V\ {\isasymTurnstile}\ p%
}%
\SNIP{lemma:imply-completeness-K-2-19}{%
\ \ \isacommand{ultimately}\isamarkupfalse%
\ \isacommand{have}\isamarkupfalse%
\ {\Cartouche{lemma:imply-completeness-K-2}{19}{0}}\isanewline%
}%
\SNIP{lemma:imply-completeness-K-2-20}{%
\ \ \ \ \isacommand{using}\isamarkupfalse%
\ valid\ \isacommand{by}\isamarkupfalse%
\ auto\isanewline%
}%
\SNIP{lemma:imply-completeness-K-2-21}{%
\ \ \isacommand{then}\isamarkupfalse%
\ \isacommand{show}\isamarkupfalse%
\ False\isanewline%
}%
\SNIP{lemma:imply-completeness-K-2-22-0}{%
{\isacharquery}{\kern0pt}M{\isacharcomma}{\kern0pt}\ {\isacharquery}{\kern0pt}V\ {\isasymTurnstile}\ {\isacharparenleft}{\kern0pt}\isactrlbold {\isasymnot}\ p{\isacharparenright}{\kern0pt}%
}%
\SNIP{lemma:imply-completeness-K-2-22}{%
\ \ \ \ \isacommand{using}\isamarkupfalse%
\ {\Cartouche{lemma:imply-completeness-K-2}{22}{0}}\ \isacommand{by}\isamarkupfalse%
\ simp\isanewline%
}%
\SNIP{lemma:imply-completeness-K-2-23}{%
\isacommand{qed}\isamarkupfalse%
\endisatagproof
{\isafoldproof}%
\isadelimproof
\endisadelimproof
\isadelimdocument
\endisadelimdocument
\isatagdocument%
}%
\SNIP{abbreviation:SystemS4-2-0-0}{%
{\isacharparenleft}{\kern0pt}{\isacharprime}{\kern0pt}i\ {\isacharcolon}{\kern0pt}{\isacharcolon}{\kern0pt}\ countable{\isacharparenright}{\kern0pt}\ fm\ {\isasymRightarrow}\ bool%
}%
\SNIP{abbreviation:SystemS4-2-0}{%
\isacommand{abbreviation}\isamarkupfalse%
\ SystemS{\isadigit{4}}{\isacharunderscore}{\kern0pt}{\isadigit{2}}\ {\isacharcolon}{\kern0pt}{\isacharcolon}{\kern0pt}\ {\Cartouche{abbreviation:SystemS4-2}{0}{0}}\ {\isacharparenleft}{\kern0pt}{\isachardoublequoteopen}{\isasymturnstile}\isactrlsub S\isactrlsub {\isadigit{4}}\isactrlsub {\isadigit{2}}\ {\isacharunderscore}{\kern0pt}{\isachardoublequoteclose}\ {\isacharbrackleft}{\kern0pt}{\isadigit{5}}{\isadigit{0}}{\isacharbrackright}{\kern0pt}\ {\isadigit{5}}{\isadigit{0}}{\isacharparenright}{\kern0pt}\ \isakeyword{where}\isanewline%
}%
\SNIP{abbreviation:SystemS4-2-1-0}{%
{\isasymturnstile}\isactrlsub S\isactrlsub {\isadigit{4}}\isactrlsub {\isadigit{2}}\ p\ {\isasymequiv}\ AxT\ {\isasymoplus}\ Ax{\isadigit{4}}\ {\isasymoplus}\ Ax{\isacharunderscore}{\kern0pt}{\isadigit{2}}\ {\isasymturnstile}\ p%
}%
\SNIP{abbreviation:SystemS4-2-1}{%
\ \ {\Cartouche{abbreviation:SystemS4-2}{1}{0}}%
}%
\SNIP{abbreviation:AxS4-2-0-0}{%
{\isacharparenleft}{\kern0pt}{\isacharprime}{\kern0pt}i\ {\isacharcolon}{\kern0pt}{\isacharcolon}{\kern0pt}\ countable{\isacharparenright}{\kern0pt}\ fm\ {\isasymRightarrow}\ bool%
}%
\SNIP{abbreviation:AxS4-2-0}{%
\isacommand{abbreviation}\isamarkupfalse%
\ AxS{\isadigit{4}}{\isacharunderscore}{\kern0pt}{\isadigit{2}}\ {\isacharcolon}{\kern0pt}{\isacharcolon}{\kern0pt}\ {\Cartouche{abbreviation:AxS4-2}{0}{0}}\ \isakeyword{where}\isanewline%
}%
\SNIP{abbreviation:AxS4-2-1-0}{%
AxS{\isadigit{4}}{\isacharunderscore}{\kern0pt}{\isadigit{2}}\ {\isasymequiv}\ AxT\ {\isasymoplus}\ Ax{\isadigit{4}}\ {\isasymoplus}\ Ax{\isacharunderscore}{\kern0pt}{\isadigit{2}}%
}%
\SNIP{abbreviation:AxS4-2-1}{%
\ \ {\Cartouche{abbreviation:AxS4-2}{1}{0}}%
\isadelimdocument
\endisadelimdocument
\isatagdocument%
}%
\SNIP{abbreviation:w-directed-preorder-0-0}{%
{\isacharparenleft}{\kern0pt}{\isacharprime}{\kern0pt}i{\isacharcomma}{\kern0pt}\ {\isacharprime}{\kern0pt}w{\isacharparenright}{\kern0pt}\ kripke\ {\isasymRightarrow}\ bool%
}%
\SNIP{abbreviation:w-directed-preorder-0}{%
\isacommand{abbreviation}\isamarkupfalse%
\ w{\isacharunderscore}{\kern0pt}directed{\isacharunderscore}{\kern0pt}preorder\ {\isacharcolon}{\kern0pt}{\isacharcolon}{\kern0pt}\ {\Cartouche{abbreviation:w-directed-preorder}{0}{0}}\ \isakeyword{where}\isanewline%
}%
\SNIP{abbreviation:w-directed-preorder-1-0}{%
w{\isacharunderscore}{\kern0pt}directed{\isacharunderscore}{\kern0pt}preorder\ M\ {\isasymequiv}\ reflexive\ M\ {\isasymand}\ transitive\ M\ {\isasymand}\ weakly{\isacharunderscore}{\kern0pt}directed\ M%
}%
\SNIP{abbreviation:w-directed-preorder-1}{%
\ \ {\Cartouche{abbreviation:w-directed-preorder}{1}{0}}%
}%
\SNIP{lemma:soundness-AxS4-2-0-0}{%
AxS{\isadigit{4}}{\isacharunderscore}{\kern0pt}{\isadigit{2}}\ p\ {\isasymLongrightarrow}\ w{\isacharunderscore}{\kern0pt}directed{\isacharunderscore}{\kern0pt}preorder\ M\ {\isasymLongrightarrow}\ w\ {\isasymin}\ {\isasymW}\ M\ {\isasymLongrightarrow}\ M{\isacharcomma}{\kern0pt}\ w\ {\isasymTurnstile}\ p%
}%
\SNIP{lemma:soundness-AxS4-2-0}{%
\isacommand{lemma}\isamarkupfalse%
\ soundness{\isacharunderscore}{\kern0pt}AxS{\isadigit{4}}{\isacharunderscore}{\kern0pt}{\isadigit{2}}{\isacharcolon}{\kern0pt}\ {\Cartouche{lemma:soundness-AxS4-2}{0}{0}}\isanewline%
}%
\SNIP{lemma:soundness-AxS4-2-1}{%
\isadelimproof
\ \ %
\endisadelimproof
\isatagproof
\isacommand{using}\isamarkupfalse%
\ soundness{\isacharunderscore}{\kern0pt}AxT\ soundness{\isacharunderscore}{\kern0pt}Ax{\isadigit{4}}\ soundness{\isacharunderscore}{\kern0pt}Ax{\isacharunderscore}{\kern0pt}{\isadigit{2}}\ \isacommand{by}\isamarkupfalse%
\ metis%
\endisatagproof
{\isafoldproof}%
\isadelimproof%
}%
\SNIP{lemma:soundnessS42-0-0}{%
{\isasymturnstile}\isactrlsub S\isactrlsub {\isadigit{4}}\isactrlsub {\isadigit{2}}\ p\ {\isasymLongrightarrow}\ w{\isacharunderscore}{\kern0pt}directed{\isacharunderscore}{\kern0pt}preorder\ M\ {\isasymLongrightarrow}\ w\ {\isasymin}\ {\isasymW}\ M\ {\isasymLongrightarrow}\ M{\isacharcomma}{\kern0pt}\ w\ {\isasymTurnstile}\ p%
}%
\SNIP{lemma:soundnessS42-0}{%
\isacommand{lemma}\isamarkupfalse%
\ soundness\isactrlsub S\isactrlsub {\isadigit{4}}\isactrlsub {\isadigit{2}}{\isacharcolon}{\kern0pt}\ {\Cartouche{lemma:soundnessS42}{0}{0}}\isanewline%
}%
\SNIP{lemma:soundnessS42-1}{%
\isadelimproof
\ \ %
\endisadelimproof
\isatagproof
\isacommand{using}\isamarkupfalse%
\ soundness\ soundness{\isacharunderscore}{\kern0pt}AxS{\isadigit{4}}{\isacharunderscore}{\kern0pt}{\isadigit{2}}\ \isacommand{{\isachardot}{\kern0pt}}\isamarkupfalse%
\endisatagproof
{\isafoldproof}%
\isadelimproof
\endisadelimproof
\isadelimdocument
\endisadelimdocument
\isatagdocument%
}%
\SNIP{lemma:imply-completeness-S4-2-0}{%
\isacommand{lemma}\isamarkupfalse%
\ imply{\isacharunderscore}{\kern0pt}completeness{\isacharunderscore}{\kern0pt}S{\isadigit{4}}{\isacharunderscore}{\kern0pt}{\isadigit{2}}{\isacharcolon}{\kern0pt}\isanewline%
}%
\SNIP{lemma:imply-completeness-S4-2-1-0}{%
{\isasymforall}{\isacharparenleft}{\kern0pt}M\ {\isacharcolon}{\kern0pt}{\isacharcolon}{\kern0pt}\ {\isacharparenleft}{\kern0pt}{\isacharprime}{\kern0pt}i\ {\isacharcolon}{\kern0pt}{\isacharcolon}{\kern0pt}\ countable{\isacharcomma}{\kern0pt}\ {\isacharprime}{\kern0pt}i\ fm\ set{\isacharparenright}{\kern0pt}\ kripke{\isacharparenright}{\kern0pt}{\isachardot}{\kern0pt}\ {\isasymforall}w\ {\isasymin}\ {\isasymW}\ M{\isachardot}{\kern0pt}\isanewline\ \ \ \ w{\isacharunderscore}{\kern0pt}directed{\isacharunderscore}{\kern0pt}preorder\ M\ {\isasymlongrightarrow}\ {\isacharparenleft}{\kern0pt}{\isasymforall}q\ {\isasymin}\ G{\isachardot}{\kern0pt}\ M{\isacharcomma}{\kern0pt}\ w\ {\isasymTurnstile}\ q{\isacharparenright}{\kern0pt}\ {\isasymlongrightarrow}\ M{\isacharcomma}{\kern0pt}\ w\ {\isasymTurnstile}\ p%
}%
\SNIP{lemma:imply-completeness-S4-2-1}{%
\ \ \isakeyword{assumes}\ valid{\isacharcolon}{\kern0pt}\ {\isacartoucheopen}{\isasymforall}{\isacharparenleft}{\kern0pt}M\ {\isacharcolon}{\kern0pt}{\isacharcolon}{\kern0pt}\ {\isacharparenleft}{\kern0pt}{\isacharprime}{\kern0pt}i\ {\isacharcolon}{\kern0pt}{\isacharcolon}{\kern0pt}\ countable{\isacharcomma}{\kern0pt}\ {\isacharprime}{\kern0pt}i\ fm\ set{\isacharparenright}{\kern0pt}\ kripke{\isacharparenright}{\kern0pt}{\isachardot}{\kern0pt}\ {\isasymforall}w\ {\isasymin}\ {\isasymW}\ M{\isachardot}{\kern0pt}\isanewline%
}%
\SNIP{lemma:imply-completeness-S4-2-2}{%
\ \ \ \ w{\isacharunderscore}{\kern0pt}directed{\isacharunderscore}{\kern0pt}preorder\ M\ {\isasymlongrightarrow}\ {\isacharparenleft}{\kern0pt}{\isasymforall}q\ {\isasymin}\ G{\isachardot}{\kern0pt}\ M{\isacharcomma}{\kern0pt}\ w\ {\isasymTurnstile}\ q{\isacharparenright}{\kern0pt}\ {\isasymlongrightarrow}\ M{\isacharcomma}{\kern0pt}\ w\ {\isasymTurnstile}\ p{\isacartoucheclose}\isanewline%
}%
\SNIP{lemma:imply-completeness-S4-2-3-0}{%
{\isasymexists}qs{\isachardot}{\kern0pt}\ set\ qs\ {\isasymsubseteq}\ G\ {\isasymand}\ {\isacharparenleft}{\kern0pt}AxS{\isadigit{4}}{\isacharunderscore}{\kern0pt}{\isadigit{2}}\ {\isasymturnstile}\ imply\ qs\ p{\isacharparenright}{\kern0pt}%
}%
\SNIP{lemma:imply-completeness-S4-2-3}{%
\ \ \isakeyword{shows}\ {\Cartouche{lemma:imply-completeness-S4-2}{3}{0}}\isanewline%
}%
\SNIP{lemma:imply-completeness-S4-2-4}{%
\isadelimproof
\endisadelimproof
\isatagproof
\isacommand{proof}\isamarkupfalse%
\ {\isacharparenleft}{\kern0pt}rule\ ccontr{\isacharparenright}{\kern0pt}\isanewline%
}%
\SNIP{lemma:imply-completeness-S4-2-5-0}{%
{\isasymnexists}qs{\isachardot}{\kern0pt}\ set\ qs\ {\isasymsubseteq}\ G\ {\isasymand}\ AxS{\isadigit{4}}{\isacharunderscore}{\kern0pt}{\isadigit{2}}\ {\isasymturnstile}\ imply\ qs\ p%
}%
\SNIP{lemma:imply-completeness-S4-2-5}{%
\ \ \isacommand{assume}\isamarkupfalse%
\ {\Cartouche{lemma:imply-completeness-S4-2}{5}{0}}\isanewline%
}%
\SNIP{lemma:imply-completeness-S4-2-6-0}{%
{\isasymforall}qs{\isachardot}{\kern0pt}\ set\ qs\ {\isasymsubseteq}\ G\ {\isasymlongrightarrow}\ {\isasymnot}\ AxS{\isadigit{4}}{\isacharunderscore}{\kern0pt}{\isadigit{2}}\ {\isasymturnstile}\ imply\ {\isacharparenleft}{\kern0pt}{\isacharparenleft}{\kern0pt}\isactrlbold {\isasymnot}\ p{\isacharparenright}{\kern0pt}\ {\isacharhash}{\kern0pt}\ qs{\isacharparenright}{\kern0pt}\ \isactrlbold {\isasymbottom}%
}%
\SNIP{lemma:imply-completeness-S4-2-6}{%
\ \ \isacommand{then}\isamarkupfalse%
\ \isacommand{have}\isamarkupfalse%
\ {\isacharasterisk}{\kern0pt}{\isacharcolon}{\kern0pt}\ {\Cartouche{lemma:imply-completeness-S4-2}{6}{0}}\isanewline%
}%
\SNIP{lemma:imply-completeness-S4-2-7}{%
\ \ \ \ \isacommand{using}\isamarkupfalse%
\ K{\isacharunderscore}{\kern0pt}Boole\ \isacommand{by}\isamarkupfalse%
\ blast\isanewline%
}%
\SNIP{lemma:imply-completeness-S4-2-8}{%
\isanewline%
}%
\SNIP{lemma:imply-completeness-S4-2-9-0}{%
{\isacharbraceleft}{\kern0pt}\isactrlbold {\isasymnot}\ p{\isacharbraceright}{\kern0pt}\ {\isasymunion}\ G%
}%
\SNIP{lemma:imply-completeness-S4-2-9}{%
\ \ \isacommand{let}\isamarkupfalse%
\ {\isacharquery}{\kern0pt}S\ {\isacharequal}{\kern0pt}\ {\Cartouche{lemma:imply-completeness-S4-2}{9}{0}}\isanewline%
}%
\SNIP{lemma:imply-completeness-S4-2-10-0}{%
Extend\ AxS{\isadigit{4}}{\isacharunderscore}{\kern0pt}{\isadigit{2}}\ {\isacharquery}{\kern0pt}S\ from{\isacharunderscore}{\kern0pt}nat%
}%
\SNIP{lemma:imply-completeness-S4-2-10}{%
\ \ \isacommand{let}\isamarkupfalse%
\ {\isacharquery}{\kern0pt}V\ {\isacharequal}{\kern0pt}\ {\Cartouche{lemma:imply-completeness-S4-2}{10}{0}}\isanewline%
}%
\SNIP{lemma:imply-completeness-S4-2-11-0}{%
Kripke\ {\isacharparenleft}{\kern0pt}mcss\ AxS{\isadigit{4}}{\isacharunderscore}{\kern0pt}{\isadigit{2}}{\isacharparenright}{\kern0pt}\ pi\ {\isacharparenleft}{\kern0pt}reach\ AxS{\isadigit{4}}{\isacharunderscore}{\kern0pt}{\isadigit{2}}{\isacharparenright}{\kern0pt}%
}%
\SNIP{lemma:imply-completeness-S4-2-11}{%
\ \ \isacommand{let}\isamarkupfalse%
\ {\isacharquery}{\kern0pt}M\ {\isacharequal}{\kern0pt}\ {\Cartouche{lemma:imply-completeness-S4-2}{11}{0}}\isanewline%
}%
\SNIP{lemma:imply-completeness-S4-2-12}{%
\isanewline%
}%
\SNIP{lemma:imply-completeness-S4-2-13-0}{%
consistent\ AxS{\isadigit{4}}{\isacharunderscore}{\kern0pt}{\isadigit{2}}\ {\isacharquery}{\kern0pt}S%
}%
\SNIP{lemma:imply-completeness-S4-2-13}{%
\ \ \isacommand{have}\isamarkupfalse%
\ {\Cartouche{lemma:imply-completeness-S4-2}{13}{0}}\isanewline%
}%
\SNIP{lemma:imply-completeness-S4-2-14}{%
\ \ \ \ \isacommand{using}\isamarkupfalse%
\ {\isacharasterisk}{\kern0pt}\ \isacommand{by}\isamarkupfalse%
\ {\isacharparenleft}{\kern0pt}metis\ {\isacharparenleft}{\kern0pt}no{\isacharunderscore}{\kern0pt}types{\isacharcomma}{\kern0pt}\ lifting{\isacharparenright}{\kern0pt}\ K{\isacharunderscore}{\kern0pt}imply{\isacharunderscore}{\kern0pt}Cons\ consistent{\isacharunderscore}{\kern0pt}def\ inconsistent{\isacharunderscore}{\kern0pt}subset{\isacharparenright}{\kern0pt}\isanewline%
}%
\SNIP{lemma:imply-completeness-S4-2-15}{%
\ \ \isacommand{then}\isamarkupfalse%
\ \isacommand{have}\isamarkupfalse%
\isanewline%
}%
\SNIP{lemma:imply-completeness-S4-2-16-0}{%
{\isacharquery}{\kern0pt}M{\isacharcomma}{\kern0pt}\ {\isacharquery}{\kern0pt}V\ {\isasymTurnstile}\ {\isacharparenleft}{\kern0pt}\isactrlbold {\isasymnot}\ p{\isacharparenright}{\kern0pt}%
}%
\SNIP{lemma:imply-completeness-S4-2-16-1}{%
{\isasymforall}q\ {\isasymin}\ G{\isachardot}{\kern0pt}\ {\isacharquery}{\kern0pt}M{\isacharcomma}{\kern0pt}\ {\isacharquery}{\kern0pt}V\ {\isasymTurnstile}\ q%
}%
\SNIP{lemma:imply-completeness-S4-2-16}{%
\ \ \ \ {\Cartouche{lemma:imply-completeness-S4-2}{16}{0}}\ {\Cartouche{lemma:imply-completeness-S4-2}{16}{1}}\isanewline%
}%
\SNIP{lemma:imply-completeness-S4-2-17-0}{%
consistent\ AxS{\isadigit{4}}{\isacharunderscore}{\kern0pt}{\isadigit{2}}\ {\isacharquery}{\kern0pt}V%
}%
\SNIP{lemma:imply-completeness-S4-2-17-1}{%
maximal\ AxS{\isadigit{4}}{\isacharunderscore}{\kern0pt}{\isadigit{2}}\ {\isacharquery}{\kern0pt}V%
}%
\SNIP{lemma:imply-completeness-S4-2-17}{%
\ \ \ \ {\Cartouche{lemma:imply-completeness-S4-2}{17}{0}}\ {\Cartouche{lemma:imply-completeness-S4-2}{17}{1}}\isanewline%
}%
\SNIP{lemma:imply-completeness-S4-2-18}{%
\ \ \ \ \isacommand{using}\isamarkupfalse%
\ canonical{\isacharunderscore}{\kern0pt}model\ \isacommand{unfolding}\isamarkupfalse%
\ list{\isacharunderscore}{\kern0pt}all{\isacharunderscore}{\kern0pt}def\ \isacommand{by}\isamarkupfalse%
\ fastforce{\isacharplus}{\kern0pt}\isanewline%
}%
\SNIP{lemma:imply-completeness-S4-2-19-0}{%
w{\isacharunderscore}{\kern0pt}directed{\isacharunderscore}{\kern0pt}preorder\ {\isacharquery}{\kern0pt}M%
}%
\SNIP{lemma:imply-completeness-S4-2-19}{%
\ \ \isacommand{moreover}\isamarkupfalse%
\ \isacommand{have}\isamarkupfalse%
\ {\Cartouche{lemma:imply-completeness-S4-2}{19}{0}}\isanewline%
}%
\SNIP{lemma:imply-completeness-S4-2-20}{%
\ \ \ \ \isacommand{by}\isamarkupfalse%
\ {\isacharparenleft}{\kern0pt}simp\ add{\isacharcolon}{\kern0pt}\ mcs\isactrlsub T{\isacharunderscore}{\kern0pt}reflexive\ mcs\isactrlsub {\isacharunderscore}{\kern0pt}\isactrlsub {\isadigit{2}}{\isacharunderscore}{\kern0pt}weakly{\isacharunderscore}{\kern0pt}directed\ mcs\isactrlsub K\isactrlsub {\isadigit{4}}{\isacharunderscore}{\kern0pt}transitive{\isacharparenright}{\kern0pt}\isanewline%
}%
\SNIP{lemma:imply-completeness-S4-2-21-0}{%
{\isacharquery}{\kern0pt}M{\isacharcomma}{\kern0pt}\ {\isacharquery}{\kern0pt}V\ {\isasymTurnstile}\ p%
}%
\SNIP{lemma:imply-completeness-S4-2-21}{%
\ \ \isacommand{ultimately}\isamarkupfalse%
\ \isacommand{have}\isamarkupfalse%
\ {\Cartouche{lemma:imply-completeness-S4-2}{21}{0}}\isanewline%
}%
\SNIP{lemma:imply-completeness-S4-2-22}{%
\ \ \ \ \isacommand{using}\isamarkupfalse%
\ valid\ \isacommand{by}\isamarkupfalse%
\ auto\isanewline%
}%
\SNIP{lemma:imply-completeness-S4-2-23}{%
\ \ \isacommand{then}\isamarkupfalse%
\ \isacommand{show}\isamarkupfalse%
\ False\isanewline%
}%
\SNIP{lemma:imply-completeness-S4-2-24-0}{%
{\isacharquery}{\kern0pt}M{\isacharcomma}{\kern0pt}\ {\isacharquery}{\kern0pt}V\ {\isasymTurnstile}\ {\isacharparenleft}{\kern0pt}\isactrlbold {\isasymnot}\ p{\isacharparenright}{\kern0pt}%
}%
\SNIP{lemma:imply-completeness-S4-2-24}{%
\ \ \ \ \isacommand{using}\isamarkupfalse%
\ {\Cartouche{lemma:imply-completeness-S4-2}{24}{0}}\ \isacommand{by}\isamarkupfalse%
\ simp\isanewline%
}%
\SNIP{lemma:imply-completeness-S4-2-25}{%
\isacommand{qed}\isamarkupfalse%
\endisatagproof
{\isafoldproof}%
\isadelimproof%
}%
\SNIP{lemma:completenessS42-0}{%
\isacommand{lemma}\isamarkupfalse%
\ completeness\isactrlsub S\isactrlsub {\isadigit{4}}\isactrlsub {\isadigit{2}}{\isacharcolon}{\kern0pt}\isanewline%
}%
\SNIP{lemma:completenessS42-1-0}{%
{\isasymforall}{\isacharparenleft}{\kern0pt}M\ {\isacharcolon}{\kern0pt}{\isacharcolon}{\kern0pt}\ {\isacharparenleft}{\kern0pt}{\isacharprime}{\kern0pt}i\ {\isacharcolon}{\kern0pt}{\isacharcolon}{\kern0pt}\ countable{\isacharcomma}{\kern0pt}\ {\isacharprime}{\kern0pt}i\ fm\ set{\isacharparenright}{\kern0pt}\ kripke{\isacharparenright}{\kern0pt}{\isachardot}{\kern0pt}\ {\isasymforall}w\ {\isasymin}\ {\isasymW}\ M{\isachardot}{\kern0pt}\ w{\isacharunderscore}{\kern0pt}directed{\isacharunderscore}{\kern0pt}preorder\ M\ {\isasymlongrightarrow}\ M{\isacharcomma}{\kern0pt}\ w\ {\isasymTurnstile}\ p%
}%
\SNIP{lemma:completenessS42-1}{%
\ \ \isakeyword{assumes}\ {\Cartouche{lemma:completenessS42}{1}{0}}\isanewline%
}%
\SNIP{lemma:completenessS42-2-0}{%
{\isasymturnstile}\isactrlsub S\isactrlsub {\isadigit{4}}\isactrlsub {\isadigit{2}}\ p%
}%
\SNIP{lemma:completenessS42-2}{%
\ \ \isakeyword{shows}\ {\Cartouche{lemma:completenessS42}{2}{0}}\isanewline%
}%
\SNIP{lemma:completenessS42-3-0}{%
{\isacharbraceleft}{\kern0pt}{\isacharbraceright}{\kern0pt}%
}%
\SNIP{lemma:completenessS42-3}{%
\isadelimproof
\ \ %
\endisadelimproof
\isatagproof
\isacommand{using}\isamarkupfalse%
\ assms\ imply{\isacharunderscore}{\kern0pt}completeness{\isacharunderscore}{\kern0pt}S{\isadigit{4}}{\isacharunderscore}{\kern0pt}{\isadigit{2}}{\isacharbrackleft}{\kern0pt}\isakeyword{where}\ G{\isacharequal}{\kern0pt}{\Cartouche{lemma:completenessS42}{3}{0}}{\isacharbrackright}{\kern0pt}\ \isacommand{by}\isamarkupfalse%
\ auto%
\endisatagproof
{\isafoldproof}%
\isadelimproof%
}%
\SNIP{abbreviation:validS42-0-0}{%
valid\isactrlsub S\isactrlsub {\isadigit{4}}\isactrlsub {\isadigit{2}}\ p\ {\isasymequiv}\ {\isasymforall}{\isacharparenleft}{\kern0pt}M\ {\isacharcolon}{\kern0pt}{\isacharcolon}{\kern0pt}\ {\isacharparenleft}{\kern0pt}nat{\isacharcomma}{\kern0pt}\ nat\ fm\ set{\isacharparenright}{\kern0pt}\ kripke{\isacharparenright}{\kern0pt}{\isachardot}{\kern0pt}\ {\isasymforall}w\ {\isasymin}\ {\isasymW}\ M{\isachardot}{\kern0pt}\ w{\isacharunderscore}{\kern0pt}directed{\isacharunderscore}{\kern0pt}preorder\ M\ {\isasymlongrightarrow}\ M{\isacharcomma}{\kern0pt}\ w\ {\isasymTurnstile}\ p%
}%
\SNIP{abbreviation:validS42-0}{%
\isacommand{abbreviation}\isamarkupfalse%
\ {\Cartouche{abbreviation:validS42}{0}{0}}%
}%
\SNIP{theorem:mainS42-0-0}{%
valid\isactrlsub S\isactrlsub {\isadigit{4}}\isactrlsub {\isadigit{2}}\ p\ {\isasymlongleftrightarrow}\ {\isasymturnstile}\isactrlsub S\isactrlsub {\isadigit{4}}\isactrlsub {\isadigit{2}}\ p%
}%
\SNIP{theorem:mainS42-0}{%
\isacommand{theorem}\isamarkupfalse%
\ main\isactrlsub S\isactrlsub {\isadigit{4}}\isactrlsub {\isadigit{2}}{\isacharcolon}{\kern0pt}\ {\Cartouche{theorem:mainS42}{0}{0}}\isanewline%
}%
\SNIP{theorem:mainS42-1}{%
\isadelimproof
\ \ %
\endisadelimproof
\isatagproof
\isacommand{using}\isamarkupfalse%
\ soundness\isactrlsub S\isactrlsub {\isadigit{4}}\isactrlsub {\isadigit{2}}\ completeness\isactrlsub S\isactrlsub {\isadigit{4}}\isactrlsub {\isadigit{2}}\ \isacommand{by}\isamarkupfalse%
\ fast%
\endisatagproof
{\isafoldproof}%
\isadelimproof%
}%
\SNIP{corollary:78fa5f98351da640-0}{%
\isacommand{corollary}\isamarkupfalse%
\isanewline%
}%
\SNIP{corollary:78fa5f98351da640-1-0}{%
w{\isacharunderscore}{\kern0pt}directed{\isacharunderscore}{\kern0pt}preorder\ M%
}%
\SNIP{corollary:78fa5f98351da640-1-1}{%
w\ {\isasymin}\ {\isasymW}\ M%
}%
\SNIP{corollary:78fa5f98351da640-1}{%
\ \ \isakeyword{assumes}\ {\Cartouche{corollary:78fa5f98351da640}{1}{0}}\ {\Cartouche{corollary:78fa5f98351da640}{1}{1}}\isanewline%
}%
\SNIP{corollary:78fa5f98351da640-2-0}{%
valid\isactrlsub S\isactrlsub {\isadigit{4}}\isactrlsub {\isadigit{2}}\ p\ {\isasymlongrightarrow}\ M{\isacharcomma}{\kern0pt}\ w\ {\isasymTurnstile}\ p%
}%
\SNIP{corollary:78fa5f98351da640-2}{%
\ \ \isakeyword{shows}\ {\Cartouche{corollary:78fa5f98351da640}{2}{0}}\isanewline%
}%
\SNIP{corollary:78fa5f98351da640-3}{%
\isadelimproof
\ \ %
\endisadelimproof
\isatagproof
\isacommand{using}\isamarkupfalse%
\ assms\ soundness\isactrlsub S\isactrlsub {\isadigit{4}}\isactrlsub {\isadigit{2}}\ completeness\isactrlsub S\isactrlsub {\isadigit{4}}\isactrlsub {\isadigit{2}}\ \isacommand{by}\isamarkupfalse%
\ fast%
\endisatagproof
{\isafoldproof}%
\isadelimproof
\endisadelimproof
\isadelimdocument
\endisadelimdocument
\isatagdocument%
}%
\SNIP{abbreviation:DoubleImp-0}{%
\isacommand{abbreviation}\isamarkupfalse%
\ DoubleImp\ {\isacharparenleft}{\kern0pt}\isakeyword{infixr}\ {\isachardoublequoteopen}\isactrlbold {\isasymlongleftrightarrow}{\isachardoublequoteclose}\ {\isadigit{2}}{\isadigit{5}}{\isacharparenright}{\kern0pt}\ \isakeyword{where}\isanewline%
}%
\SNIP{abbreviation:DoubleImp-1-0}{%
{\isacharparenleft}{\kern0pt}p\isactrlbold {\isasymlongleftrightarrow}q{\isacharparenright}{\kern0pt}\ {\isasymequiv}\ {\isacharparenleft}{\kern0pt}{\isacharparenleft}{\kern0pt}p\ \isactrlbold {\isasymlongrightarrow}\ q{\isacharparenright}{\kern0pt}\ \isactrlbold {\isasymand}\ {\isacharparenleft}{\kern0pt}q\ \isactrlbold {\isasymlongrightarrow}\ p{\isacharparenright}{\kern0pt}{\isacharparenright}{\kern0pt}%
}%
\SNIP{abbreviation:DoubleImp-1}{%
\ \ {\Cartouche{abbreviation:DoubleImp}{1}{0}}%
}%
\SNIP{inductive:System-topoS4-0-0}{%
{\isacharprime}{\kern0pt}i\ fm\ {\isasymRightarrow}\ bool%
}%
\SNIP{inductive:System-topoS4-0}{%
\isacommand{inductive}\isamarkupfalse%
\ System{\isacharunderscore}{\kern0pt}topoS{\isadigit{4}}\ {\isacharcolon}{\kern0pt}{\isacharcolon}{\kern0pt}\ {\Cartouche{inductive:System-topoS4}{0}{0}}\ {\isacharparenleft}{\kern0pt}{\isachardoublequoteopen}{\isasymturnstile}\isactrlsub T\isactrlsub o\isactrlsub p\ {\isacharunderscore}{\kern0pt}{\isachardoublequoteclose}\ {\isacharbrackleft}{\kern0pt}{\isadigit{5}}{\isadigit{0}}{\isacharbrackright}{\kern0pt}\ {\isadigit{5}}{\isadigit{0}}{\isacharparenright}{\kern0pt}\ \isakeyword{where}\isanewline%
}%
\SNIP{inductive:System-topoS4-1-0}{%
tautology\ p\ {\isasymLongrightarrow}\ {\isasymturnstile}\isactrlsub T\isactrlsub o\isactrlsub p\ p%
}%
\SNIP{inductive:System-topoS4-1}{%
\ \ A{\isadigit{1}}{\isacharprime}{\kern0pt}{\isacharcolon}{\kern0pt}\ {\Cartouche{inductive:System-topoS4}{1}{0}}\isanewline%
}%
\SNIP{inductive:System-topoS4-2-0}{%
{\isasymturnstile}\isactrlsub T\isactrlsub o\isactrlsub p\ {\isacharparenleft}{\kern0pt}{\isacharparenleft}{\kern0pt}K\ i\ {\isacharparenleft}{\kern0pt}p\ \isactrlbold {\isasymand}\ q{\isacharparenright}{\kern0pt}{\isacharparenright}{\kern0pt}\ \isactrlbold {\isasymlongleftrightarrow}\ {\isacharparenleft}{\kern0pt}{\isacharparenleft}{\kern0pt}K\ i\ p{\isacharparenright}{\kern0pt}\ \isactrlbold {\isasymand}\ K\ i\ q{\isacharparenright}{\kern0pt}{\isacharparenright}{\kern0pt}%
}%
\SNIP{inductive:System-topoS4-2}{%
{\isacharbar}{\kern0pt}\ AR{\isacharcolon}{\kern0pt}\ {\Cartouche{inductive:System-topoS4}{2}{0}}\isanewline%
}%
\SNIP{inductive:System-topoS4-3-0}{%
{\isasymturnstile}\isactrlsub T\isactrlsub o\isactrlsub p\ {\isacharparenleft}{\kern0pt}K\ i\ p\ \isactrlbold {\isasymlongrightarrow}\ p{\isacharparenright}{\kern0pt}%
}%
\SNIP{inductive:System-topoS4-3}{%
{\isacharbar}{\kern0pt}\ AT{\isacharprime}{\kern0pt}{\isacharcolon}{\kern0pt}\ {\Cartouche{inductive:System-topoS4}{3}{0}}\isanewline%
}%
\SNIP{inductive:System-topoS4-4-0}{%
{\isasymturnstile}\isactrlsub T\isactrlsub o\isactrlsub p\ {\isacharparenleft}{\kern0pt}K\ i\ p\ \isactrlbold {\isasymlongrightarrow}\ K\ i\ {\isacharparenleft}{\kern0pt}K\ i\ p{\isacharparenright}{\kern0pt}{\isacharparenright}{\kern0pt}%
}%
\SNIP{inductive:System-topoS4-4}{%
{\isacharbar}{\kern0pt}\ A{\isadigit{4}}{\isacharprime}{\kern0pt}{\isacharcolon}{\kern0pt}\ {\Cartouche{inductive:System-topoS4}{4}{0}}\isanewline%
}%
\SNIP{inductive:System-topoS4-5-0}{%
{\isasymturnstile}\isactrlsub T\isactrlsub o\isactrlsub p\ K\ i\ \isactrlbold {\isasymtop}%
}%
\SNIP{inductive:System-topoS4-5}{%
{\isacharbar}{\kern0pt}\ AN{\isacharcolon}{\kern0pt}\ {\Cartouche{inductive:System-topoS4}{5}{0}}\isanewline%
}%
\SNIP{inductive:System-topoS4-6-0}{%
{\isasymturnstile}\isactrlsub T\isactrlsub o\isactrlsub p\ p\ {\isasymLongrightarrow}\ {\isasymturnstile}\isactrlsub T\isactrlsub o\isactrlsub p\ {\isacharparenleft}{\kern0pt}p\ \isactrlbold {\isasymlongrightarrow}\ q{\isacharparenright}{\kern0pt}\ {\isasymLongrightarrow}\ {\isasymturnstile}\isactrlsub T\isactrlsub o\isactrlsub p\ q%
}%
\SNIP{inductive:System-topoS4-6}{%
{\isacharbar}{\kern0pt}\ R{\isadigit{1}}{\isacharprime}{\kern0pt}{\isacharcolon}{\kern0pt}\ {\Cartouche{inductive:System-topoS4}{6}{0}}\isanewline%
}%
\SNIP{inductive:System-topoS4-7-0}{%
{\isasymturnstile}\isactrlsub T\isactrlsub o\isactrlsub p\ {\isacharparenleft}{\kern0pt}p\ \isactrlbold {\isasymlongrightarrow}\ q{\isacharparenright}{\kern0pt}\ {\isasymLongrightarrow}\ {\isasymturnstile}\isactrlsub T\isactrlsub o\isactrlsub p\ {\isacharparenleft}{\kern0pt}{\isacharparenleft}{\kern0pt}K\ i\ p{\isacharparenright}{\kern0pt}\ \isactrlbold {\isasymlongrightarrow}\ K\ i\ q{\isacharparenright}{\kern0pt}%
}%
\SNIP{inductive:System-topoS4-7}{%
{\isacharbar}{\kern0pt}\ RM{\isacharcolon}{\kern0pt}\ {\Cartouche{inductive:System-topoS4}{7}{0}}%
}%
\SNIP{lemma:topoS4-trans-0-0}{%
{\isasymturnstile}\isactrlsub T\isactrlsub o\isactrlsub p\ {\isacharparenleft}{\kern0pt}{\isacharparenleft}{\kern0pt}p\ \isactrlbold {\isasymlongrightarrow}\ q{\isacharparenright}{\kern0pt}\ \isactrlbold {\isasymlongrightarrow}\ {\isacharparenleft}{\kern0pt}q\ \isactrlbold {\isasymlongrightarrow}\ r{\isacharparenright}{\kern0pt}\ \isactrlbold {\isasymlongrightarrow}\ p\ \isactrlbold {\isasymlongrightarrow}\ r{\isacharparenright}{\kern0pt}%
}%
\SNIP{lemma:topoS4-trans-0}{%
\isacommand{lemma}\isamarkupfalse%
\ topoS{\isadigit{4}}{\isacharunderscore}{\kern0pt}trans{\isacharcolon}{\kern0pt}\ {\Cartouche{lemma:topoS4-trans}{0}{0}}\isanewline%
}%
\SNIP{lemma:topoS4-trans-1}{%
\isadelimproof
\ \ %
\endisadelimproof
\isatagproof
\isacommand{by}\isamarkupfalse%
\ {\isacharparenleft}{\kern0pt}simp\ add{\isacharcolon}{\kern0pt}\ A{\isadigit{1}}{\isacharprime}{\kern0pt}{\isacharparenright}{\kern0pt}%
\endisatagproof
{\isafoldproof}%
\isadelimproof%
}%
\SNIP{lemma:topoS4-conjElim-0-0}{%
{\isasymturnstile}\isactrlsub T\isactrlsub o\isactrlsub p\ {\isacharparenleft}{\kern0pt}p\ \isactrlbold {\isasymand}\ q\ \isactrlbold {\isasymlongrightarrow}\ q{\isacharparenright}{\kern0pt}%
}%
\SNIP{lemma:topoS4-conjElim-0}{%
\isacommand{lemma}\isamarkupfalse%
\ topoS{\isadigit{4}}{\isacharunderscore}{\kern0pt}conjElim{\isacharcolon}{\kern0pt}\ {\Cartouche{lemma:topoS4-conjElim}{0}{0}}\isanewline%
}%
\SNIP{lemma:topoS4-conjElim-1}{%
\isadelimproof
\ \ %
\endisadelimproof
\isatagproof
\isacommand{by}\isamarkupfalse%
\ {\isacharparenleft}{\kern0pt}simp\ add{\isacharcolon}{\kern0pt}\ A{\isadigit{1}}{\isacharprime}{\kern0pt}{\isacharparenright}{\kern0pt}%
\endisatagproof
{\isafoldproof}%
\isadelimproof%
}%
\SNIP{lemma:topoS4-AxK-0-0}{%
{\isasymturnstile}\isactrlsub T\isactrlsub o\isactrlsub p\ {\isacharparenleft}{\kern0pt}K\ i\ p\ \isactrlbold {\isasymand}\ K\ i\ {\isacharparenleft}{\kern0pt}p\ \isactrlbold {\isasymlongrightarrow}\ q{\isacharparenright}{\kern0pt}\ \isactrlbold {\isasymlongrightarrow}\ K\ i\ q{\isacharparenright}{\kern0pt}%
}%
\SNIP{lemma:topoS4-AxK-0}{%
\isacommand{lemma}\isamarkupfalse%
\ topoS{\isadigit{4}}{\isacharunderscore}{\kern0pt}AxK{\isacharcolon}{\kern0pt}\ {\Cartouche{lemma:topoS4-AxK}{0}{0}}\isanewline%
}%
\SNIP{lemma:topoS4-AxK-1}{%
\isadelimproof
\endisadelimproof
\isatagproof
\isacommand{proof}\isamarkupfalse%
\ {\isacharminus}{\kern0pt}\isanewline%
}%
\SNIP{lemma:topoS4-AxK-2-0}{%
{\isasymturnstile}\isactrlsub T\isactrlsub o\isactrlsub p\ {\isacharparenleft}{\kern0pt}{\isacharparenleft}{\kern0pt}p\ \isactrlbold {\isasymand}\ {\isacharparenleft}{\kern0pt}p\ \isactrlbold {\isasymlongrightarrow}\ q{\isacharparenright}{\kern0pt}{\isacharparenright}{\kern0pt}\ \isactrlbold {\isasymlongrightarrow}\ q{\isacharparenright}{\kern0pt}%
}%
\SNIP{lemma:topoS4-AxK-2}{%
\ \ \isacommand{have}\isamarkupfalse%
\ {\Cartouche{lemma:topoS4-AxK}{2}{0}}\isanewline%
}%
\SNIP{lemma:topoS4-AxK-3}{%
\ \ \ \ \isacommand{using}\isamarkupfalse%
\ A{\isadigit{1}}{\isacharprime}{\kern0pt}\ \isacommand{by}\isamarkupfalse%
\ force\isanewline%
}%
\SNIP{lemma:topoS4-AxK-4-0}{%
{\isasymturnstile}\isactrlsub T\isactrlsub o\isactrlsub p\ {\isacharparenleft}{\kern0pt}K\ i\ {\isacharparenleft}{\kern0pt}p\ \isactrlbold {\isasymand}\ {\isacharparenleft}{\kern0pt}p\ \isactrlbold {\isasymlongrightarrow}\ q{\isacharparenright}{\kern0pt}{\isacharparenright}{\kern0pt}\ \isactrlbold {\isasymlongrightarrow}\ K\ i\ q{\isacharparenright}{\kern0pt}%
}%
\SNIP{lemma:topoS4-AxK-4}{%
\ \ \isacommand{then}\isamarkupfalse%
\ \isacommand{have}\isamarkupfalse%
\ {\isacharasterisk}{\kern0pt}{\isacharcolon}{\kern0pt}\ {\Cartouche{lemma:topoS4-AxK}{4}{0}}\isanewline%
}%
\SNIP{lemma:topoS4-AxK-5}{%
\ \ \ \ \isacommand{using}\isamarkupfalse%
\ RM\ \isacommand{by}\isamarkupfalse%
\ fastforce\isanewline%
}%
\SNIP{lemma:topoS4-AxK-6-0}{%
{\isasymturnstile}\isactrlsub T\isactrlsub o\isactrlsub p\ {\isacharparenleft}{\kern0pt}K\ i\ p\ \isactrlbold {\isasymand}\ K\ i\ {\isacharparenleft}{\kern0pt}p\ \isactrlbold {\isasymlongrightarrow}\ q{\isacharparenright}{\kern0pt}\ \isactrlbold {\isasymlongrightarrow}\ K\ i\ {\isacharparenleft}{\kern0pt}p\ \isactrlbold {\isasymand}\ {\isacharparenleft}{\kern0pt}p\ \isactrlbold {\isasymlongrightarrow}\ q{\isacharparenright}{\kern0pt}{\isacharparenright}{\kern0pt}{\isacharparenright}{\kern0pt}%
}%
\SNIP{lemma:topoS4-AxK-6}{%
\ \ \isacommand{then}\isamarkupfalse%
\ \isacommand{have}\isamarkupfalse%
\ {\Cartouche{lemma:topoS4-AxK}{6}{0}}\isanewline%
}%
\SNIP{lemma:topoS4-AxK-7}{%
\ \ \ \ \isacommand{using}\isamarkupfalse%
\ AR\ topoS{\isadigit{4}}{\isacharunderscore}{\kern0pt}conjElim\ System{\isacharunderscore}{\kern0pt}topoS{\isadigit{4}}{\isachardot}{\kern0pt}simps\ \isacommand{by}\isamarkupfalse%
\ fast\isanewline%
}%
\SNIP{lemma:topoS4-AxK-8}{%
\ \ \isacommand{then}\isamarkupfalse%
\ \isacommand{show}\isamarkupfalse%
\ {\isacharquery}{\kern0pt}thesis\isanewline%
}%
\SNIP{lemma:topoS4-AxK-9}{%
\ \ \ \ \isacommand{by}\isamarkupfalse%
\ {\isacharparenleft}{\kern0pt}meson\ {\isachardoublequoteopen}{\isacharasterisk}{\kern0pt}{\isachardoublequoteclose}\ System{\isacharunderscore}{\kern0pt}topoS{\isadigit{4}}{\isachardot}{\kern0pt}R{\isadigit{1}}{\isacharprime}{\kern0pt}\ topoS{\isadigit{4}}{\isacharunderscore}{\kern0pt}trans{\isacharparenright}{\kern0pt}\isanewline%
}%
\SNIP{lemma:topoS4-AxK-10}{%
\isacommand{qed}\isamarkupfalse%
\endisatagproof
{\isafoldproof}%
\isadelimproof%
}%
\SNIP{lemma:topoS4-NecR-0}{%
\isacommand{lemma}\isamarkupfalse%
\ topoS{\isadigit{4}}{\isacharunderscore}{\kern0pt}NecR{\isacharcolon}{\kern0pt}\isanewline%
}%
\SNIP{lemma:topoS4-NecR-1-0}{%
{\isasymturnstile}\isactrlsub T\isactrlsub o\isactrlsub p\ p%
}%
\SNIP{lemma:topoS4-NecR-1}{%
\ \ \isakeyword{assumes}\ {\Cartouche{lemma:topoS4-NecR}{1}{0}}\isanewline%
}%
\SNIP{lemma:topoS4-NecR-2-0}{%
{\isasymturnstile}\isactrlsub T\isactrlsub o\isactrlsub p\ K\ i\ p%
}%
\SNIP{lemma:topoS4-NecR-2}{%
\ \ \isakeyword{shows}{\Cartouche{lemma:topoS4-NecR}{2}{0}}\isanewline%
}%
\SNIP{lemma:topoS4-NecR-3}{%
\isadelimproof
\endisadelimproof
\isatagproof
\isacommand{proof}\isamarkupfalse%
\ {\isacharminus}{\kern0pt}\isanewline%
}%
\SNIP{lemma:topoS4-NecR-4-0}{%
{\isasymturnstile}\isactrlsub T\isactrlsub o\isactrlsub p\ {\isacharparenleft}{\kern0pt}\isactrlbold {\isasymtop}\ \isactrlbold {\isasymlongrightarrow}\ p{\isacharparenright}{\kern0pt}%
}%
\SNIP{lemma:topoS4-NecR-4}{%
\ \ \isacommand{have}\isamarkupfalse%
\ {\Cartouche{lemma:topoS4-NecR}{4}{0}}\isanewline%
}%
\SNIP{lemma:topoS4-NecR-5}{%
\ \ \ \ \isacommand{using}\isamarkupfalse%
\ assms\ \isacommand{by}\isamarkupfalse%
\ {\isacharparenleft}{\kern0pt}metis\ System{\isacharunderscore}{\kern0pt}topoS{\isadigit{4}}{\isachardot}{\kern0pt}A{\isadigit{1}}{\isacharprime}{\kern0pt}\ System{\isacharunderscore}{\kern0pt}topoS{\isadigit{4}}{\isachardot}{\kern0pt}R{\isadigit{1}}{\isacharprime}{\kern0pt}\ conjunct{\isachardot}{\kern0pt}simps{\isacharparenleft}{\kern0pt}{\isadigit{1}}{\isacharparenright}{\kern0pt}\ imply{\isachardot}{\kern0pt}simps{\isacharparenleft}{\kern0pt}{\isadigit{1}}{\isacharparenright}{\kern0pt}\ imply{\isacharunderscore}{\kern0pt}conjunct{\isacharparenright}{\kern0pt}\isanewline%
}%
\SNIP{lemma:topoS4-NecR-6-0}{%
{\isasymturnstile}\isactrlsub T\isactrlsub o\isactrlsub p\ {\isacharparenleft}{\kern0pt}K\ i\ \isactrlbold {\isasymtop}\ \isactrlbold {\isasymlongrightarrow}\ K\ i\ p{\isacharparenright}{\kern0pt}%
}%
\SNIP{lemma:topoS4-NecR-6}{%
\ \ \isacommand{then}\isamarkupfalse%
\ \isacommand{have}\isamarkupfalse%
\ {\Cartouche{lemma:topoS4-NecR}{6}{0}}\isanewline%
}%
\SNIP{lemma:topoS4-NecR-7}{%
\ \ \ \ \isacommand{using}\isamarkupfalse%
\ RM\ \isacommand{by}\isamarkupfalse%
\ force\isanewline%
}%
\SNIP{lemma:topoS4-NecR-8}{%
\ \ \isacommand{then}\isamarkupfalse%
\ \isacommand{show}\isamarkupfalse%
\ {\isacharquery}{\kern0pt}thesis\isanewline%
}%
\SNIP{lemma:topoS4-NecR-9}{%
\ \ \ \ \isacommand{by}\isamarkupfalse%
\ {\isacharparenleft}{\kern0pt}meson\ AN\ System{\isacharunderscore}{\kern0pt}topoS{\isadigit{4}}{\isachardot}{\kern0pt}R{\isadigit{1}}{\isacharprime}{\kern0pt}{\isacharparenright}{\kern0pt}\isanewline%
}%
\SNIP{lemma:topoS4-NecR-10}{%
\isacommand{qed}\isamarkupfalse%
\endisatagproof
{\isafoldproof}%
\isadelimproof%
}%
\SNIP{lemma:S4-topoS4-0-0}{%
{\isasymturnstile}\isactrlsub S\isactrlsub {\isadigit{4}}\ p\ {\isasymLongrightarrow}\ {\isasymturnstile}\isactrlsub T\isactrlsub o\isactrlsub p\ p%
}%
\SNIP{lemma:S4-topoS4-0}{%
\isacommand{lemma}\isamarkupfalse%
\ S{\isadigit{4}}{\isacharunderscore}{\kern0pt}topoS{\isadigit{4}}{\isacharcolon}{\kern0pt}\ {\Cartouche{lemma:S4-topoS4}{0}{0}}\isanewline%
}%
\SNIP{lemma:S4-topoS4-1}{%
\isadelimproof
\endisadelimproof
\isatagproof
\isacommand{proof}\isamarkupfalse%
\ {\isacharparenleft}{\kern0pt}induct\ p\ rule{\isacharcolon}{\kern0pt}\ AK{\isachardot}{\kern0pt}induct{\isacharparenright}{\kern0pt}\isanewline%
}%
\SNIP{lemma:S4-topoS4-2}{%
\ \ \isacommand{case}\isamarkupfalse%
\ {\isacharparenleft}{\kern0pt}A{\isadigit{2}}\ i\ p\ q{\isacharparenright}{\kern0pt}\isanewline%
}%
\SNIP{lemma:S4-topoS4-3}{%
\ \ \isacommand{then}\isamarkupfalse%
\ \isacommand{show}\isamarkupfalse%
\ {\isacharquery}{\kern0pt}case\ \isacommand{using}\isamarkupfalse%
\ topoS{\isadigit{4}}{\isacharunderscore}{\kern0pt}AxK\ \isacommand{{\isachardot}{\kern0pt}}\isamarkupfalse%
\isanewline%
}%
\SNIP{lemma:S4-topoS4-4}{%
\isacommand{next}\isamarkupfalse%
\ \isanewline%
}%
\SNIP{lemma:S4-topoS4-5}{%
\ \ \isacommand{case}\isamarkupfalse%
\ {\isacharparenleft}{\kern0pt}Ax\ p{\isacharparenright}{\kern0pt}\isanewline%
}%
\SNIP{lemma:S4-topoS4-6}{%
\ \ \ \ \isacommand{then}\isamarkupfalse%
\ \isacommand{show}\isamarkupfalse%
\ {\isacharquery}{\kern0pt}case\ \isanewline%
}%
\SNIP{lemma:S4-topoS4-7}{%
\ \ \ \ \ \ \isacommand{using}\isamarkupfalse%
\ AT{\isacharprime}{\kern0pt}\ A{\isadigit{4}}{\isacharprime}{\kern0pt}\ \isacommand{by}\isamarkupfalse%
\ {\isacharparenleft}{\kern0pt}metis\ AxT{\isachardot}{\kern0pt}cases\ Ax{\isadigit{4}}{\isachardot}{\kern0pt}cases{\isacharparenright}{\kern0pt}\isanewline%
}%
\SNIP{lemma:S4-topoS4-8}{%
\isacommand{next}\isamarkupfalse%
\isanewline%
}%
\SNIP{lemma:S4-topoS4-9}{%
\ \ \isacommand{case}\isamarkupfalse%
\ {\isacharparenleft}{\kern0pt}R{\isadigit{2}}\ p{\isacharparenright}{\kern0pt}\isanewline%
}%
\SNIP{lemma:S4-topoS4-10}{%
\ \ \isacommand{then}\isamarkupfalse%
\ \isacommand{show}\isamarkupfalse%
\ {\isacharquery}{\kern0pt}case\ \isanewline%
}%
\SNIP{lemma:S4-topoS4-11}{%
\ \ \ \ \isacommand{by}\isamarkupfalse%
\ {\isacharparenleft}{\kern0pt}simp\ add{\isacharcolon}{\kern0pt}\ topoS{\isadigit{4}}{\isacharunderscore}{\kern0pt}NecR{\isacharparenright}{\kern0pt}\isanewline%
}%
\SNIP{lemma:S4-topoS4-12}{%
\isacommand{qed}\isamarkupfalse%
\ {\isacharparenleft}{\kern0pt}meson\ System{\isacharunderscore}{\kern0pt}topoS{\isadigit{4}}{\isachardot}{\kern0pt}intros{\isacharparenright}{\kern0pt}{\isacharplus}{\kern0pt}%
\endisatagproof
{\isafoldproof}%
\isadelimproof%
}%
\SNIP{lemma:topoS4-S4-0}{%
\isacommand{lemma}\isamarkupfalse%
\ topoS{\isadigit{4}}{\isacharunderscore}{\kern0pt}S{\isadigit{4}}{\isacharcolon}{\kern0pt}\isanewline%
}%
\SNIP{lemma:topoS4-S4-1-0}{%
{\isacharparenleft}{\kern0pt}{\isacharprime}{\kern0pt}i\ {\isacharcolon}{\kern0pt}{\isacharcolon}{\kern0pt}\ countable{\isacharparenright}{\kern0pt}\ fm%
}%
\SNIP{lemma:topoS4-S4-1}{%
\ \ \isakeyword{fixes}\ p\ {\isacharcolon}{\kern0pt}{\isacharcolon}{\kern0pt}\ {\Cartouche{lemma:topoS4-S4}{1}{0}}\isanewline%
}%
\SNIP{lemma:topoS4-S4-2-0}{%
{\isasymturnstile}\isactrlsub T\isactrlsub o\isactrlsub p\ p\ {\isasymLongrightarrow}\ {\isasymturnstile}\isactrlsub S\isactrlsub {\isadigit{4}}\ p%
}%
\SNIP{lemma:topoS4-S4-2}{%
\ \ \isakeyword{shows}\ {\Cartouche{lemma:topoS4-S4}{2}{0}}\isanewline%
}%
\SNIP{lemma:topoS4-S4-3}{%
\isadelimproof
\endisadelimproof
\isatagproof
\isacommand{proof}\isamarkupfalse%
\ {\isacharparenleft}{\kern0pt}induct\ p\ rule{\isacharcolon}{\kern0pt}\ System{\isacharunderscore}{\kern0pt}topoS{\isadigit{4}}{\isachardot}{\kern0pt}induct{\isacharparenright}{\kern0pt}\isanewline%
}%
\SNIP{lemma:topoS4-S4-4}{%
\ \ \isacommand{case}\isamarkupfalse%
\ {\isacharparenleft}{\kern0pt}AT{\isacharprime}{\kern0pt}\ i\ p{\isacharparenright}{\kern0pt}\isanewline%
}%
\SNIP{lemma:topoS4-S4-5}{%
\ \ \isacommand{then}\isamarkupfalse%
\ \isacommand{show}\isamarkupfalse%
\ {\isacharquery}{\kern0pt}case\isanewline%
}%
\SNIP{lemma:topoS4-S4-6}{%
\ \ \ \ \isacommand{by}\isamarkupfalse%
\ {\isacharparenleft}{\kern0pt}simp\ add{\isacharcolon}{\kern0pt}\ Ax\ AxT{\isachardot}{\kern0pt}intros{\isacharparenright}{\kern0pt}\isanewline%
}%
\SNIP{lemma:topoS4-S4-7}{%
\isacommand{next}\isamarkupfalse%
\isanewline%
}%
\SNIP{lemma:topoS4-S4-8}{%
\ \ \isacommand{case}\isamarkupfalse%
\ {\isacharparenleft}{\kern0pt}A{\isadigit{4}}{\isacharprime}{\kern0pt}\ i\ p{\isacharparenright}{\kern0pt}\isanewline%
}%
\SNIP{lemma:topoS4-S4-9}{%
\ \ \isacommand{then}\isamarkupfalse%
\ \isacommand{show}\isamarkupfalse%
\ {\isacharquery}{\kern0pt}case\isanewline%
}%
\SNIP{lemma:topoS4-S4-10}{%
\ \ \ \ \isacommand{by}\isamarkupfalse%
\ {\isacharparenleft}{\kern0pt}simp\ add{\isacharcolon}{\kern0pt}\ Ax\ Ax{\isadigit{4}}{\isachardot}{\kern0pt}intros{\isacharparenright}{\kern0pt}\isanewline%
}%
\SNIP{lemma:topoS4-S4-11}{%
\isacommand{next}\isamarkupfalse%
\isanewline%
}%
\SNIP{lemma:topoS4-S4-12}{%
\ \ \isacommand{case}\isamarkupfalse%
\ {\isacharparenleft}{\kern0pt}AR\ i\ p\ q{\isacharparenright}{\kern0pt}\isanewline%
}%
\SNIP{lemma:topoS4-S4-13}{%
\ \ \ \ \isacommand{then}\isamarkupfalse%
\ \isacommand{show}\isamarkupfalse%
\ {\isacharquery}{\kern0pt}case\isanewline%
}%
\SNIP{lemma:topoS4-S4-14}{%
\ \ \ \ \ \ \isacommand{by}\isamarkupfalse%
\ {\isacharparenleft}{\kern0pt}meson\ K{\isacharunderscore}{\kern0pt}conj{\isacharunderscore}{\kern0pt}imply{\isacharunderscore}{\kern0pt}factor\ K{\isacharunderscore}{\kern0pt}conjunction{\isacharunderscore}{\kern0pt}in\ K{\isacharunderscore}{\kern0pt}conjunction{\isacharunderscore}{\kern0pt}out\ K{\isacharunderscore}{\kern0pt}imp{\isacharunderscore}{\kern0pt}trans{\isacharprime}{\kern0pt}\ K{\isacharunderscore}{\kern0pt}imply{\isacharunderscore}{\kern0pt}multi\ R{\isadigit{1}}{\isacharparenright}{\kern0pt}\isanewline%
}%
\SNIP{lemma:topoS4-S4-15}{%
\isacommand{next}\isamarkupfalse%
\isanewline%
}%
\SNIP{lemma:topoS4-S4-16}{%
\ \ \isacommand{case}\isamarkupfalse%
\ {\isacharparenleft}{\kern0pt}AN\ i{\isacharparenright}{\kern0pt}\ \isanewline%
}%
\SNIP{lemma:topoS4-S4-17-0}{%
{\isasymturnstile}\isactrlsub S\isactrlsub {\isadigit{4}}\ \isactrlbold {\isasymtop}%
}%
\SNIP{lemma:topoS4-S4-17}{%
\ \ \isacommand{then}\isamarkupfalse%
\ \isacommand{have}\isamarkupfalse%
\ {\isacharasterisk}{\kern0pt}{\isacharcolon}{\kern0pt}\ {\Cartouche{lemma:topoS4-S4}{17}{0}}\isanewline%
}%
\SNIP{lemma:topoS4-S4-18}{%
\ \ \ \ \isacommand{by}\isamarkupfalse%
\ {\isacharparenleft}{\kern0pt}simp\ add{\isacharcolon}{\kern0pt}\ A{\isadigit{1}}{\isacharparenright}{\kern0pt}\isanewline%
}%
\SNIP{lemma:topoS4-S4-19}{%
\ \ \isacommand{then}\isamarkupfalse%
\ \isacommand{show}\isamarkupfalse%
\ {\isacharquery}{\kern0pt}case\ \isanewline%
}%
\SNIP{lemma:topoS4-S4-20}{%
\ \ \ \ \isacommand{by}\isamarkupfalse%
\ {\isacharparenleft}{\kern0pt}simp\ add{\isacharcolon}{\kern0pt}\ {\isacharasterisk}{\kern0pt}\ R{\isadigit{2}}{\isacharparenright}{\kern0pt}\isanewline%
}%
\SNIP{lemma:topoS4-S4-21}{%
\isacommand{next}\isamarkupfalse%
\isanewline%
}%
\SNIP{lemma:topoS4-S4-22}{%
\ \ \isacommand{case}\isamarkupfalse%
\ {\isacharparenleft}{\kern0pt}RM\ p\ q\ i{\isacharparenright}{\kern0pt}\isanewline%
}%
\SNIP{lemma:topoS4-S4-23-0}{%
{\isasymturnstile}\isactrlsub S\isactrlsub {\isadigit{4}}\ K\ i\ {\isacharparenleft}{\kern0pt}p\ \isactrlbold {\isasymlongrightarrow}\ q{\isacharparenright}{\kern0pt}%
}%
\SNIP{lemma:topoS4-S4-23}{%
\ \ \isacommand{then}\isamarkupfalse%
\ \isacommand{have}\isamarkupfalse%
\ {\Cartouche{lemma:topoS4-S4}{23}{0}}\isanewline%
}%
\SNIP{lemma:topoS4-S4-24}{%
\ \ \ \ \isacommand{by}\isamarkupfalse%
\ {\isacharparenleft}{\kern0pt}simp\ add{\isacharcolon}{\kern0pt}\ R{\isadigit{2}}{\isacharparenright}{\kern0pt}\isanewline%
}%
\SNIP{lemma:topoS4-S4-25}{%
\ \ \isacommand{then}\isamarkupfalse%
\ \isacommand{show}\isamarkupfalse%
\ {\isacharquery}{\kern0pt}case\isanewline%
}%
\SNIP{lemma:topoS4-S4-26}{%
\ \ \ \ \isacommand{by}\isamarkupfalse%
\ {\isacharparenleft}{\kern0pt}simp\ add{\isacharcolon}{\kern0pt}\ K{\isacharunderscore}{\kern0pt}map\ RM{\isachardot}{\kern0pt}hyps{\isacharparenleft}{\kern0pt}{\isadigit{2}}{\isacharparenright}{\kern0pt}{\isacharparenright}{\kern0pt}\isanewline%
}%
\SNIP{lemma:topoS4-S4-27}{%
\isacommand{qed}\isamarkupfalse%
\ {\isacharparenleft}{\kern0pt}meson\ AK{\isachardot}{\kern0pt}intros{\isacharparenright}{\kern0pt}{\isacharplus}{\kern0pt}%
\endisatagproof
{\isafoldproof}%
\isadelimproof%
}%
\SNIP{theorem:mainS4'-0-0}{%
valid\isactrlsub S\isactrlsub {\isadigit{4}}\ p\ {\isasymlongleftrightarrow}\ {\isacharparenleft}{\kern0pt}{\isasymturnstile}\isactrlsub T\isactrlsub o\isactrlsub p\ p{\isacharparenright}{\kern0pt}%
}%
\SNIP{theorem:mainS4'-0}{%
\isacommand{theorem}\isamarkupfalse%
\ main\isactrlsub S\isactrlsub {\isadigit{4}}{\isacharprime}{\kern0pt}{\isacharcolon}{\kern0pt}\ {\Cartouche{theorem:mainS4'}{0}{0}}\isanewline%
}%
\SNIP{theorem:mainS4'-1}{%
\isadelimproof
\ \ %
\endisadelimproof
\isatagproof
\isacommand{using}\isamarkupfalse%
\ main\isactrlsub S\isactrlsub {\isadigit{4}}\ S{\isadigit{4}}{\isacharunderscore}{\kern0pt}topoS{\isadigit{4}}\ topoS{\isadigit{4}}{\isacharunderscore}{\kern0pt}S{\isadigit{4}}\ \isacommand{by}\isamarkupfalse%
\ blast%
\endisatagproof
{\isafoldproof}%
\isadelimproof%
}%
\SNIP{datatype:kripke-0}{%
\isacommand{datatype}\isamarkupfalse%
\ {\isacharparenleft}{\kern0pt}{\isacharprime}{\kern0pt}i{\isacharcomma}{\kern0pt}\ {\isacharprime}{\kern0pt}w{\isacharparenright}{\kern0pt}\ kripke\ {\isacharequal}{\kern0pt}\isanewline%
}%
\SNIP{datatype:kripke-1-0}{%
{\isacharprime}{\kern0pt}w\ set%
}%
\SNIP{datatype:kripke-1-1}{%
{\isacharprime}{\kern0pt}w\ {\isasymRightarrow}\ id\ {\isasymRightarrow}\ bool%
}%
\SNIP{datatype:kripke-1-2}{%
{\isacharprime}{\kern0pt}i\ {\isasymRightarrow}\ {\isacharprime}{\kern0pt}w\ {\isasymRightarrow}\ {\isacharprime}{\kern0pt}w\ set%
}%
\SNIP{datatype:kripke-1}{%
\ \ Kripke\ {\isacharparenleft}{\kern0pt}{\isasymW}{\isacharcolon}{\kern0pt}\ {\Cartouche{datatype:kripke}{1}{0}}{\isacharparenright}{\kern0pt}\ {\isacharparenleft}{\kern0pt}{\isasympi}{\isacharcolon}{\kern0pt}\ {\Cartouche{datatype:kripke}{1}{1}}{\isacharparenright}{\kern0pt}\ {\isacharparenleft}{\kern0pt}{\isasymK}{\isacharcolon}{\kern0pt}\ {\Cartouche{datatype:kripke}{1}{2}}{\isacharparenright}{\kern0pt}%
}%
\SNIP{primrec:semantics-0-0}{%
{\isacharparenleft}{\kern0pt}{\isacharprime}{\kern0pt}i{\isacharcomma}{\kern0pt}\ {\isacharprime}{\kern0pt}w{\isacharparenright}{\kern0pt}\ kripke\ {\isasymRightarrow}\ {\isacharprime}{\kern0pt}w\ {\isasymRightarrow}\ {\isacharprime}{\kern0pt}i\ fm\ {\isasymRightarrow}\ bool%
}%
\SNIP{primrec:semantics-0}{%
\isacommand{primrec}\isamarkupfalse%
\ semantics\ {\isacharcolon}{\kern0pt}{\isacharcolon}{\kern0pt}\ {\Cartouche{primrec:semantics}{0}{0}}\isanewline%
}%
\SNIP{primrec:semantics-1}{%
\ \ {\isacharparenleft}{\kern0pt}{\isachardoublequoteopen}{\isacharunderscore}{\kern0pt}{\isacharcomma}{\kern0pt}\ {\isacharunderscore}{\kern0pt}\ {\isasymTurnstile}\ {\isacharunderscore}{\kern0pt}{\isachardoublequoteclose}\ {\isacharbrackleft}{\kern0pt}{\isadigit{5}}{\isadigit{0}}{\isacharcomma}{\kern0pt}\ {\isadigit{5}}{\isadigit{0}}{\isacharbrackright}{\kern0pt}\ {\isadigit{5}}{\isadigit{0}}{\isacharparenright}{\kern0pt}\ \isakeyword{where}\isanewline%
}%
\SNIP{primrec:semantics-2-0}{%
{\isacharparenleft}{\kern0pt}M{\isacharcomma}{\kern0pt}\ w\ {\isasymTurnstile}\ \isactrlbold {\isasymbottom}{\isacharparenright}{\kern0pt}\ {\isacharequal}{\kern0pt}\ False%
}%
\SNIP{primrec:semantics-2}{%
\ \ {\Cartouche{primrec:semantics}{2}{0}}\isanewline%
}%
\SNIP{primrec:semantics-3-0}{%
{\isacharparenleft}{\kern0pt}M{\isacharcomma}{\kern0pt}\ w\ {\isasymTurnstile}\ Pro\ x{\isacharparenright}{\kern0pt}\ {\isacharequal}{\kern0pt}\ {\isasympi}\ M\ w\ x%
}%
\SNIP{primrec:semantics-3}{%
{\isacharbar}{\kern0pt}\ {\Cartouche{primrec:semantics}{3}{0}}\isanewline%
}%
\SNIP{primrec:semantics-4-0}{%
{\isacharparenleft}{\kern0pt}M{\isacharcomma}{\kern0pt}\ w\ {\isasymTurnstile}\ {\isacharparenleft}{\kern0pt}p\ \isactrlbold {\isasymor}\ q{\isacharparenright}{\kern0pt}{\isacharparenright}{\kern0pt}\ {\isacharequal}{\kern0pt}\ {\isacharparenleft}{\kern0pt}{\isacharparenleft}{\kern0pt}M{\isacharcomma}{\kern0pt}\ w\ {\isasymTurnstile}\ p{\isacharparenright}{\kern0pt}\ {\isasymor}\ {\isacharparenleft}{\kern0pt}M{\isacharcomma}{\kern0pt}\ w\ {\isasymTurnstile}\ q{\isacharparenright}{\kern0pt}{\isacharparenright}{\kern0pt}%
}%
\SNIP{primrec:semantics-4}{%
{\isacharbar}{\kern0pt}\ {\Cartouche{primrec:semantics}{4}{0}}\isanewline%
}%
\SNIP{primrec:semantics-5-0}{%
{\isacharparenleft}{\kern0pt}M{\isacharcomma}{\kern0pt}\ w\ {\isasymTurnstile}\ {\isacharparenleft}{\kern0pt}p\ \isactrlbold {\isasymand}\ q{\isacharparenright}{\kern0pt}{\isacharparenright}{\kern0pt}\ {\isacharequal}{\kern0pt}\ {\isacharparenleft}{\kern0pt}{\isacharparenleft}{\kern0pt}M{\isacharcomma}{\kern0pt}\ w\ {\isasymTurnstile}\ p{\isacharparenright}{\kern0pt}\ {\isasymand}\ {\isacharparenleft}{\kern0pt}M{\isacharcomma}{\kern0pt}\ w\ {\isasymTurnstile}\ q{\isacharparenright}{\kern0pt}{\isacharparenright}{\kern0pt}%
}%
\SNIP{primrec:semantics-5}{%
{\isacharbar}{\kern0pt}\ {\Cartouche{primrec:semantics}{5}{0}}\isanewline%
}%
\SNIP{primrec:semantics-6-0}{%
{\isacharparenleft}{\kern0pt}M{\isacharcomma}{\kern0pt}\ w\ {\isasymTurnstile}\ {\isacharparenleft}{\kern0pt}p\ \isactrlbold {\isasymlongrightarrow}\ q{\isacharparenright}{\kern0pt}{\isacharparenright}{\kern0pt}\ {\isacharequal}{\kern0pt}\ {\isacharparenleft}{\kern0pt}{\isacharparenleft}{\kern0pt}M{\isacharcomma}{\kern0pt}\ w\ {\isasymTurnstile}\ p{\isacharparenright}{\kern0pt}\ {\isasymlongrightarrow}\ {\isacharparenleft}{\kern0pt}M{\isacharcomma}{\kern0pt}\ w\ {\isasymTurnstile}\ q{\isacharparenright}{\kern0pt}{\isacharparenright}{\kern0pt}%
}%
\SNIP{primrec:semantics-6}{%
{\isacharbar}{\kern0pt}\ {\Cartouche{primrec:semantics}{6}{0}}\isanewline%
}%
\SNIP{primrec:semantics-7-0}{%
{\isacharparenleft}{\kern0pt}M{\isacharcomma}{\kern0pt}\ w\ {\isasymTurnstile}\ K\ i\ p{\isacharparenright}{\kern0pt}\ {\isacharequal}{\kern0pt}\ {\isacharparenleft}{\kern0pt}{\isasymforall}v\ {\isasymin}\ {\isasymW}\ M\ {\isasyminter}\ {\isasymK}\ M\ i\ w{\isachardot}{\kern0pt}\ M{\isacharcomma}{\kern0pt}\ v\ {\isasymTurnstile}\ p{\isacharparenright}{\kern0pt}%
}%
\SNIP{primrec:semantics-7}{%
{\isacharbar}{\kern0pt}\ {\Cartouche{primrec:semantics}{7}{0}}%
\isadelimdocument
\endisadelimdocument
\isatagdocument%
}%